\documentclass[12pt]{article} 
\pdfoutput=1
\usepackage[authoryear]{natbib}
\usepackage{graphicx,epstopdf}
\usepackage{setspace}
\usepackage[margin=2.5cm]{geometry}
\usepackage{array}
\usepackage{hyperref}
\usepackage{pdflscape}
\usepackage{lscape}
\usepackage{graphicx}
\usepackage[utf8]{inputenc}
\usepackage[english]{babel}
\usepackage{amsmath, amsthm, amssymb, amsfonts}
\usepackage{mathrsfs}
\usepackage{dsfont}
\usepackage{xcolor}
\usepackage{physics}
\usepackage{mathtools}
\usepackage{hyperref}
\usepackage{booktabs}
\usepackage[normalem]{ulem}

\newtheorem{theorem}{Theorem}[section]
\newtheorem{corollary}{Corollary}[theorem]
\newtheorem{lemma}[theorem]{Lemma}
\theoremstyle{definition}

\newtheorem{assumption}{Assumption}[section]

\newtheorem{remark}{Remark}[section]

\usepackage{thmtools}
\declaretheorem[style=definition]{example}

\usepackage{float}
\floatstyle{ruled}
\newfloat{algorithm}{tbp}{loa}
\providecommand{\algorithmname}{Algorithm}
\floatname{algorithm}{\protect\algorithmname}

\usepackage{ifpdf}
\ifpdf
  \DeclareGraphicsExtensions{.pdf,.png,.jpg}
\else
  \DeclareGraphicsExtensions{.eps}
\fi

\usepackage{natbib}
\setcitestyle{authoryear,open={(},close={)}}

\usepackage{enumitem}
\usepackage{subfig}
\usepackage{placeins}
\usepackage{framed}

\begin{document}

\title{Moment Restrictions for Nonlinear \\
Panel Data Models with Feedback}
\author{Stéphane Bonhomme, Kevin Dano, and Bryan S. Graham\thanks{\underline{Bonhomme}: Department of Economics, University of Chicago, 1126 E. 59th Street, Chicago, IL 60637, e-mail: \url{sbonhomme@uchicago.edu}, web: \href{https://sites.google.com/site/stephanebonhommeresearch/}{https://sites.google.com/site/stephanebonhommeresearch/}. \newline \underline{Dano}: Department of Economics, Princeton University, Washington Rd, Princeton, NJ 08544, e-mail: \url{kdano@princeton.edu}, web:  \href{https://kevindano.github.io}{https://kevindano.github.io}. \newline  \underline{Graham}: Department of Economics, University of California - Berkeley, 530 Evans Hall \#3380, Berkeley, CA 94720-3880 and National Bureau of Economic Research, e-mail: \url{bgraham@econ.berkeley.edu}, web:  \href{http://bryangraham.github.io/econometrics/}{http://bryangraham.github.io/econometrics/}. \newline We thank the Co-Editor and three anonymous referees, Manuel Arellano, Jinyong Hahn, Bo Honor\'e, Whitney Newey, Aureo de Paula, Martin Weidner, Tiemen Woutersen, and seminar participants at various places for comments. The usual disclaimers apply. Portions of the research reported were undertaken while Bonhomme and Graham were visiting CEMFI in the Fall of 2022 with support from the Spanish State Research Agency under the María de Maeztu Unit of Excellence Programme (Project No: CEX2020-001104-M).}} 
	
\date{\today 	}
	
\maketitle
	\thispagestyle{empty}
	
\begin{abstract}
    \footnotesize{Many panel data methods, while allowing for general dependence between covariates and time-invariant agent-specific heterogeneity, place strong \emph{a priori} restrictions on \emph{feedback}: how past outcomes, covariates, and heterogeneity map into future covariate levels. Ruling out feedback entirely, as often occurs in practice, is unattractive in many dynamic economic settings. We provide a general characterization of all \emph{feedback and heterogeneity robust} (FHR) moment conditions for nonlinear panel data models and present constructive methods to derive feasible moment-based estimators for specific models. We characterize semiparametric efficiency bounds in this case, quantifying the information loss associated with accommodating feedback as well as providing insight into how to construct estimators with good efficiency properties in practice. When FHR moments point-identify the target parameter -- the case we emphasize -- we show they attain the efficiency bound associated with the full model. These results apply both to the finite dimensional parameter indexing the parametric part of the model as well as to estimands that involve averages over the distribution of unobserved heterogeneity or features of the feedback process. Lastly, we illustrate our methods by providing a complete characterization of all FHR moment functions in the multi-spell mixed proportional hazards model, and compute efficient moment functions for both model parameters and average effects in this setting. 
    }

    \end{abstract}
	
\textbf{{JEL Codes:}} C23, C33
	
\textbf{{Keywords:}} {Sequential Exogeneity, Feedback, Panel Data, Duration Models, Incidental Parameters, Semiparametric Efficiency Bounds.}
	
\newpage
	
\pagenumbering{arabic}
\onehalfspacing

\renewcommand\thmcontinues[1]{Continued}

\newpage
	
\pagenumbering{arabic}
\onehalfspacing

An econometrician randomly samples units from a population of interest. For each sampled unit $i=1,\ldots,N$, let $Y_{it}$ denote a period $t=1,\ldots,T$ outcome, $X_{it}$ a corresponding vector of covariates, and $A_{i}$, a latent variable representing unmeasured unit-specific attributes. Importantly, $A_i$ is constant over time and may freely covary with the regressors, $X_{i1},\ldots,X_{iT}$. An initial condition, $Y_{i0}$, is also observed. Panel data of this type feature prominently in empirical research in economics and other fields. 
That panel data offers the possibility to ``control for" the correlated heterogeneity, $A_{i}$, is a key attraction. 

While ``fixed effects” panel data methods place no restrictions (beyond mild regularity conditions) on the joint distribution of the initial condition $Y_{i0}$ and latent heterogeneity $A_i$, they generally \emph{do} place strong restrictions on how the outcomes $Y_{i1},\ldots,Y_{iT}$ and regressors $X_{i1},\ldots,X_{iT}$ relate to each other. These restrictions involve more than substantive modeling assumptions; they also constrain what we will call the \emph{feedback process}, whereby past outcomes and covariates $Y_{it-1},Y_{it-2},\ldots,Y_{i0},X_{it-1},X_{it-2},\ldots,X_{i1}$, as well as heterogeneity $A_i$, influence current covariates $X_{it}$.

Feedback arises naturally in many dynamic economic problems. For example, a firm's optimal investment rule generally varies with its current capital stock (and hence past investment decisions) as well as past productivity shocks (and hence its output history), see, e.g., \citet{Olley_Pakes_EM96} and \citet{Blundell_Bond_ER00}. A doctor may adjust a patient's treatment protocol in a way which depends on her perceptions of their health response to past treatments \citep[e.g.,][]{robins1986new}. A worker's decision to participate in job training may, as is typically the focus in evaluation studies, influence their future labor market outcomes, but participation may also depend on their past labor market experiences \citep[e.g.,][]{Ashenfelter_RESTAT1978,Ashenfelter_Card_RESTAT1985}. 

One approach to handling feedback, indeed the leading one in empirical work, rules it out \emph{a priori}. This corresponds to maintaining the \emph{strict exogeneity} assumption formulated by \cite{Chamberlain_EM1982,Chamberlain_JE82}. Strict exogeneity assumptions underpin, albeit generally implicitly, many panel data based approaches to program evaluation (see \citealp{ghanem2022selection} on difference-in-differences methods). The overwhelming majority of nonlinear panel data estimators also require strict exogeneity \citep[see][]{arellano2001panel,arellano2011nonlinear}. Strict exogeneity, while a convenient assumption for estimation, is restrictive in many economic applications. Ironically, although \cite{Chamberlain_JE82,Chamberlain_HBE84} emphasized the testable implications of strict exogeneity, today the assumption is so common as to often go unmentioned in applied work.

A different approach, pioneered by \citet{robins1986new}, accommodates feedback, but assumes it is \emph{homogeneous}. By homogeneous we mean that the mapping from past outcomes, $Y_{it-1},Y_{it-2},\ldots,Y_{i0}$ and covariates $X_{it-1},X_{it-2},\ldots,X_{i1}$ to the current covariate, $X_{it}$, \emph{does not} vary with $A_i$: it is identical across agents. This is a powerful simplification, leading to feasible nonparametric and semiparametric estimators \citep[e.g.,][]{robins2000marginal}. However, the restriction to homogeneous feedback, like strict exogeneity, is a strong assumption. It rules out, for example, a firm's investment rule varying with its persistent productivity level. Robins' \citeyearpar{robins1986new} setup is often plausible in environments where the researcher controls $X_{it}$, such as in a dynamic experiment. Strict exogeneity and homogeneous feedback are non-nested assumptions; but both restrictions correspond to a subset of the data generating processes covered by our results.

In this paper we study nonlinear panel data models with \emph{unrestricted} heterogeneous feedback and correlated heterogeneity. Twenty-five years ago, surveying the then extant work on nonlinear panel data analysis, \citet[][p. 3265]{arellano2001panel} observed:
\begin{quote}
    The main limitation of much of the literature on nonlinear panel data methods is that it is assumed that the explanatory variables are strictly exogenous in the sense that some assumptions will be made on the errors conditional on all (including future) values of the explanatory variables. 
\end{quote}
This observation remains largely true today. \cite{Chamberlain_JOE2022}, in a paper first circulated in the early 1990s, studied a class of panel data models with multiplicative heterogeneity defined by sequential moment restrictions. Certain panel data count models are covered by his results \citep[see][]{Chamberlain_JBES92,wooldridge1997multiplicative, Blundell_Griffith_Windmeijer_JOE2002,Windmeijer_EPD2008}. Feedback in dynamic \emph{linear} panel data models with sequential moment restrictions is also well understood \citep[e.g.,][]{arellano1991some,arellano1995another,Chamberlain_JBES92,Hahn_ET1997,ai2012semiparametric}. However, outside the setting studied by \cite{Chamberlain_JOE2022} that includes linear models as a special case, very little is known about panel data models with unrestricted feedback.\footnote{\citet{Buchinsky_et_al_EL2010} show how to compute semiparametric efficiency bounds in dynamic discrete choice models that feature feedback. Recently, in independent work, \citet{botosaru2024adversarialapproachidentification} and \citet{chesher2024robust} propose innovative approaches to static and dynamic nonlinear panel data models.}

In this paper, we characterize the set of feedback and heterogeneity robust (FHR) moment conditions in nonlinear panel data models with feedback. We work in a semiparametric likelihood setting where the outcome density is indexed by a finite-dimensional parameter, while both the feedback process and the unobserved heterogeneity distribution are left unrestricted. Our results hold for fixed $T$, and cover both the finite-dimensional parameter indexing the parametric part of the model (the ``common" parameter), as well as estimands which involve averages over the distribution of unobserved heterogeneity and feedback (e.g., average partial effects, average treatment effects and other average effects).

Our exposition emphasizes settings where the set of FHR moments is
informative and point-identification holds. We characterize semiparametric efficiency bounds (SEBs) for the common parameter and average effects in such cases. We demonstrate how these results may be used to find feasible  estimating equations with good efficiency properties in practice. To illustrate the power of our approach, we include a complete characterization of all FHR moments in the multi-spell mixed proportional hazards (MPH) model with both feedback and lagged duration dependence. \cite{heckman1980does} emphasized the importance of incorporating these phenomena into duration analysis, although we are not aware of methods for doing so beyond those appearing in Section 1.4 of the unpublished dissertation of \citet{woutersen2000essays}. \cite{chamberlain1985heterogeneity} derived a marginal likelihood estimator for the multi-spell MPH under strict exogeneity, no lagged duration dependence, and a Weibull baseline hazard. \cite{hahn1994efficiency} studied SEBs in this same model, demonstrating a non-zero bound with multi-spell data. Our results show that information for model parameters remains positive, albeit attenuated, under unrestricted feedback (see also \citealp{Ridder_Woutersen_EM2003} for related work on moment-based estimation of MPH models).

Importantly, for some combinations of models and estimands, the set of informative FHR moments may be empty. The methods described below are not useful in such settings, even if the target parameter's identified set is, in fact, bounded. The binary choice logit model with feedback exemplifies this case \citep{bonhomme2023identification}.\footnote{As is also true under strict exogeneity \citep{chamberlain2010binary}, point-identification of the common parameter in nonlinear panel data models is generally challenging when the outcome is discretely-valued.} Even if informative FHR moments exist, they may not point-identify the parameter of interest. This reflects that the identified set based on FHR moment restrictions and the identified set based on all of the model's assumptions may not coincide. We provide a precise characterization of the difference between these two sets in Supplemental Appendix \ref{app: FHR_outerset}.

In the next section we formally define the class of semiparametric nonlinear panel data models with feedback. Section \ref{sec: moments} presents our first main result: a complete characterization of the set of all feedback and heterogeneity robust (FHR) moment conditions. This extends the characterization obtained by functional differencing (\citealp{Bonhomme_EM12}), which requires strict exogeneity, to models with unrestricted feedback. 

Section \ref{sec: moments_average} provides a FHR moment characterization for average effects. Section \ref{sec: SEBs} presents our semiparametric efficiency bound (SEB) analysis. There we demonstrate that the orthogonal complement of the nuisance tangent set coincides with the set of all FHR estimating equations. This coincidence has important implications for efficient estimation. Specifically, we show how to construct FHR moment functions that have good efficiency properties, leading to locally efficient estimators in the sense of \cite{newey1990semiparametric}. A key implication of the analysis presented in Section \ref{sec: SEBs} is that, when the FHR moments suffice to point-identify the target parameter, they also contain all relevant information about it: the bound calculated when only the FHR moments are the prior restriction coincides with the one calculated using all model implications.

Throughout we use the MPH model to illustrate key results in a concrete setting. Our MPH results are novel and of independent interest. We formally establish point-identification of the common parameter in the MPH model with lagged duration dependence, feedback, and a Weibull baseline hazard. We also introduce and study a new estimand, the \emph{average structural hazard} (ASH) function, and show how functionals of the (unrestricted) feedback process may be identified as well. Our SEB calculation for the MPH model, when compared with the earlier work of \cite{hahn1994efficiency} for the strictly exogenous case, provides a precise measure of the asymptotic precision costs associated with accommodating feedback.

 Closing our paper, in Section \ref{sec: twilight_zone_material} we describe how to construct FHR moments in various examples: a Poisson model; a new generalization of the MPH model that includes a type of interactive heterogeneity (and hence a relaxation of the ``proportional hazards" restriction); and a nonlinear regression model where we highlight the possibility to approximate the FHR property through regularization. A version of the paper that combines the main text, appendices, and all supplementary materials in a single file is available [\href{https://kevindano.github.io/assets/files/BDG_feedback.pdf}{here}].

\paragraph{Notation.}

In what follows we generally suppress the $i$ subscript when referring to a single random draw from the cross-sectional population. Hence, for example, $Y_t$ denotes the period $t$ outcome of a randomly sampled unit and $A$ its unobserved, time-invariant, attribute. We let $Z^{t}=\left(Z_{t},Z_{t-1},\ldots\right)$ denote the entire observed history of $Z_{t}$ and $Z^{s:t}=\left(Z_{s},\ldots,Z_{t}\right)$ its history from periods $s\leq t$ to $t$. To economize on notation, we omit almost surely statements unless there is a risk of confusion. Likewise, abusing notation we simply write ``$f(x)=0$'' instead of ``$f(x)=0$ for almost all possible $x$ values''.

\section{Panel data models with feedback} \label{sec: model}

In this section we introduce the semiparametric panel data model with feedback, connect this model to the more restrictive one which maintains strict exogeneity, and formally state our main research questions. Throughout, we illustrate key results in the context of a multi-spell mixed proportional hazards (MPH) model with lagged duration dependence and feedback.

\subsection{Setup}
Let $\left\{ \left(X_{i1},\ldots,X_{iT},Y_{i0},Y_{i1},\ldots,Y_{iT},A_{i}\right)\right\} _{i=1}^{\infty}$ be an independently and identically distributed random sequence drawn from some distribution function $F$. The sole prior restriction on $F$ is that the conditional density
of $Y_{t}$ at $y_{t}$ given the past $y^{t-1}$, regressor history
$x^{t}$, and latent unit-specific heterogeneity $a\in{\cal{A}}\subset \mathbb{R}^M$, belongs to
a known parametric family indexed by the unknown parameter $\theta\in\Theta\subset\mathbb{R}^{K}:$
\begin{equation}
   f\left(\left.y_{t}\right|y^{t-1},x^{t},a\right)=f_{\theta}\left(\left.y_{t}\right|y^{t-1},x^{t},a\right)=f_{\theta}\left(\left.y_{t}\right|y_{t-1},x_{t},a\right),\ t=1,\ldots,T,\label{eq: parametric_part_of_likelihood_feedback}
\end{equation}
for some $\theta\in\Theta$. The density $f_{\theta}\left(\left.y_{t}\right|y_{t-1},x_{t},a\right)$ is the parametric component of our setup.\footnote{While (\ref{eq: parametric_part_of_likelihood_feedback}) imposes that only the contemporaneous regressor and the first lag of the outcome matter,  additional lags could be easily accommodated. All theoretical characterizations developed in this paper continue to hold after replacing $f_{\theta}\left(\left.y_{t}\right|y_{t-1},x_{t},a\right)$ by the general density $f_{\theta}\left(\left.y_{t}\right|y^{t-1},x^{t},a\right)$, and augmenting the initial condition $(y_0,x_1)$ as needed for the specified lag order. Also, the parametric family need not be constant over time, allowing for time effects and other forms of nonstationarity.}

Nonlinear examples of \eqref{eq: parametric_part_of_likelihood_feedback} include binary choice logit models \citep{Chamberlain_ReStud80,chamberlain2010binary,bonhomme2023identification, honore2024moment} and count models \citep{Chamberlain_JBES92,Chamberlain_JOE2022, wooldridge1997multiplicative}. Specific forms for $f_{\theta}\left(\left.y_{t}\right|y_{t-1},x_{t},a\right)$ also arise in the context of dynamic structural models \citep[e.g.,][]{Aguirregabiria_Mira_JOE2010}.

\begin{example} {(\textsc{Mixed Proportional Hazards (MPH) Model})}\label{ex: mph_intro} Let $\left\{Y_{t}\right\}_{t=1}^{T}$ denote a sequence of  durations, or spell lengths, with $X_{t}$ a corresponding vector of beginning-of-spell covariates; $Y_{0}$ is an ``initial duration''. For example, $Y_{t}$ might equal time to re-arrest following an agent's $t^{th}$ release from prison while $X_{t}$ might include measures of their post-release support and supervision. Since support and supervision, $X_{t}$, might depend on the previous time to re-arrest, $Y_{t-1}$, as well as unmeasured agent attributes, $A$, heterogeneous feedback is plausible.

The instantaneous conditional hazard rate is given by
\begin{equation}\label{eq: MPH_instantaneous_hazard}
    \lambda\left(\left.y_{t}\right|Y^{t-1}=y^{t-1},X^{t}=x^{t},A=a\right)=\lambda_{\alpha}\left(y_{t}\right)\exp\left(\gamma y_{t-1}+x_{t}'\beta+a\right),
\end{equation}
with $\lambda_{\alpha}\left(y_{t}\right)$ a known parametric family of baseline hazard functions indexed by $\alpha$ (we emphasize the Weibull case with $\lambda_{\alpha}\left(y_{t}\right)=\alpha y_{t}^{\alpha-1}$, $\alpha>0$, below). Under \eqref{eq: MPH_instantaneous_hazard} the conditional density at $Y_{t}=y_{t}$ is
\begin{equation}
    f_{\theta}\left(\left.y_{t}\right|y^{t-1},x^{t},a\right)=\lambda_{\alpha}\left(y_{t}\right)\exp\left(\gamma y_{t-1}+x_{t}'\beta+a\right)\exp\left(-\rho_{\theta}\left(z_{t}\right)e^a\right),\label{eq_dens_MPH} 
\end{equation}
for $z_t=(y_t,y_{t-1},x_t')'$, $\rho_{\theta}\left(z_{t}\right)=\Lambda_{\alpha}\left(y_{t}\right)\exp\left(\gamma y_{t-1}+x_{t}'\beta\right)$, $\theta=\left(\alpha',\beta',\gamma\right)'$, and $\Lambda_{\alpha}\left(y_{t}\right)=\int_{0}^{y_{t}}\lambda_{\alpha}\left(u\right)\mathrm{d}u$ the integrated baseline hazard.
\end{example}

Returning to our general setup, sequentially factorizing the joint density of $Y_{0},Y_{1},\ldots,Y_{T},X_{1},\ldots,X_{T},A$
at $y_{0},y_{1},\ldots,y_{T},x_{1},\ldots,x_{T},a$ yields the following expression for the likelihood contribution of a single unit:
\begin{align}
    \ell\left(\left.\theta,g,\pi,\nu\right|y^{T},x^{T}\right)= & {\int} \left[\prod_{t=2}^{T}f\left(\left.x_{t},y_{t}\right|x^{t-1},y^{t-1},a\right)\right]f\left(\left.y_1\right|y_{0},x_{1},a\right)f\left(\left.a\right|y_{0},x_{1}\right)f\left(y_{0},x_{1}\right) {\mathrm{d}a} \nonumber \\
    = & {\int}\left[\prod_{t=1}^{T}\underset{\text{Parametric component}}{\underbrace{f_{\theta}\left(\left.y_{t}\right|y_{t-1},x_{t},a\right)}}\right]\times\left[\prod_{t=2}^{T}\underset{\text{Feedback process}}{\underbrace{g\left(\left.x_{t}\right|y^{t-1},x^{t-1},a\right)}}\right]\nonumber \\
    & \quad \times\underset{\text{Heterogeneity}}{\underbrace{\pi\left(\left.a\right|y_{0},x_{1}\right)}}\times\underset{\text{Initial condition}}{\underbrace{\nu\left(y_{0},x_{1}\right)}}{\mathrm{d}a},\label{eq: feedback_complete_data_likelihood}
\end{align}
where the second equality follows by imposing the parametric assumption
\eqref{eq: parametric_part_of_likelihood_feedback} and establishing
the following notations: 
\begin{enumerate}
  \item[(i)] $g\left(\left.x_{t}\right|y^{t-1},x^{t-1},a\right)$, denotes the density of $X_{t}$ at $x_{t}$ given the past $\left(y^{t-1},x^{t-1}\right)$ and heterogeneity $a$;
  \item[(ii)] $\pi\left(\left.a\right|y_{0},x_{1}\right)$, the conditional
    density of $A$ at $a$ given the initial condition $\left(y_{0},x_{1}\right)$; and
  \item[(iii)] $\nu\left(y_{0},x_{1}\right)$, the initial condition density.
\end{enumerate}
While $f_{\theta}\left(\left.y_{t}\right|y^{t-1},x^{t},a\right)$ belongs to
a parametric family, the remaining components, items (i) to
(iii) above, are all unrestricted. We call
this model the \emph{semiparametric panel data model with feedback}.

In what follows we call the $T-1$ densities $g\left(\left.x_{T}\right|y^{T-1},x^{T-1},a\right)$, $g\left(\left.x_{T-1}\right|y^{T-2},x^{T-2},a\right)$, ..., $g\left(\left.x_{2}\right|y^{1},x_{1},a\right)$
the \emph{feedback process.} This process describes how past values of the outcome influence current regressor values. Because it depends on $A$, the feedback process is \emph{heterogeneous }across units. We do not impose any stationarity over $t=1,\ldots,T$, so $g\left(\left.x_{t}\right|y^{t-1},x^{t-1},a\right)$ and $g\left(\left.x_{s}\right|y^{s-1},x^{s-1},a\right)$ for $s \neq t$ may differ arbitrarily. When $X_t$ is a policy variable, such flexibility accommodates un-modeled regime shifts (e.g., as when a change in tax policy alters firm investment behavior). We use the notation $g$ and occasionally $(g_{t})_{t=2}^T$ to denote a generic element of the set of all allowable feedback processes $\mathcal{G}$. We call $\pi\left(\left.a\right|y_{0},x_{1}\right)$
the \emph{heterogeneity distribution}; this density describes the distribution of unobserved heterogeneity, $A$, across units as well as any dependence of this heterogeneity on the {initial condition}, $\left(y_{0},x_{1}\right)$. Let $\pi$ denote a generic element of the set of all allowable heterogeneity distributions, ${\cal{P}}$. Finally we let $\nu\in\mathcal{N}$ denote an element of the set of all possible initial condition densities.

Although not emphasized in our exposition, it is straightforward to incorporate strictly exogenous regressors into the feedback model. Similarly, additional sources of heterogeneity, beyond $A$, can enter the feedback process. This generality, while important in some applications, involves no new issues and clutters notation.\footnote{To be specific: let $W^{T}$ contain all leads and lags of a vector of strictly exogeneous regressors and $B$ an additional source of heterogeneity. The results that follow are easily modified to accommodate the richer model:
\begin{align*}
    \ell\left(\left.\theta,g,\pi,\nu\right|y^{T},x^{T},w^{T}\right)=&\Biggl\{\int\int\prod_{t=1}^{T}f_{\theta}\left(\left.y_{t}\right|y_{t-1},x_{t},w_{t},a\right)\times\left[\prod_{t=2}^{T}g\left(\left.x_{t}\right|y^{t-1},x^{t-1},w^{T},a,b\right)\right]\\&\times\pi\left(\left.a,b\right|y_{0},x_{1},w^{T}\right)\mathrm{d}a\mathrm{d}b\Biggr\}\times\nu\left(y_{0},x_{1},w^{T}\right),
\end{align*}
where we assume that only the contemporaneous value of $W_{t}$ enters the parametric part of the model for simplicity.} We also note that, although not emphasized in most of our examples, $Y_t$ may be vector-valued in some settings.

\begin{remark}{\textsc{(Strict exogeneity)}} In applications, researchers routinely maintain strict exogeneity of the regressors, $X_t$, conditional on the latent attribute, $A$. Strict exogeneity imposes independence of $Y_t$ and $\left(X_{t+1},\ldots,X_{T}\right)$ conditional on $\left(Y_{0},X_{1},\ldots,X_{t},A\right)$. \cite{Chamberlain_EM1982} shows that strict exogeneity is equivalent to the following \emph{no feedback} condition:
\begin{equation}
    g\left(\left.x_{t}\right|y^{t-1},x^{t-1},a\right)=g\left(\left.x_{t}\right|y_{0},x^{t-1},a\right),\ t=2,\ldots,T.\label{eq: no_feedback_restriction}
\end{equation}
While (\ref{eq: feedback_complete_data_likelihood}) allows
$Y_{s}$ to influence $X_{t}$ for $s<t$, \eqref{eq: no_feedback_restriction} rules out such feedback (we allow dependence on the initial condition throughout). Restriction \eqref{eq: no_feedback_restriction} is a counterpart of Granger's \citeyearpar{Granger_EM1969} definition of ``$Y$ does not cause $X$" in the panel data setting with latent heterogeneity. If $X_{t}$ is a choice variable, and $Y^{t-1}$ is a component of the agent's beginning-of-period $t$ information set, then \eqref{eq: no_feedback_restriction} typically implies strong
restrictions on economic behavior \citep[e.g.,][]{Ashenfelter_Card_RESTAT1985} and/or the structure of agent's information sets \citep[e.g.,][]{Chamberlain_HBE84, chamberlain1985heterogeneity}. Our approach, by accommodating unrestricted feedback and heterogeneity, allows the researcher to proceed without maintaining such assumptions.
\end{remark}

\subsection{The goal: feedback and heterogeneity robust estimation}

In this paper we study estimation of the common parameter, $\theta$, in panel data models when \eqref{eq: parametric_part_of_likelihood_feedback}
is the \emph{only} prior restriction on $F$ (except for some mild regularity conditions). We wish 
 to construct estimators that are consistent irrespective of the precise instances of the feedback process, $g$, heterogeneity distribution, $\pi$, and initial condition, $\nu$. In what follows we call such estimators
\emph{feedback and heterogeneity robust} (FHR). Because strict exogeneity obtains as a special case, any FHR estimator remains valid under strict exogeneity. FHR estimators are natural extensions of familiar ``fixed effects” approaches to estimation in panel data models without feedback. We fully characterize the set of FHR moment functions. We also derive semiparametric efficiency bounds for $\theta$. 


One alternative to FHR estimation involves assuming that both the feedback process, $g$, and heterogeneity density, $\pi$, belong to parametric families indexed by some parameter vector $\eta$. With these additional maintained assumptions, the econometrician can then maximize the resulting likelihood, conditional on $Y_{0}$ and $X_{1}$, with respect to both $\theta$ and $\eta$. This transforms the problem from a semiparametric to a parametric one. Consistency of such parametric ``random effects'' estimators typically requires that the additionally maintained parametric restrictions on $g$ and $\pi$ hold in the sampled population. Consequently, such estimators are \emph{not} generally feedback and heterogeneity robust.

We next formally define the FHR property. Let $\theta\in\Theta$, and  $\omega=\left(g,\pi,\nu\right)\in\Omega$ collect all the nuisance parameters in our model. We assume that all elements $\omega\in\Omega$ have a common support, known to the econometrician (which may be unbounded and include the full real line, for example, for the support of the heterogeneity $A$). Let $\mathbb{E}_{\theta,\omega}\left[\cdot\right]$ denote an expectation taken under the DGP at $\left(\theta,\omega\right)$. Let $\phi_{\theta}\left(y_{0},y_{1},\ldots,y_{T},x_{1},\ldots,x_{T}\right)$ be a function of the observed data indexed by $\theta$. We say that $\phi_{\theta}$ is a \emph{FHR moment function} if, \emph{for all} $\omega$ such that $\phi_{\theta}$ is absolutely integrable under DGP $\left(\theta,\omega\right)$, we have
\begin{equation}           
\mathbb{E}_{\theta,\omega}\left[\phi_{\theta}\left(Y_{0},Y_{1},\ldots,Y_{T},X_{1},\ldots,X_{T}\right)\right]=0.\label{eq: valid_FHR_moment}
\end{equation}

The main goal of the paper is to derive moment restrictions $\phi_{\theta}$ that have the FHR property. We will first provide a complete characterization of FHR moment restrictions for $\theta$ in Section \ref{sec: moments}.\footnote{In fact, our characterization of FHR moment functions remains valid if $\theta$ is infinite-dimensional (for example, if $\theta$ contains nonparametric elements such as functions). However, the finite-dimensional $\theta$ case contains many important models, and all the examples we will use as illustrations feature a finite-dimensional $\theta$ vector. We also focus on this case in our analysis of efficiency.} Then, in Section \ref{sec: moments_average} we will show how essentially the same analysis can be used to provide a characterization of FHR moment functions for average effects of the form $$\mu(\theta,\omega)=\mathbb{E}_{\theta,\omega}\left[h_{\theta}\left(Y^{T},X^{T},A\right)\right],$$ where $h_{\theta}\left(Y^{T},X^{T},A\right)$ is a known function of $(Y^{T},X^{T},A)$ indexed by $\theta$. Average effects include a variety of causal or structural parameters, such as average partial
effects, as in \citet{Chamberlain_HBE84}, and average structural
functions, as in \citet{Blundel_Powell_WC03}. Unique to our setting, $\mu(\theta,\omega)$ may also capture summary features of the feedback process. See Section \ref{sec: moments_average} below for an example.

\section{FHR moment restrictions for common parameters}\label{sec: moments}

This section presents our first main result: a characterization of the \emph{entire set} of FHR moment conditions for $\theta$. After stating this result, we illustrate it through various examples: the mixed proportional hazard (MPH) model, where we derive all available FHR moments, and binary choice and random coefficients examples, where we show \emph{no} informative FHR moments exist.

\subsection{Main characterization of FHR moments for $\theta$}

We begin by characterizing the set of FHR moments for $\theta$.

\begin{theorem}{\textsc{(FHR Moment Characterization)}.}\label{theo_charact} Let $\theta\in\Theta$. (A) Suppose that the following $T$ restrictions hold: (i)
		\begin{equation}\int \phi_{\theta}(y^T,x^T){{\prod_{t=1}^T f_{\theta}(y_t\,|\, y_{t-1},x_{t},a)}}\mathrm{d}y^{1:T}=0,\label{eq_phi_hetero_robust}\end{equation}and (ii), for all $s=2,...,T$,
		\begin{align}
			&\int  \phi_{\theta}(y^T,x^T)\prod_{t=s}^Tf_{\theta}(y_{t}\,|\, y_{t-1},x_{t},a)\mathrm{d}y^{s:T} \textrm{ does not depend on }x^{s:T}.\label{eq_phi_feedback_robust}
		\end{align}
Then, for all $\omega=(g,\pi,\nu)\in\Omega$ such that $\mathbb{E}_{\theta,\omega}\left[\abs{\phi_{\theta}(Y^T,X^T)}\right]<\infty$, we have $$\mathbb{E}_{\theta,\omega}[\phi_{\theta}(Y^T,X^T)]=0.$$ 
(B) Suppose that, (i) the root density $\omega\mapsto \ell^{1/2}\left(\left.\theta,\omega\right|y^{T},x^{T}\right)$ is differentiable in quadratic mean at $\omega^*$ for some $\omega^*=(g^*,\pi^*,\nu^*)\in\Omega$, and (ii) there is a neighborhood $\mathcal{N}$ of $\omega^{*}$ such that $\sup_{\omega \in \mathcal{N}}\mathbb{E}_{\theta,\omega}[\|{\phi_{\theta}(Y^T,X^T)}\|^2]<\infty$. Then, the converse is also true.
\end{theorem}

\vskip .3cm

Recall our notational convention that (\ref{eq_phi_hetero_robust}) holds for almost all possible values of $x^T$, $y_0$, and $a$, while (\ref{eq_phi_feedback_robust}) holds for almost all possible values of $x^{T}$, $y^{s-1}$, and $a$, although not explicitly stated.

Suppose the data are generated according to some $(\theta_0,\omega_0)$. An implication of Part (A) of Theorem \ref{theo_charact} is that any function $\phi_{\theta}$ that satisfies \eqref{eq_phi_hetero_robust} and \eqref{eq_phi_feedback_robust}, and is additionally absolutely integrable under the population DGP ($\theta_0,\omega_0$), has zero mean at the true $\theta_0$, that is,
$$\mathbb{E}_{\theta_0,\omega_0}[\phi_{\theta_0}(Y^T,X^T)]=0.$$
To see this, simply apply Part (A) of Theorem \ref{theo_charact} with $\theta=\theta_0$ and $\omega=\omega_0$.

The first condition for the FHR property, restriction \eqref{eq_phi_hetero_robust}, ensures robustness of $\phi_\theta$ to the presence of heterogeneity of an unknown form. This condition is highlighted in \cite{Bonhomme_EM12}, in a setting with strictly exogenous covariates, as the key condition ensuring valid moment functions for $\theta$. Under strict exogeneity, \eqref{eq_phi_hetero_robust} ensures that $\phi_{\theta}(Y^T,X^T)$ is conditionally mean zero given $X^T$ and $A$. By the law of iterated expectations, this suffices to ensure that it is unconditionally mean zero as well. The functional differencing approach then provides a general recipe for constructing functions $\phi_{\theta}$ satisfying \eqref{eq_phi_hetero_robust}.

To illustrate how the presence of feedback modifies the interpretation of condition \eqref{eq_phi_hetero_robust}, consider the $T=2$ setting. In this case, the condition is
$$\int\left[\int \phi_{\theta}\left(y_0,y_1,y_2,x_1,x_2\right)f_{\theta}(y_2\,|\, y_1,x_2,a)\mathrm{d}y_2\right]f_{\theta}(y_1\,|\, y_0,x_1,a)\mathrm{d}y_1=0,$$
which coincides with
\begin{equation}                   \mathbb{E}\left[\mathbb{E}\left[\left.\phi_{\theta}\left(Y_0,Y_1,Y_2,X_1,X_2\right)\right|Y_0,Y_1,X_1,X_2=x_2,A\right]\,|\,Y_0,X_1,A \right]=0 \quad \text{ for all }x_2.\label{eq_phi_hetero_robust_T2}
\end{equation}
Here the inner expectation corresponds to an average over $y_2$ with respect to its model density, $f_{\theta}(y_{2}\,|\, y_{1},x_{2},a)$. 
This inner expectation conditions on $X_2=x_2$, while the outer expectation, which averages over $y_1$ with respect to its model density, $f_{\theta}(y_{1}\,|\, y_{0},x_{1},a)$, does not condition on $x_2$. This is a key difference between the heterogeneous feedback case, considered here, and the setting with strict exogeneity studied by \cite{Bonhomme_EM12}. Under strict exogeneity, $Y_1$ is independent of $X_2$ conditional on $(Y_0,X_1,A )$, so (\ref{eq_phi_hetero_robust_T2}) coincides with 
$$\mathbb{E}\left[\mathbb{E}\left[\left.\phi_{\theta}\left(Y_0,Y_1,Y_2,X_1,X_2\right)\right|Y_0,Y_1,X_1,X_2,A\right]\,|\,Y_0,X_1,X_2,A \right]=0,$$
with both the inner and outer expectations conditioning on $X_2$. By iterated expectations, this conditional expectation implies a zero mean condition given $A$, initial conditions, and the \emph{entire sequence of covariates},
$$\mathbb{E}\left[\left.\phi_{\theta}\left(Y_0,Y_1,Y_2,X_1,X_2\right)\right|Y_0,X_1,X_2,A\right]=0.$$
However, under feedback it is not generally the case that (\ref{eq_phi_hetero_robust_T2}) can be written as a zero mean condition given the covariates sequence $(X_1,X_2)$.

In the presence of feedback, Theorem \ref{theo_charact} additionally requires the $T-1$ conditions \eqref{eq_phi_feedback_robust}, to ensure that $\phi_\theta$ is a valid moment function. To explain this additional requirement, consider again the $T=2$ setting, in which case \eqref{eq_phi_feedback_robust} reads
$$\int \phi_{\theta}\left(y_0,y_1,y_2,x_1,x_2\right)f_{\theta}(y_2\,|\, y_1,x_2,a)\mathrm{d}y_2\quad \text{does not depend on }x_2,$$
which corresponds to the single mean independence restriction
\begin{equation}                   \mathbb{E}\left[\left.\phi_{\theta}\left(Y_0,Y_1,Y_2,X_1,X_2\right)\right|Y_0,Y_1,X_1,X_2,A\right]=\mathbb{E}\left[\left.\phi_{\theta}\left(Y_0,Y_1,Y_2,X_1,X_2\right)\right|Y_0,Y_1,X_1,A\right].\label{eq_phi_feedback_robust_T2}
\end{equation}
Equation \eqref{eq_phi_feedback_robust_T2} implies that, at the population parameter, $X_2$ does not predict $\phi_{\theta}\left(Y_0,Y_1,Y_2,X_1,X_2\right)$ conditional on $Y_0,Y_1,X_1,A$. Under this condition, which we interpret as ``feedback robustness'' of $\phi_{\theta}$, the conditioning on $X_2=x_2$ disappears in (\ref{eq_phi_hetero_robust_T2}), which implies
\begin{equation*}                   \mathbb{E}\left[\mathbb{E}\left[\left.\phi_{\theta}\left(Y_0,Y_1,Y_2,X_1,X_2\right)\right|Y_0,Y_1,X_1,A\right]\,|\,Y_0,X_1,A \right]=0,
\end{equation*}
ensuring, by iterated expectations, that $\phi_{\theta}$ is a valid moment function. Importantly, $\phi_\theta$ is then a valid moment under both unrestricted heterogeneity and unrestricted feedback: it is FHR.\footnote{This discussion provides intuition about the $T=2$ case. By a similar logic, when $T=3$ achieving feedback robustness requires imposing (\ref{eq_phi_feedback_robust}) for $s=3$ as well.}

Lastly, Theorem \ref{theo_charact} also establishes that \eqref{eq_phi_hetero_robust} and \eqref{eq_phi_feedback_robust} are, under suitable conditions, \emph{necessary} for the FHR property. Necessity of (\ref{eq_phi_hetero_robust}) holds since, if it did not, one could construct a heterogeneity density that would violate the moment condition. Necessity of (\ref{eq_phi_feedback_robust}) follows from a similar argument, since if it did not hold one could construct a feedback density violating the moment condition. To show necessity, we assume quadratic mean differentiability at $(\theta,\omega^*)$ as stated in Part (B) Condition (i). This is a standard regularity condition in semiparametric estimation, see for example \citet[][Ch. 25]{van2000asymptotic}. Part (B) Condition (ii), which imposes square-integrability of the moment function in a neighborhood of $\omega^*$, is similarly standard; e.g., see the assumptions for the generalized information equality in Lemma 5.4 of \cite{newey1994large}.\footnote{One can show that, if $\phi_{\theta}$ is bounded, then the necessity of \eqref{eq_phi_hetero_robust} and \eqref{eq_phi_feedback_robust} obtains directly without reference to $\omega^*$ or differentiability in quadratic mean.} We will rely on differentiability in quadratic mean in our analysis of efficiency in Section \ref{sec: SEBs}.

\subsection{An alternative characterization of FHR moments for $\theta$}

The following corollary to Theorem \ref{theo_charact} facilitates the construction of FHR moment functions in practice.

\begin{corollary}{\textsc{(Alternative FHR Moment Characterization)}}\label{coro_charact} Let $\omega^*\in\Omega$. An absolutely-integrable function $\phi_{\theta}$ under $(\theta,\omega^*)$  satisfies \eqref{eq_phi_hetero_robust} and \eqref{eq_phi_feedback_robust} if and only if there exist absolutely-integrable functions $\psi_{\theta,t}$ under $(\theta,\omega^*)$, for $t=1,\ldots,T-1$, such that:
\begin{align}\mathbb{E}\left[\left.\phi_{\theta}(Y^T,X^T)\right|Y^{T-1},X^T,A\right]=\sum_{t=1}^{T-1}\psi_{\theta,t}(Y^t,X^t,A), \label{eq_phi_coro}
\end{align}
with $\mathbb{E}\left[\left.\psi_{\theta,t}\left(Y^{t},X^{t},A\right)\right|Y^{t-1},X^{t},A\right]=0$ for $t=1,\ldots,T-1$.
\end{corollary}

\vskip .3cm

To understand Corollary \ref{coro_charact}, it is helpful to return to the $T=2$ setting. In this case, letting
$$\psi_{\theta,1}(Y_0,Y_1,X_1,A)=\mathbb{E}[{\phi_{\theta}(Y_0,Y_1,Y_2,X_1,X_2)}\,|\, Y_0,Y_1,X_1,A],$$
it follows from (\ref{eq_phi_feedback_robust}) that
$$\psi_{\theta,1}(Y_0,Y_1,X_1,A)=\mathbb{E}[{\phi_{\theta}(Y_0,Y_1,Y_2,X_1,X_2)}\,|\, Y_0,Y_1,X_1,X_2,A],$$
which implies (\ref{eq_phi_coro}), whereas it follows from (\ref{eq_phi_hetero_robust}) that
$$\mathbb{E}[\psi_{\theta,1}(Y_0,Y_1,X_1,A)\,|\, Y_0,X_1,A]=0.$$

The representation provided by Corollary
\ref{coro_charact} suggests a systematic recipe for constructing FHR moment functions. The first step involves finding functions  $\psi_{\theta,t}(Y^{t},X^t,A)$ that are mean zero conditional on $Y^{t-1},X^t,A$ for $t=1,\ldots,T-1$. This is straightforward since any function of $(Y^{t},X^t,A)$, suitably de-meaned, automatically fulfills this requirement. Then, given a collection of $\psi_{\theta,t}(Y^{t},X^t,A)$ functions, the second step requires solving a linear integral equation to recover a valid moment function $\phi_{\theta}$. Indeed, (\ref{eq_phi_coro}) can equivalently be written as
\begin{align}
    	\int \phi_{\theta}(y^T,x^T) f_{\theta}(\left.y_T\right|y_{T-1},x_{T},a)\mathrm{d}y_T=\sum_{t=1}^{T-1}\psi_{\theta,t}(y^t,x^t,a), \label{eq_phi_coro_int}
\end{align}
which is known as an inhomogeneous Fredholm equation of the first kind. Note that the integral operator on the left-hand side of \eqref{eq_phi_coro_int} is known given the parameter $\theta$. Solution methods for this type of linear integral equation are the subject of a large literature \citep[e.g.,][]{engl1996regularization,carrasco2007linear}.

\paragraph{Simple moments.} A special case of Corollary
\ref{coro_charact} obtains when one is able to find functions $\eta_{\theta,t}$ such that  
\begin{equation}
\mathbb{E}\left[\left.\eta_{\theta,t}\left(Y^{t},X^{t}\right)\right|Y^{t-1},X^{t},A\right]=b\left(Y_0,X_1,A\right),\label{eq: AB_condition}
\end{equation}
for some function $b$ of the heterogeneity and initial conditions. Since the expectation in \eqref{eq: AB_condition} is with respect to $f_{\theta}(\left.y_t\right|y_{t-1},x_{t},a)$, well-known properties of parametric distributions often suggest candidates for $\eta_{\theta,t}$. With a suitable $\eta_{\theta,t}$ in hand, the instrumented first difference 
\begin{equation*}\phi_{\theta}(y^T,x^T)=[\eta_{\theta,T}\left(y^{T},x^{T}\right)-\eta_{\theta,T-1}\left(y^{T-1},x^{T-1}\right)]\cdot m(y^{T-2},x^{T-1})\end{equation*}
satisfies (\ref{eq_phi_coro}), for $\psi_{\theta,T-1}(y^{T-1},x^{T-1},a)=[b\left(y_0,x_1,a\right)-\eta_{\theta,T-1}\left(y^{T-1},x^{T-1}\right)]\cdot m(y^{T-2},x^{T-1})$ and $\psi_{\theta,t}=0$ for all $t<T-1$ (for an arbitrary function $m$). This provides a FHR moment function on $\theta$. More generally, one can check that \begin{equation}\phi_{\theta}(y^t,x^t)=\left[\eta_{\theta,t}\left(y^{t},x^{t}\right)-\eta_{\theta,t-1}\left(y^{t-1},x^{t-1}\right)\right]\cdot m(y^{t-2},x^{t-1}),\quad t=2,...,T,\label{eq_AB}\end{equation} all satisfy (\ref{eq_phi_coro}), thus providing additional moments.

Moments of the form (\ref{eq_AB}) were considered by \citet{arellano1991some}
in the linear model context, by \citet{Chamberlain_JBES92,Chamberlain_JOE2022}
and \citet{wooldridge1997multiplicative} for Poisson models, and by \cite{Al-Sadoon_et_al_ER2017} for certain binary choice models. However, this family of estimating equations does not exhaust all available FHR moments, and may lead to estimators with low levels of asymptotic precision in practice. By comparison, Theorem \ref{theo_charact} and Corollary \ref{coro_charact} characterize \emph{all} available FHR moments.

\subsection{FHR moments for $\theta$ in the MPH model} \label{subsec: MPH_FHR}

Our analysis exploits several special features of the MPH model, which we introduce first (see Supplemental Appendix \ref{app: MPH_derivations} for details). Recall that $z_t=(y_t,y_{t-1},x_t')'$ and $\rho_{\theta}\left(z_{t}\right)=\Lambda_{\alpha}\left(y_{t}\right)\exp\left(\gamma y_{t-1}+x_{t}'\beta\right)$. Adapting arguments appearing in \cite{Lancaster_EATD1990}, \cite{hahn1994efficiency} and \cite{Ridder_Woutersen_EM2003}, it is straightforward to show that
\begin{equation} \label{eq: rho_is_exponential}
    \left.\rho_{\theta}\left(Z_{t}\right)\right|Y^{t-1},X^{t},A\sim\mathrm{Exponential}\left(e^{A}\right),\thinspace t=1,\ldots,T.
\end{equation}

An immediate consequence of \eqref{eq: rho_is_exponential} is that \begin{equation} \label{eq: mean_of_rho}
    \mathbb{E}\left[\left.\rho_{\theta}\left(Z_{t}\right)\right|Y^{t-1},X^{t},A\right]=e^{-A},\thinspace t=1,\ldots,T,
\end{equation}
which coincides with \eqref{eq: AB_condition} after setting $\eta_{\theta,t}(Y^t,X^t)=\rho_{\theta}(Z_t)$ and $b(Y_0,X_1,A)=e^{-A}$. This yields the FHR moment functions 
\begin{equation} \label{eq: AB_moments_MPH}
    \phi_{\theta}\left(Y^{t},X^{t}\right)=\left[\rho_{\theta}\left(Z_{t}\right)-\rho_{\theta}\left(Z_{t-1}\right)\right]\cdot m\left(Y^{t-2},X^{t-1}\right),\quad t=2,...,T,
\end{equation}
where $m$ is an arbitrary function. In unpublished dissertation research, \citet[][Ch. 1]{woutersen2000essays} presented moments similar to \eqref{eq: AB_moments_MPH}.

More generally, it turns out that the MPH model admits an explicit characterization of \emph{all} FHR moments. To proceed, let $P_{\theta,t}=\rho_{\theta}\left(Z_t\right)$. By (\ref{eq: rho_is_exponential}), conditional on $Y_{0},X_{1},A$ the $P_{\theta,t}$ for $t=1,\ldots,T$ are independent exponential random variables, each with a common rate parameter of $e^{A}$. Next, we define a one-to-one transformation of the vector $P_{\theta}^T$ into a ``forward orthogonal deviations'' part and a ``between'' part, as follows:
\begin{equation} \label{eq: MPH_within_between}\widetilde{P}_{\theta,t}=\frac{\rho_{\theta}\left(Z_{t}\right)}{\sum_{s=t}^{T}\rho_{\theta}\left(Z_{s}\right)},\ t=1,\ldots,T-1; \quad \overline{P}_{\theta}=\sum_{t=1}^{T}\rho_{\theta}\left(Z_{t}\right).
\end{equation}
Observe that $\widetilde{P}_{\theta,t}$ involves the ratio of $\rho_{\theta}\left(Z_{t}\right)$ to the sum of itself and the \emph{future} values of $\rho_{\theta}\left(Z_{s}\right)$ for $s=t+1,\ldots,T$. In Supplemental Appendix \ref{app: MPH_derivations} we establish that, conditionally on $Y_{0},X_{1},A$:
\begin{equation} \label{eq: MPH_within_between2}
\widetilde{P}_{\theta,t}\sim  \mathrm{Beta}\left(1,T-t\right);\quad  t=1,\ldots,T-1,\quad\overline{P}_{\theta}\sim  \mathrm{Gamma}\left(T,e^{A}\right),
\end{equation}
where, throughout, $\theta$ is the parameter indexing the DGP, and the elements of $\left(\widetilde{P}_{\theta,1},\ldots,\widetilde{P}_{\theta,T-1},\overline{P}_{\theta}\right)$ are mutually independent of one another.

Transformation \eqref{eq: MPH_within_between} can be thought of as a MPH-specific analog of the forward orthogonal deviations transformation used by \cite{arellano1995another} in the context of linear panel data models with predetermined regressors. It has the property that the random variables $\widetilde{P}_{\theta,t}$ are independent of both contemporaneous and lagged predetermined covariates as well as lagged values of the spell outcomes $\{(X_{s},Y_{s-1})\}_{s=1}^t$. Appealing to the same analogy, we can think of $\overline{P}_{\theta}$ as containing the ``between'' variation in $P_{\theta}^T$.

With these preliminaries in place, we can state the following exhaustive FHR moment characterization for the MPH model with feedback.

\begin{lemma}
\textsc{(Characterization of FHR moment restrictions for the MPH model)} \label{lem_moment_charac_MPH}Consider the MPH model with $T\geq2$. Let $\omega^*\in\Omega$. Then an absolutely-integrable $\phi_{\theta}$ under $(\theta,\omega^*)$ satisfies \eqref{eq_phi_hetero_robust} and \eqref{eq_phi_feedback_robust} if and only if there exist absolutely-integrable $\psi_{\theta,t}$ under $(\theta,\omega^*)$, for $t=1,...,T-1$, such that:
\begin{align*}
\phi_{\theta}\left(Y^{T},X^{T}\right)=&\sum_{t=1}^{T-1}\psi_{\theta,t}(Y_{0},\widetilde{P}_{\theta}^{t},\overline{P}_{\theta},X^{t}),
\end{align*}
where, for $t=1,\ldots,T-1$,
\[
    \mathbb{E}\left[\left.\psi_{\theta,t}(Y_{0},\widetilde{P}_{\theta}^{t},\overline{P}_{\theta},X^{t})\right|Y_{0},\widetilde{P}_{\theta}^{t-1},\overline{P}_{\theta},X^{t}\right]=0.
\]
\end{lemma}

\vskip .3cm

The proof of Lemma \ref{lem_moment_charac_MPH} is available in Supplemental Appendix \ref{app_B}. It represents a useful simplification, induced by the special structure of the MPH model, relative to the general result of Theorem \ref{theo_charact}. Note in particular that the latent heterogeneity $A$ does not appear in the characterization of Lemma \ref{lem_moment_charac_MPH}. To illustrate, consider the $T=2$ setting. In this case the lemma implies that all FHR moment functions are of the form
$$\phi_{\theta}\left(Y_{0},Y_1,Y_2,X_1,X_2\right)=\psi_{\theta,1}(Y_{0},\widetilde{P}_{\theta,1},\overline{P}_{\theta},X_1),$$
where $\mathbb{E}\left[\left.\psi_{\theta,1}(Y_{0},\widetilde{P}_{\theta,1},\overline{P}_{\theta},X_1)\right|Y_{0},\overline{P}_{\theta},X_1\right]=0$. We seek functions $\psi_{\theta,1}$ of the forward orthogonal deviations
$\widetilde{P}_{\theta,1}$, that are mean zero conditional on the between variation
$\overline{P}_{\theta}$, as well as the initial condition $\left(Y_{0},X_{1}\right)$.
It is straightforward to construct such functions by de-meaning, and we will exploit this property in our analysis of efficiency.

For the general $T$ period case, $\phi_{\theta}$ is the sum of functions
$\psi_{\theta,t}\left(Y_{0}, \widetilde{P}_{\theta}^{t},\overline{P}_{\theta},X^{t}\right)$,
for $t=1,\ldots,T-1$, that are mean independent of the between variation, $\overline{P}_{\theta}$, current
and past values of the predetermined regressors, $X^{t}$, and past
values of the spell outcomes, $Y^{t-1}$. This structure
is reminiscent of how moment conditions are typically constructed
for linear panel data models with predetermined regressors (see, especially,
\citealp{arellano1995another}). Moreover, our characterization can be used to derive efficient estimators.

\subsection{Failures of moment-based identification}\label{sec: identification_failures}

In this subsection we present two examples where the only FHR moment function is the zero function, such that GMM estimation based on FHR moments is impossible. Such failures are not unique to our heterogeneous feedback setting, as \cite{chamberlain2010binary} showed for binary choice panel data models under strict exogeneity.

\begin{example} {(\textsc{Binary Choice Model})}   
	Consider a binary choice logit model with continuous heterogeneity and sequentially exogenous covariates:
	$$\Pr(Y_{t}=1\,|\, Y_{t-1}=y_{t-1},X_{t}=x_t,A=a;\theta)=\frac{\exp(\gamma y_{t-1}+\beta'x_t+a)}{1+\exp(\gamma y_{t-1}+\beta'x_t+a)},\quad t=1,2.$$
	Any valid moment function of $\theta=(\gamma,\beta')'$ needs to satisfy (\ref{eq_phi_feedback_robust}), that is,
	\begin{align*}
		\sum_{y_2=0,1}\phi_{\theta}(y_0,y_1,y_2,x_1,x_2)\frac{\exp(\gamma y_1+\beta'x_2+a)^{y_2}}{1+\exp(\gamma y_1+\beta'x_2+a)}	\text{ does not depend on $x_2$.}
	\end{align*}
	Suppose that $\beta\neq 0$. The only functions $\phi_{\theta}$ satisfying this restriction do not depend on $y_2$ or $x_2$. Hence, by (\ref{eq_phi_hetero_robust}) we obtain
	\begin{align*}
		\sum_{y_1=0,1}\phi_{\theta}(y_0,y_1,x_1)\frac{\exp(\gamma y_0+\beta'x_1+a)^{y_1}}{1+\exp(\gamma y_0+\beta'x_1+a)}=0,
	\end{align*}
	from which we obtain that $\phi_{\theta}=0$. This shows the absence of non-trivial moment restrictions for $\theta$ in this model. Building on \citet{Chamberlain_JOE2022,chamberlain2023identification}, \citet{bonhomme2023identification} study the failure of point-identification in binary choice models with sequentially exogenous covariates, and show how to compute identified sets on parameters and average effects.\footnote{As  \citet{bonhomme2023identification} note, imposing assumptions on the feedback process, such as Markovianity, may lead to non-trivial moment restrictions in discrete choice and other models where an approach allowing for fully unrestricted feedback is uninformative. Extending our approach to accommodate such additional assumptions is an interesting topic for future work.}
	\end{example}

 \begin{example} {(\textsc{Random Coefficients Model})}   
Consider the Gaussian linear random coefficients model
	\begin{equation}
		Y_{t}=\gamma Y_{t-1}+B'X_{t}+C+\varepsilon_{t},\quad \varepsilon_{t}\,|\, Y^{t-1},X^t,A \sim {\cal{N}}(0,\sigma^2),\label{eq_RC}
	\end{equation}
	where $A=(B',C)'$ is multidimensional. Any moment function on $\theta=(\gamma,\sigma^2)'$ needs to satisfy
	\begin{equation}\label{eq_lin2}
		\int \phi_{\theta}(y_0,y_1,y_2,x_1,x_2)\exp\left(-\frac{1}{2\sigma^2}\left(y_2-\gamma y_1-b'x_2-c\right)^2\right)\mathrm{d}y_2 \text{ does not depend on $x_2$.}
	\end{equation}
	Note that, if (\ref{eq_lin2}) holds, then, for all $b$ and $x_2,\widetilde{x}_2$,
	$$ \phi_{\theta}(y_0,y_1,y_2+b'x_2,x_1,x_2)=\phi_{\theta}(y_0,y_1,y_2+b'\widetilde{x}_2,x_1,\widetilde{x}_2).$$
	This implies that $\phi_{\theta}$ does not depend on $y_2$ or $x_2$. Then (\ref{eq_phi_hetero_robust}) implies
	\begin{equation*}
		\int \phi_{\theta}(y_0,y_1,x_1)\exp\left(-\frac{1}{2\sigma^2}\left(y_1-\gamma y_0-b'x_1-c\right)^2\right)\mathrm{d}y_1 =0,
	\end{equation*}
	from which it follows that $\phi_{\theta}=0$. This shows the only FHR moment function on $\theta$ in this model is the null function. This negative result echoes an example in \citet{Chamberlain_JOE2022}.\footnote{However, some models with multivariate heterogeneity and feedback admit informative FHR moments. An example is a linear model with both predetermined and strictly exogenous covariates, where the coefficients of the former are homogeneous while the coefficients of the latter are heterogeneous. The negative result in (\ref{eq_RC}) is specific to the random coefficient $B$ appearing in front of the predetermined covariate $X_t$.} 
\end{example}

\noindent These two examples illustrate challenges that may arise when dealing with discrete outcomes or multidimensional heterogeneity. However, even if informative FHR moments are unavailable, the identified set may nevertheless be bounded, as recently shown by \citet{bonhomme2023identification} and \citet{lee2026identification} for the logit and random coefficients models, respectively. In some cases, $\theta$ may be point-identified despite the set of FHR moments being uninformative.\footnote{As an example, let $T=2$ and let $Y_{t}=\theta+A X_{t}+\varepsilon_{t}$ with $X_{t}$ continuously distributed on $\mathbb{R}$ and $\varepsilon_{t}\,|\, Y^{t-1},X^t,A \sim {\cal{N}}(0,1)$. Applying a similar logic to that in model \eqref{eq_lin2}, one can show that $\phi_{\theta}=0$ since the assumptions imply that $P(X_{2}=0)=P(X_{1}=0)=0$. However, this is a case where the parameter of interest $\theta$ is \textit{identified at 0} since $\lim_{x \to 0} \mathbb{E}[Y_{1}|X_{1}=x]=\theta$ (\citealp{Graham_Powell_EM12}).\label{footnote_irregular}} In this paper, we emphasize the possibilities which arise when the set of FHR moments is both nontrivial and point-identifying, deferring an exploration of other approaches to identification in models with heterogeneous feedback to future work.

\section{FHR moment restrictions for average effects}\label{sec: moments_average}

In this section we characterize all FHR moment conditions for average effects, $$\mu(\theta,\omega)=\mathbb{E}_{\theta,\omega}\left[h_{\theta}\left(Y^{T},X^{T},A\right)\right],$$ where $h_{\theta}\left(Y^{T},X^{T},A\right)$ is a known function of $Y^{T},X^{T},A$ indexed by $\theta$.

\begin{example}[continues=ex: mph_intro] (\textsc{Average Effects in the MPH Model}) Consider the \emph{average structural hazard (ASH)} function
\begin{align} \label{WMPH_ashf}
        \overline{\lambda}(y_{t}|y_{t-1},x_t)&=\lambda_{\alpha}(y_{t})e^{x_t'\beta+\gamma y_{t-1}}\mathbb{E}_{\theta,\omega}\left[e^A\right].
\end{align} 
The ASH corresponds to the \emph{expected} hazard, at spell time $y_t$, for a unit exogenously assigned to policy $X_t=x_t$ and history $Y_{t-1}=y_{t-1}.$ Like other quantities relevant to causal analysis, the ASH depends on the (marginal) distribution of unobserved heterogeneity $A$.

The ASH may be compared with the \emph{observed} hazard function
\[
\lambda\left(\left.y_{t}\right|y_{t-1},x_{t}\right)=\lambda_{\alpha}(y_{t})e^{x_{t}'\beta+\gamma y_{t-1}}\mathbb{E}_{\theta,\omega}\left[\left.e^{A}\right|Y_{t}>y_{t},Y_{t-1}=y_{t-1},X_{t}=x_{t}\right],
\]
which, as is well-known (e.g., \citealp{Lancaster_EATD1990,Heckman_AER1991}), suffers from spurious duration dependence: units with higher values of $A$ will exit
earlier, implying that the mean $\mathbb{E}_{\theta,\omega}\left[\left.e^{A}\right|Y_{t}>y_{t},Y_{t-1}=y_{t-1},X_{t}=x_{t}\right]$
declines with $y_{t}$. In contrast, the ASH is free of heterogeneity bias. The difference $\overline{\lambda}\left(\left.y_{t}\right|y^{\prime}_{t-1},x^{\prime}_{t}\right)-\overline{\lambda}\left(\left.y_{t}\right|y_{t-1},x_{t}\right)$ measures the expected difference in the shape of the hazard function across two units exogenously assigned to different prior spell lengths and current spell treatments. Such contrasts allow us to better understand the causal effect of manipulations of $x_t$ on exit behavior, as well as the influence of past spell lengths on current ones. Although the ASH estimand appears to be new, it is related to the  \emph{average structural function (ASF)} of \cite{Blundel_Powell_WC03}.\footnote{Letting $m(x_t,y_{t-1},a;\theta)=\mathbb{E}\left[Y_t|X_t=x_t,Y_{t-1}=y_{t-1},A=a\right]$, the ASF is
\begin{align} \label{WMPH_asf}
        \mu(x_t,y_{t-1})&=\mathbb{E}_{\theta,\omega}\left[m(x_t,y_{t-1},A;\theta)\right].
\end{align}}

Our definition of ``average effects'' is expansive enough to also include estimands involving the feedback process. As an example, with $T=2$, consider  the \emph{feedback projection coefficient}, defined as $\mu$ in the projection equation 
\begin{align} \label{WMPH_fpc}
    X_2 = \mu Y_1 + \eta'R + U, \quad \mathbb{E}_{\theta,\omega}\left[UY_1\right] =0, \quad \mathbb{E}_{\theta,\omega}\left[UR\right] =0,
\end{align}
where $R = (Y_0,X_1,A)$. The parameter $\mu$ measures the predictive strength of $Y_1$ for  $X_2$, controlling for unobserved heterogeneity $A$ and the initial condition $(Y_0,X_1)$. As such, $\mu$ provides a natural population measure of feedback.
\end{example}

In nonlinear panel data models, knowledge of $\theta$ alone does not suffice to identify the effect of an external manipulation of a regressor's value on the probability distribution of the outcome. In contrast, estimands which average over the marginal distribution of $A$, such as \eqref{WMPH_ashf}, provide easy-to-interpret summaries of such effects. Until recently the identifiability of such averages was not well understood. Recent work by \cite{honore2006bounds}, \cite{chernozhukov2013average}, \cite{aguirregabiria2021identification}, \cite{davezies2021identification}, \cite{dobronyi2021identification}, \cite{pakel2023bounds} and others, however, has shown that average effects are (partially) identified in several settings under strict exogeneity. Here we study average effects in models with feedback.

\subsection{Characterization of FHR moments for $\mu(\theta,\omega)$}

Our first result is an analog of Theorem \ref{theo_charact} for $\mu(\theta,\omega)$.
\begin{theorem}{\textsc{(Characterization of moment restrictions for average effects)}}\label{theo_charact_aveff} Let $\theta\in\Theta$. (A) Suppose that the following $T$ restrictions hold: (i)
		\begin{equation}\int \left(\varphi_{\theta}(y^T,x^T)-h_{\theta}(y^T,x^T,a)\right){{\prod_{t=1}^T f_{\theta}(y_t\,|\, y_{t-1},x_{t},a)}}\mathrm{d}y^{1:T}=0,\label{eq_aveff}\end{equation}and (ii), for all $s=2,...,T$,
		\begin{align}
			&\int \left(\varphi_{\theta}(y^T,x^T)-h_{\theta}(y^T,x^T,a)\right)\prod_{t=s}^Tf_{\theta}(y_{t}\,|\, y_{t-1},x_{t},a)\mathrm{d}y^{s:T} \textrm{ does not depend on }x^{s:T}.\label{eq_aveff_feedback_robust}
		\end{align}
Then, for all $\omega\in\Omega$ such that $\mathbb{E}_{\theta,\omega}\left[\abs{\varphi_{\theta}(Y^T,X^T)}\right]<\infty$ and $\mathbb{E}_{\theta,\omega}\left[\abs{h_{\theta}(Y^T,X^T,A)}\right]<\infty$ we have $$\mathbb{E}_{\theta,\omega}[\varphi_{\theta}(Y^T,X^T)]=\mu(\theta,\omega).$$ 
(B) Suppose (i) the root density of the complete data $(Y^T,X^T,A)$ at $(\theta,\omega)$ is differentiable in quadratic mean in $\omega$ at $\omega^*$ for some $\omega^*=(g^*,\pi^*,\nu^*)\in\Omega$, (ii) there is a neighborhood $\mathcal{N}$ of $\omega^{*}$ such that $\sup_{\omega \in \mathcal{N}}\mathbb{E}_{\theta,\omega}[\norm{\varphi_{\theta}(Y^T,X^T)}^2]<\infty$ and $\sup_{\omega \in \mathcal{N}}\mathbb{E}_{\theta,\omega}[\|h_{\theta}(Y^T,X^T,A)\|^2]<\infty$. Then, the converse is also true. 
\end{theorem}

\vskip .3cm

Note that Part (A) in Theorem \ref{theo_charact} follows directly from the corresponding part in Theorem \ref{theo_charact_aveff}, by taking $h_{\theta}=0$. As in the case of $\theta$, \eqref{eq_aveff} and \eqref{eq_aveff_feedback_robust} imply that, under absolute integrability, the true value $\mu_0=\mu(\theta_0,\omega_0)$ satisfies a moment condition. Indeed, if 
$\mathbb{E}_{\theta_0,\omega_0}\left[\abs{\varphi_{\theta_0}(Y^T,X^T)}\right]<\infty$ and $\mathbb{E}_{\theta_0,\omega_0}\left[\abs{h_{\theta_0}(Y^T,X^T,A)}\right]<\infty$, then Part (A) implies
$$\mathbb{E}_{\theta_0,\omega_0}[\varphi_{\theta_0}(Y^T,X^T)]=\mu_0.$$
This moment condition has the FHR property: \eqref{eq_aveff} ensures that $\varphi_{\theta}$ is robust to an unknown distribution of unobserved heterogeneity, whereas the $T-1$ conditions \eqref{eq_aveff_feedback_robust} endow $\varphi_{\theta}$ with robustness to heterogeneous feedback of an unknown form.

We additionally state the following corollary, which mimics Corollary \ref{coro_charact} and suggests a systematic recipe for constructing functions $\varphi_{\theta}$.
\begin{corollary}{\textsc{(Alternative  characterization of moment restrictions for average effects)}}\label{coro_charact_aveff} Let $\omega^*\in\Omega$. Suppose that $h_{\theta}$ is absolutely integrable under $(\theta,\omega^*)$. An absolutely-integrable $\varphi_{\theta}$ satisfies \eqref{eq_aveff} and \eqref{eq_aveff_feedback_robust} if, and only if, there exist absolutely-integrable $\zeta_{\theta,t}$, for $t=1,\ldots,T-1$, such that: 
\begin{align}
    	\mathbb{E}\left[\varphi_{\theta}(Y^T,X^T)-h_{\theta}(Y^T,X^T,A)\,|\,  Y^{T-1},X^T,A\right]=\sum_{t=1}^{T-1}\zeta_{\theta,t}(Y^{t},X^{t},A), \label{eq_phi_alt}
\end{align}
with $\mathbb{E}\left[\zeta_{\theta,t}(Y^{t},X^{t},A)\,|\,Y^{t-1},X^{t},A\right]=0$ for $t=1,\ldots,T-1$.
\end{corollary}

\vskip .3cm

\subsection{FHR moments for the average structural hazard} \label{subsec: ASH}

Recall the average structural hazard (ASH) estimand introduced in \eqref{WMPH_ashf} above. Since conditional on $(Y_{0},X_{1},A)$, $\overline{P}_{\theta}\sim\mathrm{Gamma}\left(T,e^{A}\right)$,
we have\begin{align}
\mathbb{E}\left[\left.\overline{P}_{\theta}^{\delta}\right|Y_{0},X_{1},A\right] & =e^{-\delta A}\frac{\Gamma\left(T+\delta\right)}{\Gamma\left(T\right)}\label{eq: gamma_moments}
\end{align}
for $\delta>-T$ (also $T\geq2$). Using \eqref{eq: gamma_moments} with $\delta=-1$
shows that $\mathbb{E}\left[e^{A}\right]=\mathbb{E}\left[\frac{T-1}{\overline{P}_{\theta}}\right]$, such that
\begin{equation}
\overline{\lambda}\left(\left.y_{t}\right|y_{t-1},x_{t}\right)=\lambda_{\alpha}(y_{t})e^{x_{t}'\beta+\gamma y_{t-1}}\mathbb{E}\left[\frac{T-1}{\overline{P}_{\theta}}\right].\label{eq: ashf_identified}
\end{equation}

We have thus found one FHR moment function for the ASH. Now, for any other FHR moment function  for the ASH, $\varphi_{\theta}$, we have that
$$\phi_{\theta}(Y^T,X^T)=\varphi_{\theta}(Y^T,X^T)-\lambda_{\alpha}(y_{t})e^{x_{t}'\beta+\gamma y_{t-1}}\frac{T-1}{\overline{P}_{\theta}}$$
is a FHR moment function for $\theta$ (where here $y_t,y_{t-1},x_t$ are fixed values at which the ASH is evaluated). Since all such moment functions $\phi_{\theta}$ are fully characterized in the MPH case by Lemma \ref{lem_moment_charac_MPH}, this characterizes all FHR moment functions for the ASH. We refer readers to Supplemental Appendix \ref{app: MPH_derivations} for results on the \emph{average structural function} and the \emph{feedback projection coefficient} introduced previously. 


\section{Efficient moment restrictions\label{sec: SEBs}}

Theorems \ref{theo_charact} and \ref{theo_charact_aveff} characterize the complete set of moment conditions available for estimating $\theta$ and $\mu(\theta,\omega)$. When this set is nonempty, GMM estimation at parametric rates may be (and often is) feasible. Moreover, many valid moment restrictions may be available (see Lemma \ref{lem_moment_charac_MPH} for the case of the MPH model). In such settings, semiparametric efficiency bound (SEB) theory provides a useful tool for selecting specific moments for estimation purposes. 

In this section we derive the form of the efficient scores for $\theta$ and $\mu(\theta,\omega)$ and, consequently, their corresponding SEBs. As we shall see, the characterizations presented earlier facilitate the derivation of these bounds. We illustrate the application of our results via a detailed analysis of efficiency bounds for the MPH model.\footnote{\cite{hahn1994efficiency} derived the information bound for $\theta$ in the single-spell case studied by \cite{elbers1982true} and \cite{Heckman_Singer_ReStud1984}. He also derived the bound for the multi-spell case under strict exogeneity (see also \citealp{chamberlain1985heterogeneity}). Relative to prior work, our analysis covers the case with feedback and, additionally, considers average effects.}

\subsection{Identification}

Prior to discussing efficiency, it is helpful to introduce two notions of identification, corresponding to two identified sets. The first one, $\Theta^{\rm I}$, is the conventional identified set based on the semiparametric model. Following \cite{honore2006bounds} and \cite{bonhomme2023identification}, $\Theta^{\rm I}$ is the set of $\theta \in \Theta$ for which there exists $\omega=\left(g,\pi,\nu\right)\in\Omega$ such that the likelihood $\ell(\theta,{\omega}\,|\, Y^T,X^T)$ is consistent with the observed data $\ell(\theta_0,{\omega}_0\,|\, Y^T,X^T)$. Here,  $\theta_0$ and $\omega_{0}=(g_0,\pi_0,\nu_0)$ denote the population parameter values. By definition, $\Theta^{\rm I}$ contains all identifying information regarding the common parameter. The second identified set, $\Theta^{\rm FHR}$, is the set of $\theta\in\Theta$ that satisfy the moment restriction $\mathbb{E}_{\theta_0,\omega_0}[\phi_{\theta}(Y^T,X^T)]=0$ for all the FHR functions $\phi_{\theta}$ given by Theorem~\ref{theo_charact}.

By definition of FHR moments, $\Theta^{\rm I}\subset \Theta^{\rm FHR}$. More precisely, we show in Supplemental Appendix \ref{app: FHR_outerset} that $\Theta^{\rm FHR}$ exploits all the restrictions of the semiparametric model, except for the fact that the density of heterogeneity $\pi$ and the feedback process $g$ are non-negative. In some instances, these inequality restrictions contain valuable identifying information regarding the parameter. For example, our analysis of the logit and random coefficients models has highlighted that, while $\Theta^{\rm I}$ is bounded in these models, there are no nontrivial FHR moments and $\Theta^{\rm FHR}$ is the entire parameter space $\Theta$.

Since we study GMM estimation using FHR moments, in the rest of this section we assume $\theta_0$ is point-identified by them, as formalized by the following assumption.

\begin{assumption}\label{ass_ident}
$\Theta^{\rm FHR}=\{\theta_0\}$ is a singleton.
\end{assumption}

When focusing on an average effect $\mu$, we will assume in addition to Assumption \ref{ass_ident} that $\mu$ is point-identified given FHR moments as well, see equation (\ref{eq_varphi}) in the next subsection. 

Assumption \ref{ass_ident} may fail for several reasons. For one, $\theta_0$ may not be point-identified in the semiparametric model, as in the logit model studied by \citet{bonhomme2023identification}. It is also possible that $\Theta^{\rm I}$ is a singleton, but $\Theta^{\rm FHR}$ is not. In such cases, a natural route is to follow a partial identification approach, as recently proposed in \citet{lee2026identification}, \citet{botosaru2024adversarialapproachidentification}, and \citet{chesher2024robust}. Note, however, that Assumption~\ref{ass_ident} does not presume that identification is
regular: the information for $\theta_0$ derived in the next subsection may be zero, in
which case the semiparametric variance bound is infinite.


Assumption \ref{ass_ident} holds in a variety of models with continuous outcomes and scalar heterogeneity. In practice, to verify that $\theta_0$ is point-identified in the sense of Assumption \ref{ass_ident}, it suffices to construct a set of FHR moment functions that guarantee identification. We now illustrate this approach in the MPH model. 

\begin{example}[continues=ex: mph_intro] \label{ex: mph_Weibull_identification}(\textsc{Identification of the Weibull MPH})
Consider the MPH model with a Weibull baseline hazard: $\lambda_{\alpha}(y_t)=\alpha y_t^{\alpha-1}, \alpha>0$. By Lemma \ref{lem_moment_charac_MPH} the following two moment functions are FHR
\begin{align*}
    \phi_{\theta}^{a}(Y^2,X^2)=\ln \left(1-\widetilde{P}_{\theta,1}\right)-\ln \widetilde{P}_{\theta,1}, \quad \phi_{\theta}^{b}(Y^2,X^2)=2+\ln \left(1-\widetilde{P}_{\theta,1}\right)+\ln \widetilde{P}_{\theta,1}.
\end{align*}
Intuitively, $\phi_{\theta}^{a}$ can be viewed as the log-transformed analogs of \eqref{eq: AB_moments_MPH} while $\phi_{\theta}^{b}$ is a mean-zero function that exploits the fact that $\widetilde{P}_{\theta_0,1}\sim U[0,1]$ and leverages symmetry (it is also, coincidentally, a component of the efficient score for $\alpha$ under the Weibull baseline hazard, derived and presented in Supplemental Appendix \ref{app: MPH_derivations}). Consider then the moment vector $ \phi_{\theta}(Y^2,X^2)=\left((Y_{0},X_1')' \phi_{\theta}^{a}(Y^2,X^2), \phi_{\theta}^{b}(Y^2,X^2)\right)'$. 
In Theorem \ref{theo_identif_weibull_MPH}, we establish that the unconditional moment equation $\mathbb{E}\left[\phi_{\theta}(Y^2,X^2)\right]=0$ admits $\theta_{0}$ as its only solution under suitable regularity conditions.  This identification result generalizes that of \cite{woutersen2000essays} for the Exponential hazard case, which corresponds to $\alpha=1$.\footnote{\cite{elbers1982true}, \cite{Heckman_Singer_ReStud1984}, and \cite{Ridder_Woutersen_EM2003} established the nonparametric identification of the single-spell MPH model under independence assumptions between the covariates and the heterogeneity distribution. \cite{Honore_Restud1993} considered the multi-spell version and established nonparametric identification under strict exogeneity. \cite{woutersen2000essays} contains additional identification results for the piecewise-constant hazard model with sequentially exogenous covariates.
}
\end{example}

\subsection{Efficiency for common parameters and average effects\label{subsec: efficiency}} 

We begin with a brief review of basic concepts in semiparametric efficiency theory; see, for example, \citet[Chap. 25]{van2000asymptotic} and \citet{newey1990semiparametric}. Let $\theta_0$, $g_0$, $\pi_0$, and $\nu_0$ denote, respectively, the common parameter, feedback process, heterogeneity distribution, and initial condition prevailing in the sampled population. Let $\omega=(g,\pi,\nu)\in\Omega$, with true value $\omega_0=(g_0,\pi_0,\nu_0)$, denote the nonparametric components of the model. We consider regular parametric submodels defined by a likelihood function for a single random draw, $\ell(\theta,\omega_{\eta}\,|\,y^T,x^{T})$, where $\eta=(\eta_{g,2},\ldots,\eta_{g,T},\eta_{\pi},\eta_{\nu})\in \mathbb{R}^{T+1}$,  with each component of $\eta$ indexing the corresponding nuisance parameter, i.e., $\omega_{\eta}=((g_{t,\eta_{g,t}})_{t=2}^T,\pi_{\eta_{\pi}},\nu_{\eta_{\nu}})$, such that $\omega_{0}=\omega_{\eta_{0}}$ for some vector $\eta_{0}\in \mathbb{R}^{T+1}$. The likelihood satisfies mean-square differentiability of its square root with respect to $(\theta,\eta)$, with its information matrix nonsingular. The semiparametric variance bound is the supremum of the Cramer Rao bounds for $\theta$ over all such regular parametric submodels.

Let $S^{\theta}$ denote the score for $\theta$, for a submodel evaluated at $\theta=\theta_0$ and $\eta=\eta_0$:
$$S^{\theta}(Y^T,X^T)=\frac{\partial \ln \ell(\theta_0,{\omega}_0\,|\, Y^T,X^T)}{\partial \theta},$$
where we leave the dependence of $S^{\theta}$ on $(\theta_0,\omega_0)$ implicit. Likewise, let $S^{\eta}$ denote the score for $\eta$. The nonparametric tangent set $\mathcal{T}_{\theta_0,\omega_0,K}$ is the mean-square closure of linear combinations $BS^{\eta}$, where $B$ are constant $K\times (T+1)$ matrices with $K$ being the dimension of $\theta$. Let $\mathcal{T}^{\perp}_{\theta_0,\omega_0,K}$ denote the orthocomplement of $\mathcal{T}_{\theta_0,\omega_0,K}$ in the Hilbert space of square-integrable mean-zero functions with inner product $\left<s_1,s_2\right>=\mathbb{E}_{\theta_0,\omega_0}\left[s_1(Y^T,X^T)'s_2(Y^T,X^T)\right]$. By definition this set consists of all $K\times 1$ elements $\phi_{\theta}$ such that $\left<\phi,s\right>=0$ for all $s\in \mathcal{T}_{\theta_0,\omega_0,K}$.

The next theorem provides a characterization of the orthocomplement of the tangent set, $\mathcal{T}^{\perp}_{\theta_0,\omega_0,K}$, which is key to the analysis of efficiency in our context.

    \begin{theorem}{\textsc{(Orthocomplement of tangent set)}}\label{theo_eff}
 		$\mathcal{T}_{\theta_0,\omega_0,K}^{\perp}$ is the linear span of square-integrable, $K$-dimensional moment functions $\phi_{\theta_0}$ that satisfy \eqref{eq_phi_hetero_robust} and \eqref{eq_phi_feedback_robust}.
  	\end{theorem}
\vskip .3cm

An implication of Theorem \ref{theo_eff} is that, if $\phi_{\theta_0} \in \mathcal{T}_{\theta_0,\omega_0,K}^{\perp}$, then it is also an element of the orthogonal complement of the nuisance tangent set associated with \emph{any other} data generating process $(\theta_0, \omega_*)$ with $\omega_* \neq \omega_0$ (i.e., we also have $\phi_{\theta_0} \in \mathcal{T}_{\theta_0,\omega_*,K}^{\perp}$). This follows from the fact that $\mathcal{T}_{\theta_0,\omega_0,K}^{\perp}$ consists of the set of FHR moments characterized in Theorem \ref{theo_charact} earlier; a set which does not vary with $\omega_0$. Indeed, it is precisely this feature of the model which makes feedback and heterogeneity robust estimation possible. Knowledge, whether \emph{a priori} or up to sampling error, of the form of the feedback process and/or heterogeneity distribution is not required for consistent estimation. This is a crucial feature of the class of models we study in this paper.

One subtlety is that, although the elements of $\mathcal{T}_{\theta_0,\omega_0,K}^{\perp}$ do not depend on the precise instance of $\omega_{0}$, the definition of the reference Hilbert space does depend on it. Let $\phi_{\theta_{0}}$ be a valid FHR moment that is absolutely integrable under $(\theta_0,\omega^*)$. Then, while $\phi_{\theta_{0}}(Y^{T},X^{T})$ has zero mean under both $(\theta_0,\omega_0)$ and $(\theta_0,\omega^*)$, in general its variance differs under $\omega_0$ and $\omega^*$. Consequently, the ability to precisely estimate $\theta_{0}$ using a particular $\phi_{\theta}$ generally varies with the population feedback process and heterogeneity distribution, although the validity of $\phi_{\theta}$ as a moment function does not. This connects to the discussion of locally efficient estimation below.

To understand Theorem \ref{theo_eff} it is helpful to consider the implications of restricting $\omega_0$ such that it belongs to a parametric family (indexed by, say, $\eta$). This is the approach taken in, for example, parametric random-effects analysis \citep[e.g.,][]{Chamberlain_ReStud80, chamberlain1985heterogeneity}. In that setting the residualized score, $\widetilde{S}^{\theta}=S^{\theta}-\mathbb{E}\left[S^{\theta}S^{\eta\prime}\right]\times \mathbb{E}\left[S^{\eta}S^{\eta\prime}\right]^{-1}S^{\eta}$, will  generally vary with $\eta$: consistent estimation of $\theta$ requires knowledge of $\eta$ (up to sampling error), and it requires correct specification of the parametric models of feedback and heterogeneity. This is not required in our approach; indeed (elements of) $\omega$ may even be unidentified, while $\theta$ remains $\sqrt{N}$-estimable.  

Theorem \ref{theo_eff} is reminiscent of the situation which arises in average treatment effect (ATE) estimation under unconfoundedness with a \emph{known} propensity score \citep{robins1994estimation,hahn1998role}. In that setting the set of consistent estimating equations for the ATE does not depend on the form of the conditional distribution of the potential outcomes given covariates. Theorem \ref{theo_eff} extends prior work for the case of strictly exogenous regressors showing that $\mathcal{T}^{\perp}_{\theta_0,\omega_0,K}$ is characterized by moment functions that have zero means conditional on \emph{all} covariates, initial conditions, and heterogeneity (see, e.g., \citealp{hahn1994efficiency} for the MPH model, and \citealp{dano2023transition} for dynamic logit models). 

Of course, in many semiparametric estimation problems, consistent estimation of features of the nonparametric model component \emph{is} a requirement for consistent estimation of $\theta$. Examples include the binary choice model with random utility components drawn from an unknown distribution independent of the regressors \citep[see][]{newey1990semiparametric} and ATE estimation under unconfoundedness with an \emph{unknown} propensity score.

\paragraph{Efficient score for $\theta$.} Under regularity conditions,\footnote{Namely that $\ell(\theta,\omega|y^T,x^{T})$ is smooth in a neighborhood of $(\theta_0,\omega_0)$, and that the information matrix is nonsingular (see for example Theorem 3.2 in \citealp{newey1990semiparametric}).} the semiparametric variance bound for $\theta_0$ is equal to the inverse of the variance of the efficient score
\begin{equation}\label{eq_effscore}
\phi_{\theta_0,\omega_0}^{\rm eff}(Y^T,X^T)=\Pi\left(S^{\theta}(Y^T,X^T)\,|\,\mathcal{T}_{\theta_0,\omega_0,K}^{\perp}\right),
\end{equation} 
where $\Pi\left(.\,|\,\mathcal{T}_{\theta_0,\omega_0,K}^{\perp}\right)$ denotes the orthogonal projection onto $\mathcal{T}_{\theta_0,\omega_0,K}^{\perp}$. Note that the projection is well defined since $\mathcal{T}_{\theta_0,\omega_0,K}^{\perp}$ is closed and linear.\footnote{An equivalent and more frequent formulation of the efficient score is $\phi_{\theta_0,\omega_0}^{\rm eff}=S^{\theta}-\Pi\left(S^{\theta}\,|\,\mathcal{T}_{\theta_0,\omega_0,K}\right)$ (e.g., \citealp{hahn1994efficiency}), which is interpreted as the residual from the population regression of $S^{\theta}$ on the nuisance tangent set. The equivalent formulation based on $\mathcal{T}^{\perp}_{\theta_0,\omega_0,K}$ is convenient to work with in our setting; see \citet[][pp. 57-58]{VanDerLaan_Robin_Book2003}.} As a result, the efficient moment restriction for $\theta_0$ is
\begin{equation}\mathbb{E}_{\theta_0,\omega_0}\left[\phi_{\theta_0,\omega_0}^{\rm eff}(Y^T,X^T)\right]=0.\label{eq_eff_theta}\end{equation}The efficient score $\phi_{\theta_0,\omega_0}^{\rm eff}$ in (\ref{eq_effscore}) is defined through a projection onto the orthocomplement $\mathcal{T}_{\theta_0,\omega_0,K}^{\perp}$, which we have fully characterized in Theorem \ref{theo_eff}. 


\paragraph{Efficient score for $\mu$.} We now turn to an analysis of efficiency for average effects. Suppose there exists a moment function $\varphi_{\theta_0}$ that identifies an $L\times 1$ average effect of interest $\mu(\theta_0,\omega_0)$ given by\footnote{Standard implicit smoothness conditions are required, namely that, for a regular parametric submodel, $\sup_{(\theta,\eta) \in \mathcal{N}} \mathbb{E}_{\theta,\omega_{\eta}}\left[\norm{\varphi_{\theta}(Y^T,X^T)}^2\right]<\infty$ in a neighborhood $\mathcal{N}$ of $(\theta_0,\eta_0)$, such that a generalized information equality holds (see \citealp{brown1998efficient}).}
\begin{equation}\mu(\theta_0,\omega_0)=\mathbb{E}_{\theta_0,\omega_0}\left[h_{\theta_0}(Y^T,X^T,A)\right]=\mathbb{E}_{\theta_0,\omega_0}\left[\varphi_{\theta_0}(Y^T,X^T)\right],\label{eq_varphi}
\end{equation}
and suppose that $\theta_0$ is identified from \eqref{eq_eff_theta}. Let
$${\varphi}^{\rm eff}_{\theta_0,\omega_0}(Y^T,X^T)=\varphi_{\theta_0}(Y^T,X^T)-\Pi(\varphi_{\theta_0}(Y^T,X^T)\,|\,\mathcal{T}_{\theta_0,\omega_0,L}^{\perp}),$$
where the orthocomplement of the tangent set, $\mathcal{T}_{\theta_0,\omega_0,L}^{\perp}$, is given by Theorem \ref{theo_eff}, except for the fact that the relevant dimension is $L$ instead of $K$. 

Theorem 1 in \cite{brown1998efficient} shows that the joint efficient moment restrictions for $\theta_0$ and $\mu_0=\mu(\theta_0,\omega_0)$ are given by \eqref{eq_eff_theta} and
\begin{equation}
\mathbb{E}_{\theta_0,\omega_0}\left[{\varphi}^{\rm eff}_{\theta_0,\omega_0}(Y^T,X^T)-\mu_0\right]=0.\label{eq_eff_mu}
\end{equation}
The semiparametric variance bound for $\mu_0$ equals the lower $L\times L$ block of the asymptotic variance of the joint GMM estimator $(\widehat{\theta},\widehat{\mu})$ based on \eqref{eq_eff_theta} and \eqref{eq_eff_mu}.

The construction of $\varphi^{\rm eff}_{\theta_0,\omega_0}$ relies on a function  $\varphi_{\theta_0}$ that satisfies \eqref{eq_varphi}. In some models one can find such a function (a Riesz representer), which does not depend on the nonparametric component $\omega_0$. This can be done by exploiting Theorem \ref{theo_charact_aveff}. See, for example, the discussion of the average structural hazard function in the MPH earlier. Moreover, $\varphi^{\rm eff}_{\theta_0,\omega_0}$ is not affected by the particular choice of $\varphi_{\theta_0}$.\footnote{This follows from noting that $\varphi_{\theta_0}(Y^T,X^T)-\Pi(\varphi_{\theta_0}(Y^T,X^T)\,|\,\mathcal{T}_{\theta_0,\omega_0,L}^{\perp})=\mu_{0}+\Pi(\varphi_{\theta_0}(Y^T,X^T)\,|\,\mathcal{T}_{\theta_0,\omega_0,L})$, and that, if $\varphi_{\theta_0}^{1}$ and $\varphi_{\theta_0}^{2}$ both satisfy \eqref{eq_varphi}, then $\Pi(\varphi_{\theta_0}^{1}(Y^T,X^T)\,|\,\mathcal{T}_{\theta_0,\omega_0,L})=\Pi(\varphi_{\theta_0}^{2}(Y^T,X^T)\,|\,\mathcal{T}_{\theta_0,\omega_0,L})$ as we show in Supplemental Appendix \ref{app_B}.} 


\subsection{Moment restrictions based on working models \label{subseq:workingmodels}} 
Although the characterization of feasible moment conditions provided by Theorem \ref{theo_eff} is invariant to the specific instance of $\omega_0$ indexing the sampled population, the projection \eqref{eq_effscore} generally \emph{does} vary with $\omega_0$. Consequently, although knowledge of $\omega_0$ is not required for consistent estimation, it is generally valuable for improving asymptotic precision. Moreover, constructing an estimator which attains the semiparametric efficiency bound for all possible feedback processes, $g_0$, heterogeneity distributions, $\pi_0$, and initial conditions, $\nu_0$ (i.e., for all $\omega_0 \in \Omega$) generally requires nonparametrically estimating features of these model components. This may be practically difficult, or even impossible, in short panels as considered here.

An alternative approach involves constructing \emph{locally efficient} estimators \citep[e.g.,][]{newey1990semiparametric, Graham_Pinto_Egel_ReStud12}. Let $\widetilde{\omega}=(\widetilde{g},\widetilde{\pi},\widetilde{\nu})\in\Omega$ denote candidate ``working models" for the feedback process, heterogeneity distribution, and initial condition. We show how to construct method-of-moments estimators that (i) attain the bound for $\theta_0$ (or $\mu_0$) when these working models ``happen to characterize the sampled population" (i.e., $\widetilde{\omega}=\omega_0$, but this is not part of the prior restriction) and (ii) remain $\sqrt{N}$-consistent irrespective of the true $\omega_0$ characterizing the sampled population (i.e., when $\widetilde{\omega}\neq\omega_0$). A key property in our setting is that consistency holds for arbitrary $\widetilde{\omega}$ (i.e., our working models may be misspecified), only subject to regularity conditions.

Given working models $\widetilde{\omega}$, let 
$$\widetilde{S}^{\theta}(Y^T,X^T)=\frac{\partial \ln \ell(\theta_0,\widetilde{\omega}\,|\, Y^T,X^T)}{\partial \theta}$$
denote the score for $\theta_0$. Next define the counterpart, \emph{under the working models}, to the efficient score $\phi_{\theta_0,\omega_0}^{\rm eff}$ for $\theta_0$ as
$$\widetilde{\phi}_{\theta_0,\widetilde{\omega}}^{\rm eff}(Y^T,X^T)=\widetilde{\Pi}\left(\widetilde{S}^{\theta}(Y^T,X^T)\,|\, \mathcal{T}_{\theta_0,\widetilde{\omega},K}^{\perp}\right),$$
where $\widetilde{\Pi}$ denotes the projection operator under the working models, that is, 
\begin{align} \label{def_loceff_score}
    \widetilde{\phi}_{\theta_0,\widetilde{\omega}}^{\rm eff}=\arg\underset{\phi\in\mathcal{T}_{\theta_{0},\widetilde{\omega},K}^{\perp}}{\min} \mathbb{E}_{\theta_0,\widetilde{\omega}}\left[\left\| \widetilde{S}^{\theta}(Y^T,X^T)-\phi\left(Y^T,X^T\right)\right\|^2\right].
\end{align}
We next proceed similarly for $\mu(\theta,\omega)$: the counterpart to the efficient score $\varphi_{\theta_0,\omega_0}^{\rm eff}$ is
\begin{align} \label{def_loceff_score_mu}
    \widetilde{\varphi}^{\rm eff}_{\theta_0,\widetilde{\omega}}(Y^T,X^T)=\varphi_{\theta_0}(Y^T,X^T)-\widetilde{\Pi}(\varphi_{\theta_0}(Y^T,X^T)\,|\, \mathcal{T}_{\theta_0,\widetilde{\omega},L}^{\perp}).
\end{align}

\noindent Finally, consider the following moment restrictions for $\theta_0$ and $\mu_0=\mu(\theta_0,\omega_0)$:
\begin{align}\mathbb{E}_{\theta_0,\omega_0}\left[\widetilde{\phi}^{\rm eff}_{\theta_0,\widetilde{\omega}}(Y^T,X^T)\right]=0,\label{def_loceff_score_est1}\\
 \mathbb{E}_{\theta_0,\omega_0}\left[\widetilde{\varphi}^{\rm eff}_{\theta_0,\widetilde{\omega}}(Y^T,X^T)-\mu_0\right]=0,\label{def_loceff_score_est2}
\end{align}
where we note that the expectations are taken under the true DGP $(\theta_0,\omega_0)$.

We can now state the following result.

\begin{theorem}\label{theo_working_model}
For any working models $\widetilde{\omega}\in\Omega$ such that $\widetilde{\phi}^{\rm eff}_{\theta_0,\widetilde{\omega}}$ and $\widetilde{\varphi}^{\rm eff}_{\theta_0,\widetilde{\omega}}$ are absolutely integrable under DGP $(\theta_0,\omega_0)$, the moment restrictions \eqref{def_loceff_score_est1} and \eqref{def_loceff_score_est2} hold. Moreover, if $\widetilde{\omega}=\omega_0$, then \eqref{def_loceff_score_est1} and \eqref{def_loceff_score_est2} coincide with the efficient moment restrictions for $\theta_0$ and $\mu_{0}$.
\end{theorem}

Theorem \ref{theo_working_model} articulates a \emph{locally} efficient approach to estimation. The moment functions $\widetilde{\phi}^{\rm eff}_{\theta_0,\widetilde{\omega}}$ and $\widetilde{\varphi}^{\rm eff}_{\theta_0,\widetilde{\omega}}$ have the FHR property, irrespective of whether the working models used to derive them actually characterize the sampled population. However, if $\widetilde{\omega}=\omega_0$ ``happens to hold" in the sampled population, then \eqref{def_loceff_score_est1} and \eqref{def_loceff_score_est2} coincide with the efficient moment restrictions for $\theta_0$ and $\mu_{0}$ (when $\widetilde{\omega}=\omega_0$ is \emph{not} part of the prior restriction used to calculate the efficiency bound).

The functions $\widetilde{\phi}^{\rm eff}_{\theta_0,\widetilde{\omega}}$ and $\widetilde{\varphi}^{\rm eff}_{\theta_0,\widetilde{\omega}}$ are FHR because calculation \eqref{def_loceff_score} returns an element in the orthogonal complement of the nuisance tangent set by construction. Calculation \eqref{def_loceff_score} provides a principled way to select a particular FHR moment, one that is optimal if the working model happens to hold in the sampled population. Heuristically, the method-of-moments estimator based upon \eqref{def_loceff_score_est1} and \eqref{def_loceff_score_est2} will be more precise when the working models are ``approximately true", but this -- to repeat -- is not required for consistency.

In practice $\widetilde{\omega}$ may be a fixed set of working models chosen by the researcher. Alternatively the researcher may posit that these models belong to parametric families $\omega_{\eta}$ indexed by an unknown finite-dimensional parameter $\eta$. These models may be misspecified, in the sense that there may not exist any $\eta_{0}$ such that $\omega_{\eta_0}=\omega_0$. Nevertheless the moments \eqref{def_loceff_score_est1} and \eqref{def_loceff_score_est2} will be valid for any $\eta$. There are different strategies for picking a particular $\eta$. First, the researcher may choose a particular (non-stochastic) $\eta$ via introspection. Second, she might maximize the likelihood under the working models with respect to $\theta$ and $\eta$. While the resulting estimate of $\theta$ will generally be inconsistent, the corresponding estimate of $\eta$ can be used to define the working models $\widetilde{\omega}$ under which $\widetilde{\phi}^{\rm eff}_{\theta_0,\widetilde{\omega}}$ and $\widetilde{\varphi}^{\rm eff}_{\theta_0,\widetilde{\omega}}$ are calculated.

A third approach is to select $\eta$ by maximizing an empirical counterpart to the information for $\theta_0$, as in \citet{lindsay1985using}. Let $\widehat{\eta}$ be such an estimator of $\eta$, and let $\eta^*$ be its large-$N$ probability limit. If one defines $\widetilde{\omega}=\omega_{{\eta}^*}$, then Theorem \ref{theo_working_model} implies that \eqref{def_loceff_score_est1} and \eqref{def_loceff_score_est2} are satisfied at true parameter values $\theta_0$ and $\mu_{0}$ (under absolute integrability of the functions). This suggests that the GMM estimators $\widehat{\theta}$ and $\widehat{\mu}$ based on \eqref{def_loceff_score_est1} and \eqref{def_loceff_score_est2} that use $\omega_{\widehat\eta}$ in lieu of $\widetilde{\omega}$ are consistent and asymptotically normal under suitable conditions. We leave details about efficient estimation to future work.

Lastly, it is important to stress that our approach based on working models is fundamentally different from parametric random-effects maximum likelihood estimation. Indeed, plugging in misspecified parametric models $\omega_{\eta}$ into the likelihood function (\ref{eq: feedback_complete_data_likelihood}), and maximizing that likelihood with respect to $\theta$ and $\eta$, generally results in an inconsistent estimator of $\theta_0$. In contrast, the moment restrictions \eqref{def_loceff_score_est1} and \eqref{def_loceff_score_est2} remain valid even when the working models are (globally) misspecified.

We close this section, by recalling that, in Section \ref{sec: moments} we presented several examples where restricting ourselves to working with FHR moments alone results in a complete identification failure (i.e., the ``throwing away" of valuable information). In contrast, an implication of the analysis in this section is that when point-identification based on FHR moments holds, working only with FHR moments involves no efficiency penalty (subject to standard regularity conditions).

\section{Efficiency in the multi-spell MPH with unrestricted feedback}\label{sec: mph_deep_dive}
In this section we specialize the general results presented above in order to analyze semiparametric efficiency in the MPH model. We focus on the $T=2$ special case in what follows. 

\subsection{Efficiency bounds analysis}

\textbf{Efficient score for $\mathbf{\theta}$}. Using Lemma~\ref{lem_moment_charac_MPH} and Theorem~\ref{theo_eff}, we show in Supplemental Appendix \ref{app: MPH_derivations} that, for \( T = 2 \), the efficient score for $\theta$ in the MPH model with feedback equals
\begin{align} \label{eq: efficient_score_theta_mph_T=2_feedback}
    \phi_{\theta_0,\omega_0}^{\rm eff}(Y^2,X^2)
    = \mathbb{E}\left[S^{\theta}(Y^2, X^2) \mid Y_0, \widetilde{P}_{\theta_0,1}, \overline{P}_{\theta_0}, X_1\right]
    - \mathbb{E}\left[S^{\theta}(Y^2, X^2) \mid Y_0, \overline{P}_{\theta_0}, X_1\right].
\end{align}
With some additional algebra we show that the efficient score for the \( \beta \) subvector is
\begin{align}
    \phi_{\theta_0,\omega_0}^{\rm eff,\beta}(Y^2,X^2) ={}&
    -X_1\left(\widetilde{P}_{\theta_0,1} - \tfrac{1}{2}\right) \overline{P}_{\theta_0} \mathbb{E}\left[e^A \mid Y_0, \overline{P}_{\theta_0}, X_1\right] \notag \\
    &+ \mathbb{E}\left[X_2 \left(1 - (1 - \widetilde{P}_{\theta_0,1}) \overline{P}_{\theta_0} e^A\right) \mid Y_0, \widetilde{P}_{\theta_0,1}, \overline{P}_{\theta_0}, X_1\right] \notag \\
    &- \mathbb{E}\left[X_2 \left(1 - (1 - \widetilde{P}_{\theta_0,1}) \overline{P}_{\theta_0} e^A\right) \mid Y_0, \overline{P}_{\theta_0}, X_1\right].
    \label{eq: efficient_score_beta_mph_T=2_feedback}
\end{align}
The first term in \eqref{eq: efficient_score_beta_mph_T=2_feedback} does not involve the feedback process and resembles the efficient score for $\beta$ under strict exogeneity originally derived by \cite{hahn1994efficiency}:
\begin{align}
    \phi_{\theta_0,\omega_0}^{\rm eff,SE,\beta}(Y^2,X^2)
    = \left(X_2 - X_1\right)\left(\widetilde{P}_{\theta_0,1} - \tfrac{1}{2}\right) \overline{P}_{\theta_0} \mathbb{E}\left[e^A \mid Y_0, \overline{P}_{\theta_0}, X_1, X_2\right].
    \label{eq: efficient_score_beta_mph_T=2_strictexog}
\end{align}
The second and third terms of \eqref{eq: efficient_score_beta_mph_T=2_feedback}, in contrast, are specific to the feedback case, involving averages over the second-period covariate $X_{2}$. More generally, the efficient score for $\theta$ under strict exogeneity equals:
\begin{align}
    \phi_{\theta_0,\omega_0}^{\rm eff,SE}(Y^2,X^2)
    &= \mathbb{E}\left[S^{\theta}(Y^2, X^2) \mid Y_0, \widetilde{P}_{\theta_0,1}, \overline{P}_{\theta_0}, X_1, X_2\right] - \mathbb{E}\left[S^{\theta}(Y^2, X^2) \mid Y_0, \overline{P}_{\theta_0}, X_1, X_2\right].
    \label{eq: efficient_score_theta_mph_T=2_strictexog}
\end{align}
In the presence of feedback and latent heterogeneity, $X_2$ is an endogenous variable and cannot be conditioned on. The precise forms for the efficient scores for $\alpha$ and $\gamma$ are presented in Supplemental Appendix \ref{app: MPH_derivations}.

From (\ref{eq: efficient_score_beta_mph_T=2_feedback}) we see that the efficient score depends on both heterogeneity and feedback, due to the presence of $A$ and $X_2$ in the expectations. Efficient estimation would require estimating the efficient score. Since this may be challenging in practice, in this section we rely on working models for feedback and heterogeneity and assess the performance of the implied FHR moments through numerical calculations.\\

\noindent \textbf{Efficient moment functions for the ASH}. The average structural hazard (ASH) defined in \eqref{WMPH_ashf} is identified by the FHR moment function in \eqref{eq: ashf_identified}, namely $\varphi_{\theta_0}(Y^T,X^T)=\lambda_{\alpha_{0}}(y_{t})e^{x_{t}'\beta_{0}+\gamma_{0} y_{t-1}}\frac{T-1}{\overline{P}_{\theta_{0}}}$. Using the characterization of projections onto the orthocomplement of the tangent set in the MPH model given in Supplemental Appendix \ref{app: MPH_derivations}, one can readily show that $\Pi(\varphi_{\theta_0}(Y^T,X^T)\,|\,\mathcal{T}_{\theta_0,\omega_0,L}^{\perp})=0$. It then follows from the discussion in Section~\ref{subsec: efficiency} that the efficient moment function for the ASH is ${\varphi}^{\rm eff}_{\theta_0,\omega_0}(Y^T,X^T)=\varphi_{\theta_0}(Y^T,X^T)=\lambda_{\alpha_{0}}(y_{t})e^{x_{t}'\beta_{0}+\gamma_{0} y_{t-1}}\frac{T-1}{\overline{P}_{\theta_{0}}}$. 

Interestingly, while the mean of $\varphi_{\theta_0}(Y^T,X^T)$ equals the ASH at $\theta=\theta_0$ for any $T \geq 2$, boundedness of its variance requires $T \geq 3$.\footnote{$\left.\frac{1}{\overline{P}_{\theta_{0}}}\right|X_1,Y_0,A$ is an $\mathrm{InverseGamma}\left(T,\exp(A)\right)$ random variable, with finite mean when $T \geq2 $, but finite variance only when $T \geq 3$.} An implication is that the information bound for the ASH is zero when $T=2$, but, provided $\mathbb{E}[e^{2A}]<\infty$, generally positive for $T \geq 3$. However, with $\widehat{\theta}=(\widehat{\alpha},\widehat{\beta}',\widehat{\gamma})'$ a consistent estimate of the common parameter, the corresponding analog ASH estimator
\begin{align}
    \widehat{\overline{\lambda}}(y_{t}|x_t,y_{t-1};\widehat{\theta})=\lambda_{\widehat{\alpha}}(y_{t})e^{x_t'\widehat{\beta}+\widehat{\gamma} y_{t-1}}\frac{1}{N}\sum_{i=1}^N\frac{T-1}{\overline{P}_{i\widehat{\theta}}},\label{eq_ASH_est}
\end{align}
is nevertheless consistent when $T = 2$. When $T\geq 3$, GMM estimation with ${\phi}^{\rm eff}_{\theta_0,\omega_0}(Y^T,X^T)$ and ${\varphi}^{\rm eff}_{\theta_0,\omega_0}(Y^T,X^T)$ is $\sqrt{N}$ consistent, indeed, semiparametrically efficient. See Supplemental Appendix \ref{app: MPH_derivations} for details.\footnote{Note that (\ref{eq_ASH_est}) coincides with the method of moments estimator of \cite{brown1998efficient} when the condition \( \varphi_{\theta_0}(Y^T, X^T) - \mu_0 \in \mathcal{T}_{\theta_0,\omega_0,L} \) holds, where \( \mu_0 =\mu(\theta_0,\omega_0)\) denotes the average effect of interest and \( \varphi_{\theta_0} \) is an identifying moment function. This condition is satisfied in the case of the ASH for the choice $\varphi_{\theta_0}(Y^T,X^T)=\lambda_{\alpha_{0}}(y_{t})e^{x_{t}'\beta_{0}+\gamma_{0} y_{t-1}}\frac{T-1}{\overline{P}_{\theta_{0}}}$.}\\

\subsection{Numerical illustrations} \label{subsec: numerical_experiment}
In this subsection we summarize our findings from two numerical experiments designed to (i) illustrate the efficiency loss associated with accommodating feedback, and (ii) assess the performance of locally efficient estimators based on working models. Our context is the MPH model. In both experiments we impose a Weibull baseline hazard, set $T=2$, and fix the common parameter  at $\theta_{0}=(\alpha_{0},\beta_{0},\gamma_{0})=(\frac{3}{4},\frac{3}{4}\ln 2, -\frac{1}{10})$.
Our experiments are meant to numerically approximate the asymptotic precision of various methods of estimation, not to assess the accuracy of such approximations in small samples. Details on implementation can be found in Supplemental Appendix \ref{app: numericalexperiment_details}.

The initial duration is drawn from an exponential distribution: $Y_{0}\sim \text{Exponential}(\frac{3}{2})$, and the first-period covariate is a randomized binary treatment: $X_{1}\sim \text{Bernoulli}(\frac{1}{2})$. The heterogeneity distribution equals $V= e^{A}\sim \text{Gamma}(5,5)$, independent of $Y_0,X_1$. The second-period covariate, $X_2$, is a Bernoulli random variable with success probability $p(Y_0,Y_1,X_1,V)$, specified differently across the two experiments to reflect alternative assumptions about the DGP.\footnote{We have computed variance bounds in DGPs with different values for these parameters and found qualitatively similar results (not reported here for brevity).} In Experiment (A), we set $p(Y_0,Y_1,X_1,V)=1-\exp(-(Y_{0}+X_{1})V)$, which produces a DGP with strictly exogenous covariates but correlated unobserved heterogeneity. In Experiment (B), we instead let $p(Y_0,Y_1,X_1,V)=1-\exp(-(Y_{0}+X_{1}+Y_{1})V)$, which introduces feedback from past outcomes to future covariates, while continuing to include correlated heterogeneity. 


We compare the asymptotic standard errors of two GMM estimators: the first uses the efficient score under strict exogeneity \eqref{eq: efficient_score_theta_mph_T=2_strictexog}, while the second uses the efficient score under feedback \eqref{eq: efficient_score_theta_mph_T=2_feedback}.

The close connection between our FHR moment characterization (Theorem \ref{theo_charact}) and the relevant semiparametric efficiency bound theory (Theorem \ref{theo_eff}), raises interesting practical questions regarding estimation. While many possible FHR moments are available in the MPH setting, the precision with which they recover $\theta$ varies with the population values of the feedback process and heterogeneity distribution. As a principled, yet practical, approach, we explore the properties of a locally efficient estimator based upon simple working models for $\widetilde{\omega}=(\widetilde{g},\widetilde{\pi})$ (a model for $\nu$ is not needed in this case).

Our chosen working models are deliberately rudimentary, intended to illustrate how a researcher might build parsimonious working models while retaining favorable efficiency properties in more realistic settings. Specifically, in contrast to what prevails in the sampled population, the working model for the feedback process maintains that $X_{2}\sim\mathrm{Bernoulli}(p)$, for a constant $p$. Observe that here the working model for the feedback process involves no feedback. For the heterogeneity distribution, we calculate the locally efficient score under a diffuse Gamma prior of $\widetilde{\pi}(v) = \frac{1}{v} \mathds{1}\{v > 0\}$, independent of initial conditions.

Putting all these pieces together yields the following locally efficient score for $\beta$:
\begin{align*}
\widetilde{\phi}_{\theta_0,\widetilde{\omega}}^{\rm eff,\beta}(Y^2,X^2) =
2\left(\mathbb{E}_{\widetilde{g}}\left[X_2\right] - X_1\right)\left(\widetilde{P}_{\theta_0,1} - \tfrac{1}{2}\right)=2\left(p - X_1\right)\left(\widetilde{P}_{\theta_0,1} - \tfrac{1}{2}\right),
\end{align*}
which is simpler than its optimal counterpart \eqref{eq: efficient_score_beta_mph_T=2_feedback} yet, of course, still feedback and heterogeneity robust. Additional details, along with the full expression for the score $\widetilde{\phi}_{\theta_0,\widetilde{\omega}}^{\rm eff}(Y_0,\widetilde{P}_{\theta_0,1}, \overline{P}_{\theta_0}, X_1)$, are provided in Supplemental Appendix \ref{app: numericalexperiment_details}. 
\\\indent As a final point of comparison, we also report the limiting standard errors of a just-identified GMM estimator that employs the moment function:
\begin{equation}
    \label{eq_simple_moment}
    \phi_{\theta_{0}}(Y^2,X^2)=\left(\begin{array}{c}
     X_{1}\left(\ln(1-\widetilde{P}_{\theta_{0},1})-\ln(\widetilde{P}_{\theta_{0},1})\right)\\
    Y_{0}\left(\ln(1-\widetilde{P}_{\theta_{0},1})-\ln(\widetilde{P}_{\theta_{0},1})\right) \\
    2+\ln(1-\widetilde{P}_{\theta_{0},1})+\ln(\widetilde{P}_{\theta_{0},1})
    \end{array}\right).
\end{equation}
This is the moment function used to show point-identification of the Weibull MPH model (Example \ref{ex: mph_Weibull_identification}) earlier. While this set of moment conditions lacks an overt efficiency justification, it reflects a common strategy of choosing a small set of ``simple" moments for estimation purposes. We include it primarily to benchmark the efficiency gains provided by the locally efficient approach based on working models. 

Table \ref{tab1_montecarlo_SE} reports the asymptotic standard errors for each estimate of $\theta_0$ in Experiment (A). We make several observations. First, comparing the first and second rows reveals that accommodating feedback -- when it is, in fact, absent from the DGP -- results in some efficiency loss for the slope coefficient $\beta$, with a 29\% increase in standard error. Strict exogeneity is a strong assumption and imposing it, when it is valid to do so, improves asymptotic precision.

Although allowing for feedback degrades the precision with which we can learn $\beta$, this is  not really the case for $\alpha$ (the Weibull baseline hazard parameter) and $\gamma$ (the lagged duration dependence parameter) in design (A). This finding is consistent with the structure of the efficient scores: those for $\alpha$ and $\gamma$ are very similar under both strict exogeneity and feedback. By contrast, the efficient score for $\beta$ differs markedly across the two settings.

A second observation, inspecting the third row of the table, is that the precision loss associated with using the locally efficient estimator based on the working models $\widetilde{\omega}$ is moderate. Recall that our working models do not characterize the sampled population in design (A). For $\beta$ and $\gamma$, respectively, we observe a 28\% and 25\% increase in standard error when comparing rows 2 and 3. The efficiency losses are concentrated on the coefficients for the predetermined covariate and the lagged dependent variable, with only minimal deterioration for the parameter of the baseline hazard $\alpha$. 

Finally, the approach based on working models leads to large improvements relative to using the ``simple" moment functions \eqref{eq_simple_moment}. As seen in row 4, the standard errors associated with the simple GMM estimator are substantially larger for all three parameters. For example, the standard error for $\gamma$ is nearly seven times higher than that obtained using the estimator based on the working models $\widetilde{\omega}$ (compare rows 3 and 4, column 3).

\begin{table}[!htb]
\caption{Asymptotic standard errors relative to the semiparametric efficiency bound with strict exogeneity in Experiment (A)} \label{tab1_montecarlo_SE}
\centerline{
\scalebox{1}{
\begin{tabular}{lccccccccccccc} \toprule  \toprule 
  & $\alpha$ & $\beta$ & $\gamma$  \\ 
Eff.score SE &1.0&1.0&1.0\\ 
Eff.score FB &1.001&1.290&1.004\\ 
Locally Eff.score FB &1.034&1.655&1.254\\ 
Simple moments &2.708&2.978&8.678\\ 
\bottomrule \bottomrule 
\end{tabular}
}}
\par
{\footnotesize\textsc{Notes: }\textit{SE denotes strict exogeneity, FB denotes feedback. The DGP satisfies strict exogeneity.}}
\end{table}

\indent In Experiment (B), we repeat our analysis, but for a population where heterogeneous feedback is, in fact, present. Accordingly, Table \ref{tab1_montecarlo_FB} compares the asymptotic standard errors of the locally efficient estimator based on $\widetilde{\phi}_{\theta_0,\widetilde{\omega}}^{\rm eff,\beta}$ to that of the ``simple'' GMM estimator based on (\ref{eq_simple_moment}) under the second DGP described above (an estimate based on the efficient score under strict exogeneity would not be consistent in this design). As a benchmark, we use the globally semiparametrically efficient estimator based upon the true efficient score that uses \eqref{eq: efficient_score_theta_mph_T=2_feedback}. The key takeaways mirror those of Table \ref{tab1_montecarlo_SE}. First, the efficiency loss from relying on locally efficient scores is modest. Second, and perhaps more importantly, locally efficient estimators significantly outperform the alternative of using a set of ``simple" moments, reaffirming the practical advantages of approaches based on working models.
\begin{table}[!htb]
\caption{Asymptotic standard errors relative to the semiparametric efficiency bound with feedback in Experiment (B)} \label{tab1_montecarlo_FB}
\centerline{
\scalebox{1}{
\begin{tabular}{lccc} \toprule  \toprule 
  & $\alpha$ & $\beta$ & $\gamma$ \\ 
Eff.score FB &1.0&1.0&1.0\\ 
Locally Eff.score FB &1.030&1.186&1.302\\ 
Simple moments &3.379&3.259&9.866\\ 
\bottomrule \bottomrule 
\end{tabular}

}}
\vspace{0.1cm}
\textit{{\footnotesize \textsc{Notes: }FB denotes feedback. The DGP does not satisfy strict exogeneity.}}
\end{table}


In Supplemental Appendix \ref{app: numericalexperiment_details}, we report an auxiliary numerical exercise aiming to illustrate the efficiency cost of leaving the feedback process unrestricted when its structure is actually known.  We specifically consider a DGP with homogeneous feedback and find that the precision penalty incurred by not specifying the feedback process is relatively mild. As in Table \ref{tab1_montecarlo_SE} with strict exogeneity, the losses are concentrated on the coefficient of the predetermined regressor, $\beta$, while the asymptotic standard errors for $\alpha$ and $\gamma$ are virtually unchanged. The locally efficient estimator also performs well in this case.

\section{FHR moment restrictions in other models}\label{sec: twilight_zone_material}

Our characterizations can be used to find FHR moment functions for many models. We have already analyzed the MPH model in detail. In this section we provide additional analytical examples and discuss how to obtain moment functions more generally. Given that the mathematical structure for model parameters and average effects is similar, we focus our discussion on $\theta$.

To derive new moment conditions, researchers can apply the two-step procedure underlying Corollary \ref{coro_charact}. In the first step, we construct a function $\psi_{\theta}=\sum_{t=1}^{T-1}\psi_{\theta,t}$, for $\psi_{\theta,t}$ such that 
\begin{equation}\mathbb{E}\left[\psi_{\theta,t}(Y^t,X^{t},A)\,|\,Y^{t-1},X^t,A\right]=0,\quad t=1,\ldots,T-1.\label{eq_psi_zeromean}\end{equation}
In the second step, in the final time period, we invert the linear integral equation 
\begin{align}
    	\int \phi_{\theta}(y^T,x^T) f_{\theta}(\left.y_T\right|y_{T-1},x_{T},a)\mathrm{d}y_T=\psi_{\theta}(y^{T-1},x^{T-1},a),\label{eq_int}
\end{align}
to get a FHR moment function $\phi_{\theta}(y^T,x^T)$. Naturally, this last task requires the function $\psi_{\theta}(y^{T-1},x^{T-1},a)$ to lie in the range of the integral operator induced by the parametric model. In many examples of interest, equation (\ref{eq_int}) can actually be inverted in closed form, yielding explicit expressions for functions $\phi_{\theta}$. As an initial example, we first consider a Poisson regression model.

\begin{example} {(\textsc{Poisson Model})}\label{ex: count_intro}
    Let $\left\{ Y_{t}\right\} _{t=1}^{T}$ denote a sequence of counts, for example the number of patents awarded to a firm in a year, as in \cite{Blundell_Griffith_Windmeijer_JOE2002}, and $X_{t}$ a vector of time-varying regressors. For $t=1,\ldots,T$ we have
   \begin{equation} \label{eq: poisson_panel_data_model}
         \left.Y_{t}\right|Y^{t-1},X^{t},A\sim\mathrm{Poisson}\left(\exp\left(\gamma Y_{t-1}+X_{t}'\beta+A\right)\right), 
   \end{equation} 
   with $Y_{0}$ an initial count. \cite{Chamberlain_JBES92} and \cite{wooldridge1997multiplicative} proposed GMM estimators for $\theta=\left(\beta',\gamma \right)'$ in this model. \cite{Windmeijer_EPD2008} reviews extant results.
   
\noindent For simplicity, consider the case $T=2$ and denote $Z_{t}=(X_t',Y_{t-1})'$ and $\theta=(\beta',\gamma)'$. Following the logic laid out above, we start out by picking a moment function $\psi_{\theta}$ such that
$\mathbb{E}\left[\psi_{\theta}(Y_{0},Y_{1},X_{1},A)\,|\,Y_{0},X_{1},A\right]=0$. An example of such a function is provided by the score for period $t=1$ in the likelihood which conditions on $A$, $\psi_{\theta}(y_{0},y_{1},x_{1},a)=z_1\left(y_1-\exp(z_1'\theta+a)\right)$.\\
Next, the challenge is to solve for $\phi_{\theta}$ in (\ref{eq_int}); here corresponding to finding a solution to
\begin{align*}
    \sum_{y_2=0}^{\infty} \phi_{\theta}(y_{0},y_{1},y_2,x_{1},x_{2})\exp(-\exp(z_2'\theta+a))\frac{\exp(z_2'\theta+a)^{y_2}}{y_2!}=\psi_{\theta}(y_{0},y_{1},x_{1},a).
\end{align*}
After multiplying by $\exp(\exp(z_2'\theta+a))$ and letting $v=\exp(a)$, this is equivalent to
\begin{align*}
   \sum_{y_2=0}^{\infty} \phi_{\theta}(y_{0},y_{1},y_2,x_{1},x_{2})\exp(y_2z_2'\theta) \frac{v^{y_2}}{y_2!}=\psi_{\theta}(y_{0},y_{1},x_{1},\ln v)e^{v\exp(z_2'\theta)}.
\end{align*}
This formulation reveals that, for a solution to exist, we require that $v\mapsto \psi_{\theta}(y_{0},y_{1},x_{1},\ln v)e^{v\exp(z_2'\theta)}$ admits a Taylor series at $v=0$, with coefficients given by $\phi_{\theta}(y_{0},y_{1},y_2,x_{1},x_{2})\exp(y_2z_2'\theta)$. Appealing to the uniqueness of the Taylor series, we then infer that:
\begin{align}\label{momentfunc_poisson}
    \phi_{\theta}(y_{0},y_{1},y_{2},x_{1},x_{2})=\frac{\partial^{y_2}}{\partial v^{y_2}}\bigg|_{v=0}\, \left[{\psi_{\theta}}(y_{0},y_{1},x_{1},\ln v)e^{v\exp(z_2'\theta)}\right]\exp(-y_2z_2'\theta).
\end{align}
By Corollary \ref{coro_charact}, all FHR moment functions for $\theta$ take the form (\ref{momentfunc_poisson}), for some appropriate mean-zero function $\psi_{\theta}$. 
For instance, in the case where $\psi_{\theta}$ is the score in period $t=1$, we obtain $\phi_{\theta}(y_{0},y_{1},y_{2},x_{1},x_{2})
        =z_1\left(y_1-y_2e^{(z_1-z_2)'\theta}\right)$, which is proportional to the moment function of \cite{Chamberlain_JBES92} and \cite{wooldridge1997multiplicative}. However, using \eqref{momentfunc_poisson} provides many additional valid moment restrictions on $\theta$. For example, by the second-moment properties of the Poisson distribution,
\begin{equation}
    \psi_{\theta}(y_{0},y_{1},x_{1},a)=\left[y_{1}\left(y_{1}-1\right)-\exp\left(z_{1}'\theta+a\right)^{2}\right]\cdot m\left(z_1\right)    
\end{equation}
is analytic in $v=\exp(a)$ and also satisfies \eqref{eq_psi_zeromean}, from which we get the moment
\begin{equation}             \phi_{\theta}(y_{0},y_{1},y_{2},x_{1},x_{2})=\left[y_{1}\left(y_{1}-1\right)-\frac{y_{2}\left(y_{2}-1\right)\exp\left(z_{1}'\theta\right)^{2}}{\exp\left(z_{2}'\theta\right)^{2}}\right]\cdot m\left(z_{1}\right).\label{eq: new_poisson_moments}  
\end{equation}
Additional FHR moment functions, based upon higher order moments of the Poisson distribution, are straightforward to construct.\footnote{Moreover, by Theorem \ref{theo_eff}, the functions $\phi_{\theta}$ in \eqref{momentfunc_poisson} span the orthocomplement of the tangent set. This suggests that one could compute the efficient score for $\theta$ by projecting the $\theta$-score on that set of functions, although we leave the derivation of the precise form of the efficient score to future work.}
\end{example}

As an example of a new model, one for which no valid FHR moment conditions are already known, we next consider a Mixed Interactive Hazards (MIH) model.

 \begin{example} {(\textsc{Mixed Interactive Hazard (MIH) Model})}   \label{ex_5}
The MIH model relaxes the proportionality assumption of the MPH model; the conditional density of the $t$-th spell equals:
\begin{align}f_{\theta}(y_{t}\,|\, y_{t-1},x_t,a)&=\exp(\gamma y_{t-1}+x_t'\beta+a+\left(x_t'\delta\right)\cdot a)\lambda_{\alpha}(y_t)\notag\\
		&\quad \times \exp\left(-\exp(\gamma y_{t-1}+x_t'\beta+a+\left(x_t'\delta\right)\cdot a)\Lambda_{\alpha}(y_t)\right),\label{eq_dens_MIH}\end{align}
	where $\theta=(\alpha',\beta',\gamma,\delta')'$. Note that \eqref{eq_dens_MIH} simplifies to \eqref{eq_dens_MPH} when $\delta=0$. However, $\delta\neq 0 $ allows for more general interaction effects between the covariate and unobserved heterogeneity. The MIH model is an example of a ``generalized hazards'' model (\citealp{bonev2020nonparametric}). Let $\psi_{\theta}(y_0,y_1,x_1,a)$ be a function satisfying \eqref{eq_psi_zeromean}. The integral equation (\ref{eq_int}), after employing the change of variable $y_2\mapsto p_2$, equals
\begin{align*}
    \int_0^{\infty} \phi_{\theta}(y_0,y_1,p_2,x_1,x_2)e^{(1+x_2'\delta) a}\exp\left(-e^{(1+x_2'\delta) a}p_2\right)\mathrm{d}p_2=\psi_{\theta}(y_0,y_1,x_1,a).
\end{align*}
 Multiplying both sides by $e^{-(1+x_2'\delta) a}$, this is equivalent to
\begin{align*}
    {\cal{L}}\left[\phi_{\theta}(y_0,y_1,\cdot,x_1,x_2)\right]\left(e^{(1+x_2'\delta)a}\right)=e^{-(1+x_2'\delta) a}\psi_{\theta}(y_0,y_1,x_1,a),
\end{align*}
where ${\cal{L}}[g](s)=\int_0^{+\infty} g(z)\exp(-sz)\mathrm{d}z$ denotes the Laplace transform operator. Letting $s=e^{(1+x_2'\delta)a}$, we effectively wish to solve
$${\cal{L}}\left[\phi_{\theta}(y_0,y_1,\cdot,x_1,x_2)\right](s)=s^{-1}\psi_{\theta}\left(y_0,y_1,x_1,\frac{\ln s}{1+x_2'\delta}\right),$$
and, provided $s \mapsto s^{-1}\psi_{\theta}\left(y_0,y_1,x_1,\frac{\ln s}{1+x_2'\delta}\right)$ lies in the range of $\cal{L}$, we can back out $\phi_{\theta}$ using the inverse Laplace transform. To illustrate, take
$\psi_{\theta}(y_0,y_1,x_1,a)=p_1-\exp(-(1+x_1'\delta)a)$,
from which we obtain
$${\cal{L}}\left[\phi_{\theta}(y_0,y_1,\cdot,x_1,x_2)\right](s)=s^{-1}p_1-s^{-\frac{1+x_1'\delta}{1+x_2'\delta}-1},$$
which, provided $\frac{1+\delta'x_1}{1+\delta'x_2}>-1$,\footnote{A similar restriction on predetermined covariates features in the moment restrictions of the censored regression model of \cite{honore2004estimation}.} admits the solution
$$\phi_{\theta}(y_0,y_1,y_2,x_1,x_2)=p_1-\frac{p_2^{\frac{1+x_1'\delta}{1+x_2'\delta}}}{\Gamma\left(1+\frac{1+x_1'\delta}{1+x_2'\delta}\right)}.$$
where $\Gamma$ is the Gamma function. This gives the FHR moment function
	\begin{equation}
		 \phi_{\theta}(y_{0},y_{1},y_{2},x_{1},x_{2})=\left[p_{1}-\frac{p_{2}^{\frac{1+x_{1}'\delta}{1+x_{2}'\delta}}}{\Gamma\left(1+\frac{1+x_{1}'\delta}{1+x_{2}'\delta}\right)}\right]\cdot m(y_0,x_1).\label{eq_MPH_ii4}
	\end{equation}
    More generally, one can obtain closed-form expressions if we choose $\psi_{\theta}$ as a polynomial function of $p_1$.
\end{example}

Examples \ref{ex: count_intro} and \ref{ex_5} illustrate how, using Corollary \ref{coro_charact}, one can derive moment functions by operator inversion. When suitable functions $\psi_{\theta}$ exist, closed-form inversion delivers explicit moment functions. In other settings, it may be that the inverse is not available in closed form, and numerical inversion techniques need to be used. However, operator inverses may be discontinuous, motivating the need for regularization methods (see, e.g., \citealp{engl1996regularization} and \citealp{carrasco2007linear}). In the last part of this section we describe an example where this strategy can be successfully applied.

 \begin{example} {(\textsc{Nonlinear Regression Model})}   
Consider the nonlinear panel data regression model
	\begin{equation}
		Y_{t}=m_{\beta}(Y_{t-1},X_{t},A)+\varepsilon_{t},\quad \varepsilon_{t}\,|\, Y^{t-1},X^t,A\sim {\cal{N}}(0,\sigma^2),
	\end{equation}
	where $a\mapsto m_{\beta}(y_{t-1},x_t,a)$ is a continuously differentiable and strictly increasing bijection on $\mathbb{R}$. Here $\theta=(\beta',\sigma^2)'$. 
    Let $\lambda>0$, and let $k:\mathbb{R}\to[0,\infty)$ be a symmetric probability kernel with Fourier transform $\kappa(s)=\int_{-\infty}^{+\infty} k(u)\exp(\boldsymbol{i}su)\mathrm{d}u$,
that is smooth and compactly supported, so that in particular, $\kappa(0)=1$ and $\kappa'(0)=0$. Define
\begin{align*}
K_{\lambda}(z)=\frac{1}{2\pi}\int_{-\infty}^{+\infty}
\lambda\kappa(\lambda\tau)
\exp\left(\lambda\boldsymbol{i}\tau z+\frac{1}{2}\sigma^2\tau^2\right)
\mathrm{d}\tau .
\end{align*}
$K_{\lambda}(z)$ is akin to the deconvolution kernel introduced by \citet{stefanski1990deconvolving} in the context of a deconvolution problem with normal measurement error. Given a conditionally mean-zero function $\psi_{\theta}(y_0,y_1,x_1,a)$, continuous in $a$, we define
\begin{equation}\phi_{\theta}^{\lambda}(y_0,y_1,y_2,x_1,x_2)=\int_{-\infty}^{+\infty} \psi_{\theta}(y_0,y_1,x_1,a)\frac{\partial m_{\beta}(y_1,x_2,a)}{\partial a}\frac{1}{\lambda}K_{\lambda}\left(\frac{m_{\beta}(y_1,x_2,a)-y_2}{\lambda}\right)\mathrm{d}a.\label{eq_nonlin_phi}\end{equation}
	Applying Corollary \ref{coro_charact} and a regularization strategy, we show in Supplemental Appendix \ref{app: nonlinearregression_details} that (under some regularity conditions)	\begin{equation}\mathbb{E}_{\theta,\omega}\left[\phi^{\lambda}_{\theta}(Y_0,Y_1,Y_2,X_1,X_2)\right]\rightarrow 0 \,\text{ as } \, \lambda \downarrow 0.\label{eq_tobeshown}\end{equation}
	In this sense, $\phi_{\theta}^{\lambda}$ provides an approximately valid moment function for $\theta$; hence it is ``approximately FHR''. As a special case, consider the linear Gaussian model where $m_{\beta}(y_{t-1},x_t,a)=\gamma y_{t-1}+x_t'\beta+a$, and take $\psi_{\theta}(y_0,y_1,x_1,a)=z_1(y_1-\gamma y_0-x_1'\beta-a)$. Then \eqref{eq_nonlin_phi} simplifies to
	\begin{align*}\phi_{\theta}^{\lambda}(y_0,y_1,y_2,x_1,x_2)&=\int z_1(y_1-\gamma y_0-x_1'\beta-a)\frac{1}{\lambda}K_{\lambda}\left(\frac{\gamma y_1+x_2'\beta+a-y_2}{\lambda}\right)\mathrm{d}a\\
		&=z_1\left[(y_1-\gamma y_0-x_1'\beta)-(y_2-\gamma y_1-x_2'\beta)\right],\end{align*}
	which corresponds to the Arellano-Bond moment function for the case $T=2$. Note that $\phi_{\theta}^{\lambda}$ does not depend on $\lambda $ in this case, so the regularization is immaterial. Studying the properties of estimators based on regularized moment functions such as $\phi_{\theta}^{\lambda}$ is an interesting question for future work.

\end{example}

	{\small
		\bibliography{biblio}
	}

	\appendix
	\renewcommand{\thesection}{\Alph{section}}

\section{Proofs of Main Results}

This Appendix contains proofs of our main results. Equation numbering continues in sequence with that established in the main text. All notation is as established in the main text, unless stated otherwise.

\subsection{Proof of Theorem \ref{theo_charact}}
\textbf{Part (A)}. Suppose that $\phi_{\theta}(Y^T,X^T)$ is absolutely integrable under DGP $(\theta,\omega)$. Then, by the law of iterated expectations, we have
\begin{align*}
    \mathbb{E}_{\theta,\omega}[\phi_{\theta}(Y^T,X^T)]&=\mathbb{E}_{\theta,\omega}\left[\mathbb{E}_{\theta,\omega}[\phi_{\theta}(Y^T,X^T)\,|\, Y^{T-1},X^{T-1},A]\right]\\
    &=\mathbb{E}_{\theta,\omega}\left[\int \phi_{\theta}(Y^{T-1},y_T,X^{T-1},x_T)f_{\theta}(y_T\,|\, Y_{T-1},x_T,A)g(x_T\,|\, Y^{T-1},X^{T-1},A)\mathrm{d}y_T\mathrm{d}x_T \right]\\
    &=\mathbb{E}_{\theta,\omega}\left[\int \phi_{\theta}(Y^{T-1},y_T,X^{T-1},x_T)f_{\theta}(y_T\,|\, Y_{T-1},x_T,A)\mathrm{d}y_T\right],	
\end{align*}
for any arbitrary $x_T$ value, where the last equality follows from \eqref{eq_phi_feedback_robust} for $s=T$. By successive applications of the law of iterated expectations and \eqref{eq_phi_feedback_robust} for $s=T-1,...,2$, we get
\begin{align*}
    \mathbb{E}_{\theta,\omega}[\phi_{\theta}(Y^T,X^T)]
    &=\mathbb{E}_{\theta,\omega}\left[\int \phi_{\theta}(Y_{0},y^{1:T},X_{1},x^{2:T})\prod_{t=2}^T f_{\theta}(y_{t}\,|\, y_{t-1},x_{t},A)f_{\theta}(y_{1}\,|\, Y_{0},X_{1},A)\mathrm{d}y^{1:T}\right]
\end{align*}
for any collection of regressor values $x^{2:T}$. Finally, using \eqref{eq_phi_hetero_robust} implies $\mathbb{E}_{\theta,\omega}[\phi_{\theta}(Y^T,X^T)]=0$.  \\

\noindent \textbf{Part (B).}
Suppose that, for all $\omega\in\Omega$ such that $\mathbb{E}_{\theta,\omega}\left[\abs{\phi_{\theta}(Y^T,X^T)}\right]<\infty$ we have $\mathbb{E}_{\theta,\omega}[\phi_{\theta}(Y^T,X^T)]=0$. Let $\eta\mapsto \omega_{\eta}$ denote a smooth path such that $\omega_{\eta^*}=\omega^*$ at some $\eta^*$. By (B)(ii), there exists a $\kappa$-ball around $\eta^*$, $B_{\kappa}(\eta^*)$,  such that for all $\eta \in B_{\kappa}(\eta^*)$, $\mathbb{E}_{\theta,\omega_{\eta}}\left[\abs{\phi_{\theta}(Y^T,X^T)}\right]<\infty$ and thus $\mathbb{E}_{\theta,\omega_{\eta}}[\phi_{\theta}(Y^T,X^T)]=0$. Next, by (B)(i) and (B)(ii), we can apply Lemma 5.4 in \cite{newey1994large}, and conclude that $\eta \mapsto \mathbb{E}_{\theta,\omega_{\eta}}\left[\phi_{\theta}(Y^T,X^T)\right]$ is differentiable at $\eta^{*}$ with derivative $\mathbb{E}_{\theta,\omega^*}[\phi_{\theta}(Y^T,X^T)S^{\eta}(Y^T,X^T)']$ where $S^{\eta}(Y^T,X^T)$ denotes the score at $\eta^*$. Hence, since $\mathbb{E}_{\theta,\omega_{\eta}}[\phi_{\theta}(Y^T,X^T)]=0$ for all $\eta \in B_{\kappa}(\eta^*)$, we have $\mathbb{E}_{\theta,\omega^*}[\phi_{\theta}(Y^T,X^T)S^{\eta}(Y^T,X^T)']=0$. Moreover, by differentiability in quadratic mean, $S^{\eta}$ is square-integrable under DGP $(\theta,\omega^*)$. Claim (B) then follows from Lemma \ref{lem_moment_perp_score}, whose proof is in Supplemental Appendix \ref{app: proof_key_lemma}. 

\begin{lemma} \label{lem_moment_perp_score}
 Suppose that $\mathbb{E}_{\theta,\omega^*}\left[\phi_{\theta}(Y^T,X^T)\right]=0$,  $\mathbb{E}_{\theta,\omega^{*}}[\norm{\phi_{\theta}(Y^T,X^T)}^2]<\infty$ and $\mathbb{E}_{\theta,\omega^*}[\phi_{\theta}(Y^T,X^T)S^{\eta}(Y^T,X^T)']=0$ for all score functions $S^{\eta}$ of smooth parametric submodels. Then (\ref{eq_phi_hetero_robust})-(\ref{eq_phi_feedback_robust}) hold. 
\end{lemma}

\subsection{Proof of Corollary \ref{coro_charact}}

   Assume \eqref{eq_phi_coro}. Then, since $\sum_{t=1}^{T-1}\psi_{\theta,t}(y^t,x^{t},a)$ does not depend on $x_T$, (\ref{eq_phi_feedback_robust}) holds for $s=T$. Next, integrating (\ref{eq_phi_coro}) with respect to $f_{\theta}(y_{T-1}\,|\, y_{T-2},x_{T-1},a)$ implies that
   \begin{align*}
       \int \phi_{\theta}(y^T,x^T) f_{\theta}(y_T\,|\, y_{T-1},x_{T},a)f_{\theta}(y_{T-1}\,|\, y_{T-2},x_{T-1},a)\mathrm{d}y^{T-1:T}=\sum_{t=1}^{T-2}\psi_{\theta,t}(y^t,x^{t},a),
   \end{align*}
   where we have used the fact that $\mathbb{E}\left[\psi_{\theta,T-1}(Y^{T-1},X^{T-1},A)\,|\,Y^{T-2},X^{T-1},A\right]=0$.  In particular, this implies  \eqref{eq_phi_feedback_robust} for $s=T-1$. Integrating against $f_{\theta}(y_{T-2}\,|\, y_{T-3},x_{T-2},a)$ then yields
    \begin{align*}
       \int  \phi_{\theta}(y^T,x^T)\prod_{t=T-2}^Tf_{\theta}(y_{t}\,|\, y_{t-1},x_{t},a)\mathrm{d}y^{T-2:T}=\sum_{t=1}^{T-3}\psi_{\theta,t}(y^t,x^{t},a),
   \end{align*}
   since $\mathbb{E}\left[\psi_{\theta,T-2}(Y^{T-2},X^{T-2},A)\,|\,Y^{T-3},X^{T-2},A\right]=0$. This implies (\ref{eq_phi_feedback_robust}) for $s=T-2$. Continuing this reasoning, we easily conclude that for $s=2,\ldots,T$,
      \begin{align*}
       \int  \phi_{\theta}(y^T,x^T)\prod_{t=s}^Tf_{\theta}(y_{t}\,|\, y_{t-1},x_{t},a)\mathrm{d}y^{s:T}=\sum_{t=1}^{s-1}\psi_{\theta,t}(y^t,x^{t},a),
   \end{align*}
   implying (\ref{eq_phi_feedback_robust}) for $s=2,\ldots,T$. Finally, since $\mathbb{E}\left[\psi_{\theta,1}(Y_0,Y_{1},X_{1},A)\,|\,Y_0,X_{1},A\right]=0$, integrating the following identity with respect to $f_{\theta}(y_{1}\,|\, y_{0},x_{1},a)$ yields (\ref{eq_phi_hetero_robust}):
   \begin{align*}
       \int  \phi_{\theta}(y^T,x^T)\prod_{t=2}^Tf_{\theta}(y_{t}\,|\, y_{t-1},x_{t},a)\mathrm{d}y^{2:T}=\psi_{\theta,1}(y_0,y_{1},x_{1},a).
   \end{align*}
   Conversely, suppose that (\ref{eq_phi_hetero_robust}) and \eqref{eq_phi_feedback_robust} hold. Equation \eqref{eq_phi_feedback_robust} for $s=T$ implies that we can write:
   \begin{align}
		&\int  \phi_{\theta}(y^T,x^T)f_{\theta}(y_{T}\,|\, y_{T-1},x_{T},a)\mathrm{d}y_{T}=\overline{\phi}_{\theta,T-1}\left(y^{T-1},x^{T-1},a\right),
\end{align}
for some function $\overline{\phi}_{\theta,T-1}\left(y^{T-1},x^{T-1},a\right)$ that does not depend on $x_T$. Then, \eqref{eq_phi_feedback_robust} for $s=T-1$ entails that:
\begin{align*}
    \mathbb{E}\left[\overline{\phi}_{\theta,T-1}\left(Y^{T-1},X^{T-1},A\right)\,|\,Y^{T-2},X^{T-1},A\right]=\overline{\phi}_{\theta,T-2}(Y^{T-2},X^{T-2},A)
\end{align*}
for some function $\overline{\phi}_{\theta,T-2}(Y^{T-2},X^{T-2},A)$ that does not depend on $X^{T-1:T}$. Equivalently:
\begin{align*}
    \overline{\phi}_{\theta,T-1}(Y^{T-1},X^{T-1},A)=\overline{\phi}_{\theta,T-2}(Y^{T-2},X^{T-2},A)+\psi_{\theta,T-1}(Y^{T-1},X^{T-1},A),
\end{align*}
with $\mathbb{E}\left[\psi_{\theta,T-1}(Y^{T-1},X^{T-1},A)\,|\,Y^{T-2},X^{T-1},A\right]=0$. Next, \eqref{eq_phi_feedback_robust} for $s=T-2$ implies that: 
\begin{align*}
    \overline{\phi}_{\theta,T-2}(Y^{T-2},X^{T-2},A)=\overline{\phi}_{\theta,T-3}(Y^{T-3},X^{T-3},A)+\psi_{\theta,T-2}(Y^{T-2},X^{T-2},A),
\end{align*}
for some function $\overline{\phi}_{\theta,T-3}(Y^{T-3},X^{T-3},A)$ that does not depend on $X^{T-2:T}$, with $\mathbb{E}\left[\psi_{\theta,T-2}(Y^{T-2},X^{T-2},A)\,|\,Y^{T-3},X^{T-2},A\right]=0$. Continuing this argument based on restriction \eqref{eq_phi_feedback_robust}, we obtain, for $s=2,\ldots,T-1$,
$\overline{\phi}_{\theta,s}(Y^{s},X^{s},A)=\overline{\phi}_{\theta,s-1}(Y^{s-1},X^{s-1},A)+\psi_{\theta,s}(Y^{s},X^{s},A)$,
such that  $\mathbb{E}\left[\psi_{\theta,s}(Y^{s},X^{s},A)\,|\,Y^{s-1},X^{s},A\right]=0$. Furthermore,  
$\overline{\phi}_{\theta,1}=\psi_{\theta,1}$,
with $\mathbb{E}\left[\psi_{\theta,1}(Y_0,Y_{1},X_{1},A)\,|\,Y_0,X_1,A\right]=0$ by (\ref{eq_phi_hetero_robust}). Collecting terms, we have shown that 
   \begin{align*}
		&\int  \phi_{\theta}(y^T,x^T)f_{\theta}(y_{T}\,|\, y_{T-1},x_{T},a)\mathrm{d}y_{T}=\overline{\phi}_{\theta,T-1}(y^{T-1},x^{T-1},a)=\sum_{t=1}^{T-1} \psi_{\theta,t}(y^{t},x^{t},a),
\end{align*}
with the property that $\mathbb{E}\left[\psi_{\theta,t}(Y^{t},X^{t},A)\,|\,Y^{t-1},X^t,A\right]=0$ for $t=1,\ldots,T-1$. Hence \eqref{eq_phi_coro}.

\subsection{Proof of Theorem \ref{theo_charact_aveff}}
Part (A) follows from replacing $\phi_{\theta}(y^T,x^T)$ by $\varphi_{\theta}(y^T,x^T)-h_{\theta}(y^T,x^T,a)$ in the proof of Theorem \ref{theo_charact}, and from applying the triangle inequality to $\varphi_{\theta}(y^T,x^T)-h_{\theta}(y^T,x^T,a)$. To show Part (B), adapt Lemma \ref{lem_moment_perp_score} and its proof to accommodate the fact that $\varphi_{\theta}(y^T,x^T)-h_{\theta}(y^T,x^T,a)$ depends on $a$.

\subsection{Proof of Corollary \ref{coro_charact_aveff}}
	
The proof is the same as the proof of Corollary \ref{coro_charact}, except for the fact that $\phi_{\theta}(y^T,x^T)$ is replaced by $\varphi_{\theta}(y^T,x^T)-h_{\theta}(y^T,x^T,a)$.

\subsection{Proof of Theorem \ref{theo_eff}}
Let $\ell(\theta,\omega_{\eta}|y^T,x^{T})$ denote a smooth parametric submodel, where  $\eta=(\eta_{g,2},\ldots,\eta_{g,T},\eta_{\pi},\eta_{\nu})\in \mathbb{R}^{T+1}$,  $\omega_{\eta}=((g_{t,\eta_{g,t}})_{t=2}^T,\pi_{\eta_{\pi}},\nu_{\eta_{\nu}})$, with   $\omega_{\eta_0}=\omega_0$ for some $\eta_{0}\in \mathbb{R}^{T+1}$.
 Following the logic of Lemma \ref{lem_moment_perp_score}, but now using $\eta_{0}$ in lieu of $\eta^{*}$, each submodel yields a $(T+1)$-dimensional score vector
$    S^{\eta}(y^T,x^T)=\left(S^{\eta_{g,2}}(y^T,x^T),\ldots,S^{\eta_{g,T}}(y^T,x^T),S^{\eta_{\pi}}(y^T,x^T),S^{\eta_{\nu}}(y^T,x^T)\right)'$,
where the first $T-1$ components $S^{\eta_{g,t}}(y^T,x^T), t=2\ldots, T$, are the scores for the feedback process at each period; the penultimate component $S^{\eta_{\pi}}(y^T,x^T)$ corresponds to the heterogeneity component;  and the final term $S^{\eta_{\nu}}(y^T,x^T)$ is the score for the initial condition. \\
By definition, the nonparametric tangent set $\mathcal{T}_{\theta_0,\omega_0,K}$ is the mean-square closure of elements $BS^{\eta}$, where $B$ is a constant $K\times (T+1)$ matrix. Its orthocomplement is
\begin{align*}
\mathcal{T}_{\theta_0,\omega_0,K}^{\perp}=&\left\{\phi\in \mathbb{R}^K\,|\, \mathbb{E}[\phi]=0, \mathbb{E}[ \phi'\phi ]<\infty \text{ with } \mathbb{E}\left[\phi's\right]=0, \quad \text{for all } s\in \mathcal{T}_{\theta_0,\omega_0,K}\right\}\\=&\left\{\phi\in \mathbb{R}^K\,|\, \mathbb{E}[\phi]=0, \mathbb{E}[\phi' \phi]<\infty \text{ with } \mathbb{E}\left[\phi S^{\eta\prime}\right]=0, \right.\\
    &\left.\text{ for all scores } S^{\eta} \text{ of smooth parametric submodels}\right\},
\end{align*}
by Lemma A.1 in \cite{newey1990semiparametric}. Now, $\phi_{\theta_0}\in \mathcal{T}_{\theta_0,\omega_0,K}^{\perp}$ and $\mathbb{E}\left[\phi_{\theta_0}(y^T,x^T)S^{\eta}(y^T,x^T)'\right]=0$ are equivalent. Thus, by Lemma \ref{lem_moment_perp_score}, we conclude that $\mathcal{T}_{\theta_0,\omega_0,K}^{\perp}$ consists of the functions $\phi_{\theta_0}$ whose components satisfies conditions \eqref{eq_phi_hetero_robust} and \eqref{eq_phi_feedback_robust} of Theorem \ref{theo_charact}.

\subsection{Proof of Theorem \ref{theo_working_model}}

By construction, $\widetilde{\phi}_{\theta_0,\widetilde{\omega}}^{\rm eff}(Y^T,X^T)=\widetilde{\Pi}\left(\widetilde{S}^{\theta}(Y^T,X^T)\,|\, \mathcal{T}_{\theta_0,\widetilde{\omega},K}^{\perp}\right)$ is an element of $\mathcal{T}_{\theta_0,\widetilde{\omega},K}^{\perp}$. Thus, it follows from Theorem \ref{theo_eff} that $\widetilde{\phi}_{\theta_0,\widetilde{\omega}}^{\rm eff}$ satisfies conditions \eqref{eq_phi_hetero_robust} and \eqref{eq_phi_feedback_robust}.  The moment restriction \eqref{def_loceff_score_est1} is then a consequence of $\widetilde{\phi}_{\theta_0,\widetilde{\omega}}^{\rm eff}$ being absolutely integrable under DGP $(\theta_0,\omega_0)$ and  Theorem \ref{theo_charact}. By the same argument, 
$\mathbb{E}_{\theta_0,\omega_0}\left[\widetilde{\Pi}(\varphi_{\theta_0}(Y^T,X^T)\,|\, \mathcal{T}_{\theta_0,\widetilde{\omega},L}^{\perp})\right]=0$. The moment condition \eqref{def_loceff_score_est2} then follows from the definition of $\widetilde{\varphi}^{\rm eff}_{\theta_0,\widetilde{\omega}}$ in \eqref{def_loceff_score_mu} and \eqref{eq_varphi}. Lastly, when $\widetilde\omega=\omega_0$ we have $\widetilde{\phi}_{\theta_0,\widetilde{\omega}}^{\rm eff}={\phi}_{\theta_0,{\omega}_0}^{\rm eff}$ and $\widetilde{\varphi}_{\theta_0,\widetilde{\omega}}^{\rm eff}={\varphi}_{\theta_0,{\omega}_0}^{\rm eff}$, yielding the efficient moment functions for $\theta_0$ and $\mu_0$, respectively.

\clearpage

\begin{center}
\textbf{\LARGE ONLINE SUPPLEMENTAL MATERIAL}
\end{center}

\section{Proofs of auxiliary results\label{app_B}}

This Appendix contains additional proofs, derivations and supplemental results. Equation numbering continues in sequence with that established in the main text and Appendix A.

\subsection{Proof of Lemma \ref{lem_moment_perp_score}} \label{app: proof_key_lemma}
Let $\ell(\theta,\omega_{\eta}| y^T,x^{T})$ denote a smooth parametric submodel indexed by the vector $\eta=(\eta_{g,2},\ldots,\eta_{g,T},\eta_{\pi},\eta_{\nu})\in \mathbb{R}^{T+1}$, where $\omega_{\eta}=((g_{t,\eta_{g,t}})_{t=2}^T,\pi_{\eta_{\pi}},\nu_{\eta_{\nu}})$ and  $\omega_{\eta^*}=\omega^*$. We have $(T+1)$ types of scores associated with submodels: one associated with the initial condition density, $\nu_{\eta_{\nu}}$, one associated with the heterogeneity distribution, $\pi_{\eta_{\pi}}$, and $T-1$ scores  associated with the feedback processes for $X_2,\ldots,X_T$, $(g_{t,\eta_{g,t}})_{t=2}^T$. We begin with the heterogeneity component, which yields scores of the form: 
\begin{align*}
    S^{\eta_{\pi}}(y^T,x^T)=\frac{\int \nabla_{\eta_{\pi}}\ln \pi^*(a\,|\, y_0,x_1) p^*(y^{T},x^{T},a)\mathrm{d}a}{\int p^*(y^T,x^T,a)da},
\end{align*}
where 
\begin{align*}
    &p^*(y^T,x^T,a)=\prod_{t=1}^Tf_{\theta}(y_{t}\,|\, y_{t-1},x_{t},a)\prod_{t=2}^Tg^*(x_{t}\,|\, y^{t-1},x^{t-1},a)\pi^*(a\,|\, y_{0},x_{1}),
\end{align*}
and $ \nabla_{\eta_{\pi}}\ln \pi^*$ denotes the score of the submodel component $\eta_{\pi}\mapsto \pi_{\eta_{\pi}}$ at $\omega^*$.\\
Next we consider the scores for the $(T-1)$ feedback components:
$$S^{\eta_{g,t}}(y^T,x^T)=\frac{\int \nabla_{\eta_{g,t}}\ln g^*(x_t\,|\, y^{t-1},x^{t-1},a) p^*(y^T,x^T,a)\mathrm{d}a}{\int p^*(y^T,x^T,a)da},\quad t=2,...,T.$$
Finally, for the initial condition, the scores take the form:
$$S^{\eta_{\nu}}(y^T,x^T)=\nabla_{\eta_{\nu}}\ln \nu^*(y_0,x_1).$$
Let $S^\eta(y^T,x^T)=\left(S^{\eta_{g,2}}(y^T,x^T),\ldots,S^{\eta_{g,T}}(y^T,x^T),S^{\eta_{\pi}}(y^T,x^T),S^{\eta_{\nu}}(y^T,x^T)\right)'$. Without loss of generality, suppose $\phi_{\theta}$ is scalar. By assumption, we have
\begin{align} \label{tauperp_hetero}
    0&=\mathbb{E}_{\theta,\omega^*}\left[\phi_{\theta}\left(Y^T,X^T\right)S^{\eta_{\pi}}(Y^T,X^T)\right] \\
    &=\int \phi_{\theta}(y^T,x^T)\nabla_{\eta_{\pi}}\ln \pi^*(a\,|\, y_0,x_1) p^*(y^T,x^T,a)\nu^*(y_0,x_1)\mathrm{d}a\mathrm{d}y^T\mathrm{d}x^T,  \notag
\end{align}
and, for $ t=2,\ldots,T$,
\begin{align} \label{tauperp_feedback}
    0&=\mathbb{E}_{\theta,\omega^*}\left[\phi_{\theta}(Y^T,X^T)S^{\eta_{g,t}}(Y^T,X^T)\right] \\
    &=\int \phi_{\theta}(y^T,x^T)\nabla_{\eta_{g,t}}\ln g^*(x_t\,|\, y^{t-1},x^{t-1},a) p^*(y^T,x^T,a)\nu^*(y_0,x_1)\mathrm{d}a\mathrm{d}y^T\mathrm{d}x^T, \notag
\end{align}
and 
\begin{align} \label{tauperp_initcond}
    0&=\mathbb{E}_{\theta,\omega^*}\left[\phi_{\theta}(Y^T,X^T)S^{\eta_{\nu}}(Y^T,X^T)\right] \\
    &=\int \phi_{\theta}(y^T,x^T) \nabla_{\eta_{\nu}}\ln \nu^*(y_0,x_1)  p^*(y^T,x^T,a)\nu^*(y_0,x_1)\mathrm{d}a\mathrm{d}y^T\mathrm{d}x^T. \notag
\end{align}
Note that equation \eqref{tauperp_feedback} for $t=T$, after integrating over $y_T$, coincides with:
\begin{align*}
    0&=\int \nabla_{\eta_{g,T}}\ln g^*(x_T\,|\, y^{T-1},x^{T-1},a)\mathbb{E}_{\theta,\omega^*}\left[\phi_{\theta}(Y^T,X^T)|Y^{T-1}=y^{T-1},X^T=x^T,A=a\right] \\
    &\quad \times \prod_{t=1}^{T-1}f_{\theta}(y_{t}\,|\, y_{t-1},x_{t},a)\prod_{t=2}^Tg^*(x_{t}\,|\, y^{t-1},x^{t-1},a)\pi^*(a\,|\, y_{0},x_{1})\nu^*(y_0,x_1)\mathrm{d}a \mathrm{d}y^{T-1}\mathrm{d}x^T.
\end{align*}
Now, since $\nabla_{\eta_{g,T}}\ln g^*(x_T\,|\, y^{T-1},x^{T-1},a)$ is unrestricted except for the fact that it has zero mean conditional on $(y^{T-1},x^{T-1},a)$ and is square-integrable, we can specifically choose
\begin{align*}
    \nabla_{\eta_{g,T}}\ln g^*(x_T\,|\, y^{T-1},x^{T-1},a)&=\mathbb{E}_{\theta,\omega^*}\left[\phi_{\theta}(Y^T,X^T)|Y^{T-1}=y^{T-1},X^T=x^T,A=a\right] \\
    &\quad -\mathbb{E}_{\theta,\omega^*}\left[\phi_{\theta}(Y^T,X^T)|Y^{T-1}=y^{T-1},X^{T-1}=x^{T-1},A=a\right],
\end{align*}
in which case \eqref{tauperp_feedback} (for $t=T$) evaluates to (for $\mathbb{V}_{\theta,\omega^*}$ the variance under $(\theta,\omega^*)$):
\begin{align*}
    \int \mathbb{V}_{\theta,\omega^*}\left(\mathbb{E}_{\theta,\omega^*}\left[\phi_{\theta}(Y^T,X^T)|Y^{T-1}=y^{T-1},X^T=x^T,A=a\right]|Y^{T-1}=y^{T-1},X^{T-1}=x^{T-1},A=a \right) \\
    \times \prod_{t=1}^{T-1}f_{\theta}(y_{t}\,|\, y_{t-1},x_{t},a)\prod_{t=2}^{T-1}g^*(x_{t}\,|\, y^{t-1},x^{t-1},a)\pi^*(a\,|\, y_{0},x_{1})\nu^*(y_0,x_1)\mathrm{d}a\mathrm{d}y^{T-1}\mathrm{d}x^{T-1}=0,
\end{align*}
which holds if, and only if, $\mathbb{E}_{\theta,\omega^*}\left[\phi_{\theta}(Y^T,X^T)|Y^{T-1}=y^{T-1},X^T=x^T,A=a\right]$ does not depend on $x_T$.  Conversely, if $\mathbb{E}_{\theta,\omega^*}\left[\phi_{\theta}(Y^T,X^T)|Y^{T-1}=y^{T-1},X^T=x^T,A=a\right]$ does not depend on $x_T$, it is clear that  equation \eqref{tauperp_feedback} for $t=T$ is satisfied since the elements $\nabla_{\eta_{g,T}}\ln g^*(x_T\,|\, y^{T-1},x^{T-1},a)$ have zero mean conditional on $(y^{T-1},x^{T-1},a)$. Therefore,  \eqref{tauperp_feedback} for $t=T$ is equivalent to the requirement that $\mathbb{E}_{\theta,\omega^*}\left[\phi_{\theta}(Y^T,X^T)|Y^{T-1}=y^{T-1},X^T=x^T,A=a\right]$ does not depend on $x_T$.\\
Next, consider \eqref{tauperp_feedback} for $t=T-1$. Exploiting the result above, we can integrate over $y_{T-1:T}$ and $x_T$ to get:
\begin{align*}
    0&=\int \nabla_{\eta_{g,T-1}}\ln g^*(x_{T-1}\,|\, y^{T-2},x^{T-2},a)\mathbb{E}_{\theta,\omega^*}\left[\phi_{\theta}(Y^T,X^T)|Y^{T-2}=y^{T-2},X^{T-1}=x^{T-1},A=a\right] \\
    &\quad \times \prod_{t=1}^{T-2}f_{\theta}(y_{t}\,|\, y_{t-1},x_{t},a)\prod_{t=2}^{T-1}g^*(x_{t}\,|\, y^{t-1},x^{t-1},a)\pi^*(a\,|\, y_{0},x_{1})\nu^*(y_0,x_1)\mathrm{d}a\mathrm{d}y^{T-2}\mathrm{d}x^{T-1}.
\end{align*}
Observe that the lack of restrictions on $\nabla_{\eta_{g,T-1}}\ln g^*(x_{T-1}\,|\, y^{T-2},x^{T-2},a)$, besides it being mean zero conditional on $(y^{T-2},x^{T-2},a)$ and being square-integrable, implies, by a logic analogous to that used for the $t=T$ case, that \eqref{tauperp_feedback} for $t=T-1$ is equivalent to the requirement that
\begin{align*}&\mathbb{E}_{\theta,\omega^*}\left[\phi_{\theta}(Y^T,X^T)|Y^{T-2}=y^{T-2},X^{T-1}=x^{T-1},A=a\right]\\&\quad =\int \phi_{\theta}(y^T,x^T)\prod_{t=T-1}^{T}f_{\theta}(y_{t}\,|\, y_{t-1},x_{t},a)\mathrm{d}y^{T-1:T}\end{align*}
does not depend on $x^{T-1:T}$. By inductive reasoning, we conclude that, for all $s=2,\ldots,T$,
\begin{align} \label{tauperp_feedback_bis}
        \int  \phi_{\theta}(y^T,x^T)\prod_{t=s}^Tf_{\theta}(y_{t}\,|\, y_{t-1},x_{t},a)\mathrm{d}y^{s:T} \text{ does not depend on }x^{s:T}.
\end{align}
That is, $\phi_{\theta}$ satisfies \eqref{eq_phi_feedback_robust}.\\
Next, consider equation \eqref{tauperp_hetero}. Using the results immediately above we get, after integrating over $y^{1:T}$ and $x^{2:T}$:
\begin{align*}
\int \nabla_{\eta_{\pi}}\ln \pi^*(a\,|\, y_0,x_1)\mathbb{E}_{\theta,\omega^*}\left[\phi_{\theta}(Y^T,X^T)|Y_0=y_0,X_1=x_1,A=a\right] \pi^*(a\,|\, y_0,x_1)\nu^*(y_0,x_1)\mathrm{d}a\mathrm{d}y_{0}\mathrm{d}x_{1}=0, 
\end{align*}
and since $\nabla_{\eta_{\pi}}\ln \pi^*(a\,|\, y_0,x_1)$ is unrestricted (beyond having zero mean conditional on $(y_0,x_1)$ and being square-integrable), choosing specifically
\begin{align*}
    \nabla_{\eta_{\pi}}\ln \pi^*(a\,|\, y_0,x_1)&=\mathbb{E}_{\theta,\omega^*}\left[\phi_{\theta}(Y^T,X^T)|Y_0=y_0,X_1=x_1,A=a\right]\\&\quad-\mathbb{E}_{\theta,\omega^*}\left[\phi_{\theta}(Y^T,X^T)|Y_0=y_0,X_1=x_1\right]
\end{align*}
implies that
\begin{align*}
&\int \mathbb{V}_{\theta,\omega^*}\left(\mathbb{E}_{\theta,\omega^*}\left[\phi_{\theta}(Y^T,X^T)|Y_0=y_0,X_1=x_1,A=a\right]|Y_0=y_0,X_1=x_1\right) \nu^*(y_0,x_1)\mathrm{d}y_{0}\mathrm{d}x_{1}=0. 
\end{align*}
Thus, $\mathbb{E}_{\theta,\omega^*}\left[\phi_{\theta}(Y^T,X^T)|Y_0=y_0,X_1=x_1,A=a\right]$ does not depend on $a$. 
Conversely, if $\mathbb{E}_{\theta,\omega^*}\left[\phi_{\theta}(Y^T,X^T)|Y_0=y_0,X_1=x_1,A=a\right]$ does not depend on $a$, then \eqref{tauperp_hetero} also holds since $\nabla_{\eta_{\pi}}\ln \pi^*(a\,|\, y_0,x_1)$ has mean-zero conditional on $(y_0,x_1)$. We therefore conclude that \eqref{tauperp_hetero} is equivalent to
\begin{align*}
    &\mathbb{E}_{\theta,\omega^*}\left[\phi_{\theta}(Y^T,X^T)|Y_0=y_0,X_1=x_1,A=a\right] \\
    &=\int \phi_{\theta}(y^T,x^T)\prod_{t=1}^Tf_{\theta}(y_{t}\,|\, y_{t-1},x_{t},a)\prod_{t=2}^Tg^*(x_{t}\,|\, y^{t-1},x^{t-1},a)\mathrm{d}y^{1:T}\mathrm{d}x^{2:T} \\
    &=\int \phi_{\theta}(y^T,x^T)\prod_{t=1}^Tf_{\theta}(y_{t}\,|\, y_{t-1},x_{t},a)\mathrm{d}y^{1:T} \text{ (by \eqref{tauperp_feedback_bis} for $s=2$)}
\end{align*}
is a constant independent of $a$.  Hence, we can write
\begin{align}\label{tauperp_hetero_bis}
        \int \phi_{\theta}(y^T,x^T)\prod_{t=1}^Tf_{\theta}(y_{t}\,|\, y_{t-1},x_{t},a)\mathrm{d}y^{1:T}=\mathbb{E}_{\theta,\omega^*}\left[\phi_{\theta}(Y^T,X^T)|Y_0=y_0,X_1=x_1\right].
\end{align}
Next, we also note that \eqref{tauperp_initcond} can simply be rewritten as:
\begin{align*}
    0=\int \mathbb{E}_{\theta,\omega^*}\left[\phi_{\theta}(Y^T,X^T)|Y_0=y_0,X_1=x_1\right] \nabla_{\eta_{\nu}}\ln \nu^*(y_0,x_1) \nu^*(y_0,x_1)\mathrm{d}y_0\mathrm{d}x_1.
\end{align*}
Since $\nabla_{\eta_{\nu}}\ln \nu^*(y_0,x_1)$ is unrestricted besides being mean-zero and square-integrable, a similar reasoning to before implies that  \eqref{tauperp_initcond} is equivalent to $\mathbb{E}_{\theta,\omega^*}\left[\phi_{\theta}(Y^T,X^T)|Y_0=y_0,X_1=x_1\right]=\mathbb{E}_{\theta,\omega^*}\left[\phi_{\theta}(Y^T,X^T)\right]=0$. This result in combination with \eqref{tauperp_hetero_bis} implies that $\phi_{\theta}$ satisfies \eqref{eq_phi_hetero_robust}.

\subsection{Proof of Lemma \ref{lem_moment_charac_MPH}} \label{app: proof_mainlemma_MPH}
Recall the bijective transformation between $(y_{1},\ldots,y_{T})$
and $(p_{1},\ldots,p_{T})$ given by 
\begin{align*}
p_{t}=\Lambda_{\alpha}(y_{t})e^{\gamma y_{t-1}+x_{t}'\beta},\quad t=1,\ldots,T,
\end{align*}
as well as the second (bijective) transformation between $(p_{1},\ldots,p_{T})$
and $(\widetilde{p}_{1},\ldots,\widetilde{p}_{T-1},\overline{p})$
given by 
\begin{align*}
\widetilde{p}_{t} & =\frac{p_{t}}{\sum_{s=t}^{T}p_{s}},\quad t=1,\ldots,T-1,\\
\overline{p} & =\sum_{t=1}^{T}p_{t}.
\end{align*}
The distribution of $\left(P_{1},\ldots,P_{T}\right)$ is stated in
Lemma \ref{lem: P1_PT_are_exponential} and that of $(\widetilde{P}_{1},\ldots,\widetilde{P}_{T-1},\overline{P})$
in Lemma \ref{lem: helmert_tranformation_MPH}.

From these two bijections we get the following two equivalent representations
of $\phi_{\theta}$: 
\begin{align}
\xi_{\theta}(y_{0},p^{T},x^{T})= & \phi_{\theta}\left(y_{0},\Lambda_{\alpha}^{-1}(p_{1}\exp(-\gamma y_{0}-x_{1}'\beta)),\right.\notag \\
& \left.\ldots,\Lambda_{\alpha}^{-1}(p_{T}\exp(-\gamma y_{T-1}-x_{T}'\beta)),x^{T}\right)\label{eq: phi_in_terms_of_rho}\\
\psi_{\theta}(y_{0},\widetilde{p}^{T-1},\overline{p},x^{T})= & \xi_{\theta}(y_{0},\widetilde{p}_{1}\overline{p},(1-\widetilde{p}_{1})\widetilde{p}_{2}\overline{p},\ldots,\prod_{s=1}^{T-2}(1-\widetilde{p}_{s})\widetilde{p}_{T-1}\overline{p},\prod_{s=1}^{T-1}(1-\widetilde{p}_{s})\overline{p},x^{T}).\label{eq: phi_in_terms_rho_tilde_bar}
\end{align}

In the present context, condition (\ref{eq_phi_feedback_robust})
of Theorem \ref{theo_charact} for $s=T$ corresponds to a requirement on (\ref{eq: phi_in_terms_of_rho})
of
\begin{align*}
\int_{0}^{+\infty}\xi_{\theta}(y_{0},p^{T},x^{T})e^{a}e^{-e^{a}p_{T}}\mathrm{d}p_{T}\text{ does not depend on }x_{T},
\end{align*}
where we have used the change of variable, $p_{T}=\Lambda_{\alpha}\left(y_{T}\right)\exp\left(\gamma y_{T-1}+x_{T}'\beta\right)$
(see Lemma \ref{lem: P1_PT_are_exponential}). This implies that $\mathcal{L}[p_{T}\mapsto\xi_{\theta}(y_{0},p^{T-1},p_{T},x^{T-1},x_{T})][e^{a}]=\mathcal{L}[p_{T}\mapsto\xi_{\theta}(y_{0},p^{T-1},p_{T},x^{T-1},\widetilde{x}_{T})][e^{a}]$
for any $\widetilde{x}_{T}$, where ${\cal{L}}$ denotes the Laplace transform operator. By the uniqueness of the Laplace transform,
it follows that $\xi_{\theta}(y_{0},p^{T-1},p_{T},x^{T-1},x_{T})=\xi_{\theta}(y_{0},p^{T-1},p_{T},x^{T-1},\widetilde{x}_{T})$,
meaning that $\xi_{\theta}$ does not depend on $x_{T}$. Hereafter we therefore suppress the dependence of $\phi_{\theta},\xi_{\theta},\psi_{\theta}$
on $x_{T}$. 

Next, consider the decomposition 
\begin{align*}
 & \psi_{\theta}(Y_{0},\widetilde{P}^{T-1},\overline{P},X^{T-1})=\\
 & \left(\psi_{\theta}(Y_{0},\widetilde{P}^{T-1},\overline{P},X^{T-1})-\mathbb{E}\left[\left.\psi_{\theta}(Y_{0},\widetilde{P}^{T-1},\overline{P},X^{T-1})\right|Y_{0},\widetilde{P}^{T-2},\overline{P},X^{T-1}\right]\right)\\
 & +\left(\mathbb{E}\left[\left.\psi_{\theta}(Y_{0},\widetilde{P}^{T-1},\overline{P},X^{T-1})\right|Y_{0},\widetilde{P}^{T-2},\overline{P},X^{T-1}\right]-\mathbb{E}\left[\left.\psi_{\theta}(Y_{0},\widetilde{P}^{T-1},\overline{P},X^{T-1})\right|Y_{0},\widetilde{P}^{T-3},\overline{P},X^{T-2}\right]\right)\\
 & +\left(\mathbb{E}\left[\left.\psi_{\theta}(Y_{0},\widetilde{P}^{T-1},\overline{P},X^{T-1})\right|Y_{0},\widetilde{P}^{T-3},\overline{P},X^{T-2}\right]-\mathbb{E}\left[\left.\psi_{\theta}(Y_{0},\widetilde{P}^{T-1},\overline{P},X^{T-1})\right|Y_{0},\widetilde{P}^{T-4},\overline{P},X^{T-3}\right]\right)\\
 & \vdots\\
 & +\mathbb{E}\left[\left.\psi_{\theta}(Y_{0},\widetilde{P}^{T-1},\overline{P},X^{T-1})\right|Y_{0},\widetilde{P}_{1},\overline{P},X^{2}\right],
\end{align*}
or, succinctly, 
\begin{align*}
\psi_{\theta}(Y_{0},\widetilde{P}^{T-1},\overline{P},X^{T-1})=\sum_{t=1}^{T-2}\psi_{\theta,t}(Y_{0},\widetilde{P}^{t},\overline{P},X^{t+1})+\psi_{\theta,T-1}(Y_{0},\widetilde{P}^{T-1},\overline{P},X^{T-1}),
\end{align*}
where 
\begin{align*}
\psi_{\theta,T-1}(Y_{0},\widetilde{P}^{T-1},\overline{P},X^{T-1}) & =\psi_{\theta}(Y_{0},\widetilde{P}^{T-1},\overline{P},X^{T-1})-\mathbb{E}\left[\left.\psi_{\theta}(Y_{0},\widetilde{P}^{T-1},\overline{P},X^{T-1})\right|Y_{0},\widetilde{P}^{T-2},\overline{P},X^{T-1}\right],
\end{align*}
for all $t=2,\ldots,T-2$,
\begin{multline*}
\psi_{\theta,t}\left(Y_{0},\widetilde{P}^{t},\overline{P},X^{t+1}\right)\\
=\mathbb{E}\left[\left.\psi_{\theta}(Y_{0},\widetilde{P}^{T-1},\overline{P},X^{T-1})\right|Y_{0},\widetilde{P}^{t},\overline{P},X^{t+1}\right]-\mathbb{E}\left[\left.\psi_{\theta}(Y_{0},\widetilde{P}^{T-1},\overline{P},X^{T-1})\right|Y_{0},\widetilde{P}^{t-1},\overline{P},X^{t}\right],
\end{multline*}
 and 
\begin{align*}
\psi_{\theta,1}(Y_{0},\widetilde{P}_{1},\overline{P},X^{2}) & =\mathbb{E}\left[\left.\psi_{\theta}(Y_{0},\widetilde{P}^{T-1},\overline{P},X^{T-1})\right|Y_{0},\widetilde{P}_{1},\overline{P},X^{2}\right].
\end{align*}
The law of iterated expectations readily implies
that: 
\begin{align*}
 & \mathbb{E}\left[\left.\psi_{\theta,T-1}(Y_{0},\widetilde{P}^{T-1},\overline{P},X^{T-1})\right|Y_{0},\widetilde{P}^{T-2},\overline{P},X^{T-1}\right]=0,\\
 & \mathbb{E}\left[\left.\psi_{\theta,t}(Y_{0},\widetilde{P}^{t},\overline{P},X^{t+1})\right|Y_{0},\widetilde{P}^{t-1},\overline{P},X^{t}\right]=0,\quad t=2,\ldots,T-2.
\end{align*}
It only remains to show that (i) for all $t=1,\ldots,T-2$, $\psi_{\theta,t}(Y_{0},\widetilde{P}^{t},\overline{P},X^{t+1})$
does not depend on $X_{t+1}$ and (ii) 
\begin{align*}
\mathbb{E}\left[\left.\psi_{\theta}(Y_{0},\widetilde{P}^{T-1},\overline{P},X^{T-1})\right|Y_{0},\overline{P},X_{1}\right]=0.
\end{align*}
We start by establishing (i). To that end, fix any $s\in\{2,\ldots,T-1\}$.
By Theorem \ref{theo_charact} equation (\ref{eq_phi_feedback_robust}),
we have 
\begin{align*}
\int_{0}^{\infty}\xi_{\theta}(y_{0},p^{T},x^{T-1})e^{(T-s+1)a}e^{-e^{a}\sum_{t=s}^{T}p_{t}}\mathrm{d}p^{s:T}\text{ does not depend on }x^{s:T},
\end{align*}
using the change of variables, $p_{t}=\Lambda_{\alpha}\left(y_{t}\right)\exp\left(\gamma y_{t-1}+x_{t}'\beta\right)$
for $t=s,\ldots,T$. Next, recall that $\widetilde{p}_{t}=\frac{p_{t}}{\sum_{k=t}^{T}p_{k}},\quad t=s,\ldots,T-1$,
and introduce $\overline{p}_{s}=\sum_{k=s}^{T}p_{k}$. By Lemma \ref{lem: helmert_tranformation_MPH}
and a second change of variables, the previous condition is equivalent
to 
\begin{align*}
\int_{0}^{\infty}\Upsilon_{\theta}(y_{0},p^{s-1},\overline{p}_{s},x^{T-1})\frac{e^{(T-s+1)a}}{\Gamma(T-s+1)}\left(\overline{p}_{s}\right)^{T-s}e^{-e^{a}\overline{p}_{s}}\mathrm{d}\overline{p}_{s}\text{ does not depend on }x^{s:T},
\end{align*}
where 
\begin{align*}
\Upsilon_{\theta}(y_{0},p^{s-1},\overline{p}_{s},x^{T-1}) & =\int_{0}^{1}\xi_{\theta}(y_{0},p^{s-1},\widetilde{p}_{s}\overline{p}_{s},(1-\widetilde{p}_{s})\widetilde{p}_{s+1}\overline{p}_{s},\ldots,\prod_{k=s}^{T-2}(1-\widetilde{p}_{k})\widetilde{p}_{T-1}\overline{p}_{s},\prod_{k=s}^{T-1}(1-\widetilde{p}_{k})\overline{p}_{s},x^{T-1})\\
 & \times\prod_{t=s}^{T-1}\frac{\Gamma(T-t+1)}{\Gamma(T-t)}(1-\widetilde{p}_{t})^{T-t-1}\mathrm{d}\widetilde{p}^{s:T-1}.
\end{align*}
Hence, 
\begin{align*}
\mathcal{L}[\overline{p}_{s}\mapsto\left(\overline{p}_{s}\right)^{T-s}\Upsilon_{\theta}(y_{0},p^{s-1},\overline{p}_{s},x^{s-1},x^{s:T-1})][e^{a}]=\mathcal{L}[\overline{p}_{s}\mapsto\left(\overline{p}_{s}\right)^{T-s}\Upsilon_{\theta}(y_{0},p^{s-1},\overline{p}_{s},x^{s-1},\widetilde{x}^{s:T-1})][e^{a}],
\end{align*}
for any possible values $\widetilde{x}^{s:T-1}$. By uniqueness of the
Laplace transform, we conclude that $\Upsilon_{\theta}$ does not
depend on $x^{s:T-1}$. Since there is a bijective transformation
between $(p^{s-1},\overline{p}_{s})$ and $(\widetilde{p}^{s-1},\overline{p})$
given by $\widetilde{p}_{t}=\frac{p_{t}}{\sum_{k=t}^{s-1}p_{k}+\overline{p}_{s}},\quad t=1,\ldots,s-1$,
and $\overline{p}=\sum_{k=1}^{s-1}p_{k}+\overline{p}_{s}$, we have: 
\begin{align*}
\Upsilon_{\theta}(y_{0},p^{s-1},\overline{p}_{s},x^{T-1}) & =\int_{0}^{1}\psi_{\theta}(y_{0},\widetilde{p}^{T-1},\overline{p},x^{T-1})\prod_{t=s}^{T-1}\frac{\Gamma(T-t+1)}{\Gamma(T-t)}(1-\widetilde{p}_{t})^{T-t-1}\mathrm{d}\widetilde{p}^{s:T-1}\\
 & =\mathbb{E}\left[\left.\psi_{\theta}(Y_{0},\widetilde{P}^{T-1},\overline{P},X^{T-1})\right|Y_{0}=y_{0},\widetilde{P}^{s-1}=\widetilde{p}^{s-1},\overline{P}=\overline{p},X^{s}=x^{s}\right],
\end{align*}
which does not depend on $x^{s:T-1}$. The second equality follows from the fact that the ``forward orthogonal transforms'' $\widetilde{P}^{s:T-1}$ are independent of $\widetilde{P}^{s-1},\overline{P},X^{s}$ by part (i) of Lemma \ref{lem: P1_PT_are_exponential} and Lemma \ref{lem: helmert_tranformation_MPH}. This shows that, for $s\in\{3,\ldots,T-1\}$,
\begin{multline*}
\psi_{\theta,s-1}(Y_{0},\widetilde{P}^{s-1},\overline{P},X^{s})\\
=\mathbb{E}\left[\left.\psi_{\theta}(Y_{0},\widetilde{P}^{T-1},\overline{P},X^{T-1})\right|Y_{0},\widetilde{P}^{s-1},\overline{P},X^{s}\right]-\mathbb{E}\left[\left.\psi_{\theta}(Y_{0},\widetilde{P}^{T-1},\overline{P},X^{T-1})\right|Y_{0},\widetilde{P}^{s-2},\overline{P},X^{s-1}\right]
\end{multline*}
does not depend on $X^{s:T}$, and that, for $s=2$ 
\begin{align*}
\psi_{\theta,1}(Y_{0},\widetilde{P}_{1},\overline{P},X^{2})=\mathbb{E}\left[\left.\psi_{\theta}(Y_{0},\widetilde{P}^{T-1},\overline{P},X^{T-1})\right|Y_{0},\widetilde{P}_{1},\overline{P},X^{2}\right]
\end{align*}
does not depend on $X^{2:T}$. This proves (i). 

Finally, by Theorem \ref{theo_charact} equation (\ref{eq_phi_hetero_robust}),
we have 
\begin{align*}
\int_{0}^{\infty}\xi_{\theta}(y_{0},p^{T},x^{T-1})e^{Ta}e^{-e^{a}\sum_{t=1}^{T}p_{t}}\mathrm{d}p^{T}=0,
\end{align*}
where we have once more used the transformation $p_{t}=\Lambda_{\alpha}\left(y_{t}\right)\exp\left(\gamma y_{t-1}+x_{t}'\beta\right)$
for $t=1,\ldots,T$. Consider now, $\widetilde{p}_{t}=\frac{p_{t}}{\sum_{k=t}^{T}p_{k}},\quad t=1,\ldots,T-1$
and $\overline{p}=\sum_{k=1}^{T}p_{k}$. By Lemma \ref{lem: helmert_tranformation_MPH}
and a second change of variables, the previous condition is equivalent
to 
\begin{align*}
\int_{0}^{\infty}\Upsilon_{\theta}(y_{0},\overline{p},x^{T-1})\frac{e^{Ta}}{\Gamma(T)}\overline{p}^{T-1}e^{-e^{a}\overline{p}}\mathrm{d}\overline{p}=0,
\end{align*}
where now $\Upsilon_{\theta}(y_{0},\overline{p},x^{T-1})=\left(\int_{0}^{1}\psi_{\theta}(y_{0},\widetilde{p}^{T-1},\overline{p},x^{T-1})\prod_{t=1}^{T-1}\frac{\Gamma(T-t+1)}{\Gamma(T-t)}(1-\widetilde{p}_{t})^{T-t-1}\mathrm{d}\widetilde{p}^{T-1}\right)$.
Hence, 
\begin{align*}
\mathcal{L}[\overline{p}\mapsto\overline{p}^{T-1}\Upsilon_{\theta}(y_0,\overline{p},x^{T-1})][e^{a}]=0,
\end{align*}
which by uniqueness of the Laplace transform implies that 
\begin{align*}
0=\Upsilon_{\theta}(Y_0,\overline{P},X^{T-1})=\mathbb{E}\left[\left.\psi_{\theta}(Y_{0},\widetilde{P}^{T-1},\overline{P},X^{T-1})\right|Y_{0},\overline{P},X_{1}\right],
\end{align*}
where the second equality is again a consequence of Lemma \ref{lem: helmert_tranformation_MPH}. This establishes (ii) which finally yields the desired representation.

To prove the converse, suppose that we have the representation
\begin{align*}
\phi_{\theta}\left(Y^{T},X^{T}\right)=&\sum_{t=1}^{T-1}\psi_{\theta,t}(Y_{0},\widetilde{P}_{\theta}^{t},\overline{P}_{\theta},X^{t}),
\end{align*}
where, for $t=1,\ldots,T-1$, $\mathbb{E}\left[\left.\psi_{\theta,t}(Y_{0},\widetilde{P}_{\theta}^{t},\overline{P}_{\theta},X^{t})\right|Y_{0},\widetilde{P}_{\theta}^{t-1},\overline{P}_{\theta},X^{t}\right]=0$. We have
\begin{align*}
&\int_{0}^{+\infty} \phi_{\theta}(y^T,x^T){{\prod_{t=1}^T f_{\theta}(y_t\,|\, y_{t-1},x_{t},a)}}\mathrm{d}y^{1:T} \\
&=
\int_{0}^{\infty}\sum_{t=1}^{T-1}\psi_{\theta,t}\left(y_{0},\frac{p_1}{\sum_{s=1}^Tp_{s}},\ldots,\frac{p_t}{\sum_{s=t}^Tp_{s}},\sum_{s=1}^Tp_{s},x^{t}\right)e^{Ta}e^{-e^{a}\sum_{t=1}^{T}p_{t}}\mathrm{d}p^{T} \\
&=
\int_0^\infty\int_{(0,1)^{T-1}}\sum_{t=1}^{T-1}\psi_{\theta,t}(y_0,\widetilde p^t,\overline p,x^t)\prod_{j=1}^{T-1}
\frac{\Gamma(T-j+1)}{\Gamma(T-j)}(1-\widetilde p_j)^{T-j-1}\frac{e^{Ta}}{\Gamma(T)}\overline p^{T-1}e^{-e^a\overline p} \mathrm d\widetilde p^{T-1}\,\mathrm d\overline p\\
&=\sum_{t=1}^{T-1}\int_0^\infty\int_{(0,1)^{t-1}}\left[\int_0^1\psi_{\theta,t}(y_0,\widetilde p^{t-1},\widetilde p_t,\overline p,x^t)\frac{\Gamma(T-t+1)}{\Gamma(T-t)}(1-\widetilde p_t)^{T-t-1}\,\mathrm d\widetilde p_t\right]\\
&\times\prod_{j=1}^{t-1}\frac{\Gamma(T-j+1)}{\Gamma(T-j)}(1-\widetilde p_j)^{T-j-1}\frac{e^{Ta}}{\Gamma(T)}\overline p^{T-1}e^{-e^a\overline p}\mathrm d\widetilde p^{t-1}\mathrm d\overline p\\
&=0.
\end{align*}
The first and second equalities employ the successive change of variables $(y_{1},\ldots,y_{T})\to (p_{1},\ldots,p_{T})$ and $(p_{1},\ldots,p_{T})\to  (\widetilde{p}_{1},\ldots,\widetilde{p}_{T-1},\overline{p})$ introduced at the beginning of the proof. The last equality follows from recognizing that by Lemma \ref{lem: helmert_tranformation_MPH},
\begin{align*}
    &\int_0^1\psi_{\theta,t}(y_0,\widetilde p^{t-1},\widetilde p_t,\overline p,x^t)\frac{\Gamma(T-t+1)}{\Gamma(T-t)}(1-\widetilde p_t)^{T-t-1}\,\mathrm d\widetilde p_t \\
    &=  \mathbb{E}\left[\left.\psi_{\theta,t}(Y_{0},\widetilde{P}_{\theta}^{t},\overline{P}_{\theta},X^{t})\right|Y_{0}=y_{0},\widetilde{P}_{\theta}^{t-1}=\widetilde p^{t-1},\overline{P}_{\theta}=\overline p,X^{t}=x^t\right]\\
    &=0.
\end{align*}
This shows that \eqref{eq_phi_hetero_robust} holds. 

Next, take any $s\in\{2,\ldots,T\}$, and consider again $\overline{p}_{s}=\sum_{k=s}^{T}p_{k}$ from above. We have
\begin{align*}
	&\int_{0}^{+\infty}  \phi_{\theta}(y^T,x^T)\prod_{t=s}^Tf_{\theta}(y_{t}\,|\, y_{t-1},x_{t},a)\mathrm{d}y^{s:T}  \\
    &=\int_{0}^{\infty}\sum_{t=1}^{T-1}\psi_{\theta,t}\left(y_{0},\frac{p_1}{\sum_{l=1}^Tp_{l}},\ldots,\frac{p_t}{\sum_{l=t}^Tp_{l}},\sum_{l=1}^Tp_{l},x^{t}\right)e^{(T-s+1)a}e^{-e^{a}\overline{p}_{s}}\mathrm{d}p^{s:T}\\
    &=\int_0^\infty\int_{(0,1)^{T-s}}\sum_{t=1}^{T-1}\psi_{\theta,t}(y_0,\widetilde p^t,\overline p,x^t)\prod_{j=s}^{T-1}\frac{\Gamma(T-j+1)}{\Gamma(T-j)}(1-\widetilde p_j)^{T-j-1}\\
    &\qquad\times\frac{e^{(T-s+1)a}}{\Gamma(T-s+1)}\overline{p}_{s}^{\,T-s}e^{-e^a\overline{p}_{s}}\,\mathrm d\widetilde p^{s:T-1}\,\mathrm d\overline{p}_{s},
\end{align*}
where the last equality uses the change of variables $p^{s:T}\to(\widetilde p^{s:T-1},\overline{p}_{s})$ and Lemma \ref{lem: helmert_tranformation_MPH}. Recall that $(\widetilde{p}^{s-1},\overline{p})$ is a bijective function of $(p^{s-1},\overline{p}_{s})$; in particular, for $t\leq s-1$ the term $\psi_{\theta,t}(y_0,\widetilde p^t,\overline p,x^t)$ does not depend on $\widetilde p^{s:T-1}$.

We treat the two groups of terms separately. For $t\geq s$, integrating with respect to $\widetilde p_t$ and using the mean-zero property of $\psi_{\theta,t}$ as above gives
\begin{align*}
    \int_0^1\psi_{\theta,t}(y_0,\widetilde p^{t-1},\widetilde p_t,\overline p,x^t)\frac{\Gamma(T-t+1)}{\Gamma(T-t)}(1-\widetilde p_t)^{T-t-1}\,\mathrm d\widetilde p_t=0,
\end{align*}
so every such term vanishes. For $t\leq s-1$, the remaining Beta densities integrate to one, and the term reduces to
\begin{align*}
    \int_0^\infty\psi_{\theta,t}(y_0,\widetilde p^t,\overline p,x^t)\frac{e^{(T-s+1)a}}{\Gamma(T-s+1)}\overline{p}_{s}^{\,T-s}e^{-e^a\overline{p}_{s}}\,\mathrm d\overline{p}_{s},
\end{align*}
which is a function of $(y_0,p^{s-1},x^{t},a)$ only. Collecting terms,
\begin{align*}
	&\int_{0}^{+\infty}  \phi_{\theta}(y^T,x^T)\prod_{t=s}^Tf_{\theta}(y_{t}\,|\, y_{t-1},x_{t},a)\mathrm{d}y^{s:T}\\
    &=\sum_{t=1}^{s-1}\int_0^\infty\psi_{\theta,t}(y_0,\widetilde p^t,\overline p,x^t)\frac{e^{(T-s+1)a}}{\Gamma(T-s+1)}\overline{p}_{s}^{\,T-s}e^{-e^a\overline{p}_{s}}\,\mathrm d\overline{p}_{s},
\end{align*}
which depends on $x^{T}$ only through $x^{s-1}$, since $t\leq s-1$. Hence \eqref{eq_phi_feedback_robust} holds. For $s=T$ the product over $j$ and the $\widetilde{p}$ integration are empty, every $t\leq T-1=s-1$, and the same conclusion applies. This completes the proof of the converse.
        
\subsection{Supplementary lemma} \label{app: supplementary_lemmata}

\begin{lemma} \label{lemma_unique}
Let $T\geq 2$ and suppose that $\varphi_{\theta_{0}}^{1}(Y^T,X^T),\varphi_{\theta_{0}}^{2}(Y^T,X^T)$ are two $L\times 1$ square-integrable moment functions under $(\theta_0,\omega_{0})$ verifying \eqref{eq_varphi} as well as \eqref{eq_aveff} and \eqref{eq_aveff_feedback_robust} of Theorem \ref{theo_charact_aveff}. Then, $\Pi(\varphi_{\theta_0}^{1}(Y^T,X^T)\,|\,\mathcal{T}_{\theta_0,\omega_0,L})=\Pi(\varphi_{\theta_0}^{2}(Y^T,X^T)\,|\,\mathcal{T}_{\theta_0,\omega_0,L})$.
\end{lemma}

\begin{proof}
    By assumption, $\phi_{\theta_{0}}\left(Y^{T},X^{T}\right)=\varphi_{\theta_{0}}^{1}(Y^T,X^T)-\varphi_{\theta_{0}}^{2}(Y^T,X^T)\in \mathcal{T}_{\theta_0,\omega_0,L}^{\perp}$ since it satisfies the conditions \eqref{eq_phi_hetero_robust}-\eqref{eq_phi_feedback_robust} of Theorem \ref{theo_charact}. Then, we can write
    \begin{align*}
        \varphi_{\theta_{0}}^{1}(Y^T,X^T)&=\varphi_{\theta_{0}}^{2}(Y^T,X^T)+\phi_{\theta_{0}}\left(Y^{T},X^{T}\right) \\
        &=\mu(\theta_{0},\omega_{0})+\Pi(\varphi_{\theta_0}^{2}(Y^T,X^T)\,|\,\mathcal{T}_{\theta_0,\omega_0,L})+\Pi(\varphi_{\theta_0}^{2}(Y^T,X^T)\,|\,\mathcal{T}_{\theta_0,\omega_0,L}^{\perp})+\phi_{\theta_{0}}\left(Y^{T},X^{T}\right).
    \end{align*}
    By linearity of $\mathcal{T}_{\theta_0,\omega_0,L}^{\perp}$, $\Pi(\varphi_{\theta_0}^{2}(Y^T,X^T)\,|\,\mathcal{T}_{\theta_0,\omega_0,L}^{\perp})+\phi_{\theta_{0}}\left(Y^{T},X^{T}\right)\in \mathcal{T}_{\theta_0,\omega_0,L}^{\perp}$. Therefore, standard properties of Hilbert spaces imply that $\Pi(\varphi_{\theta_0}^{1}(Y^T,X^T)\,|\,\mathcal{T}_{\theta_0,\omega_0,L})=\Pi(\varphi_{\theta_0}^{2}(Y^T,X^T)\,|\,\mathcal{T}_{\theta_0,\omega_0,L})$.
\end{proof}
 
\section{Nonlinear regression: derivation of equation \eqref{eq_tobeshown}} \label{app: nonlinearregression_details}

In this section, we maintain the following assumptions
\begin{enumerate}[label=(\roman*)]
    \item For every $(y_{t-1},x_t)$, $a\mapsto m_{\beta}(y_{t-1},x_t,a)$ is a continuously differentiable and strictly increasing bijection from $\mathbb{R}$ to $\mathbb{R}$.
    \item For every $(y_0,y_1,x_1)$, $a\mapsto \psi_{\theta}(y_0,y_1,x_1,a)$ is continuous in $a$.
    \item $|\psi_{\theta}\left(Y_{0},Y_{1},X_{1},m_{\beta}^{-1}\left(Y_1,X_2,\tau\right)\right)|\leq B(1+|\tau|^r)$ for some $r>0$ and random variable $B\geq 0$ such that $\mathbb{E}(B^2)<\infty$. 
    \item $\mathbb{E}\left[\abs{m_{\beta}(Y_{1},X_{2},A)}^{2r}\right]<\infty$.
    \item $\int_{-\infty}^{+\infty}\abs{\tau}^{r}k(\tau)d\tau<\infty$.
\end{enumerate}

Let $\psi_{\theta}(y_0,y_1,x_1,a)$ denote a moment function satisfying \eqref{eq_psi_zeromean}. Following Corollary \ref{coro_charact}, we search for functions $\phi_{\theta}$ satisfying equation \eqref{eq_phi_coro}, i.e.,
	\begin{align*}&\int_{-\infty}^{+\infty}\phi_{\theta}(y_0,y_1,y_2,x_1,x_2)\frac{1}{\sqrt{2\pi\sigma^2}}\exp\left(-\frac{1}{2\sigma^2}(y_2-m_{\beta}(y_1,x_2,a))^2\right)dy_2=\psi_{\theta}(y_0,y_1,x_1,a),
    \end{align*}
    which is a convolution equation. Thus, if we let $\tau = m_{\beta}(y_1,x_2,a)$, an application of the Convolution Theorem yields
    \begin{align*}
        \mathcal{F}\left[\phi_{\theta}(y_0,y_1,.,x_1,x_2)\right][s]\exp(-\frac{\sigma^2}{2}s^2)= \mathcal{F}\left[\psi_{\theta}(y_0,y_1,x_1,m_{\beta}^{-1}(y_1,x_2,\tau))\right][s],
    \end{align*}
    where $\mathcal{F}\left[g\right][s]=\int_{-\infty}^{+\infty}g(z)\exp(\boldsymbol{i}zs)dz$ denotes the Fourier transform operator, and we have used that the characteristic function of a $\mathcal{N}(0,\sigma^2)$ is given by $s\mapsto \exp(-\frac{\sigma^2}{2}s^2)$. More explicitly, we have
    \begin{align*}
        &\int_{-\infty}^{+\infty} \phi_{\theta}(y_0,y_1,y_2,x_1,x_2)\exp(\boldsymbol{i}s y_2)dy_2 \\
        &= \int_{-\infty}^{+\infty} \psi_{\theta}(y_0,y_1,x_1,m_{\beta}^{-1}(y_1,x_2,\tau)) \exp(\boldsymbol{i}s \tau+\frac{\sigma^2}{2}s^2)d\tau  \\
        &=\int_{-\infty}^{+\infty} \psi_{\theta}(y_0,y_1,x_1,a)\pdv{m_{\beta}(y_1,x_2,a)}{a} \exp(\boldsymbol{i}s m_{\beta}(y_1,x_2,a)+\frac{\sigma^2}{2}s^2)da,
    \end{align*}
    where we have used the change in variables $a\mapsto \tau=m_{\beta}(y_1,x_2,a)$.\\
However, the inverse problem may not have solutions since the inverse Fourier transform
	\begin{align*}
    \frac{1}{2\pi}\int_{-\infty}^{+\infty} \exp(-\boldsymbol{i}s y_2)\left[\int_{-\infty}^{+\infty} \psi_{\theta}(y_0,y_1,x_1,a)\pdv{m_{\beta}(y_1,x_2,a)}{a} \exp(\boldsymbol{i}s m_{\beta}(y_1,x_2,a)+\frac{\sigma^2}{2}s^2)da \right]ds
    \end{align*}
    may not be well-defined. To address this issue, our strategy is to regularize the problem, and compute the regularized inverse 
	\begin{align*}
        &\phi_{\theta}^{\lambda}(y_0,y_1,y_2,x_1,x_2)\\
		&=\frac{1}{2\pi}\int_{-\infty}^{+\infty} \kappa(\lambda s) \exp(-\boldsymbol{i}s y_2)\left[\int_{-\infty}^{+\infty} \psi_{\theta}(y_0,y_1,x_1,a)\pdv{m_{\beta}(y_1,x_2,a)}{a} \exp(\boldsymbol{i}s m_{\beta}(y_1,x_2,a)+\frac{\sigma^2}{2}s^2)da \right]ds\\
		&=\int_{-\infty}^{+\infty}\psi_{\theta}(y_0,y_1,x_1,a)\pdv{m_{\beta}(y_1,x_2,a)}{a}\left[\frac{1}{2\pi}\int_{-\infty}^{+\infty} \kappa(\lambda s)\exp\left(\boldsymbol{i}s (m_{\beta}(y_1,x_2,a)-y_2)+\frac{1}{2}\sigma^2s^2\right)ds\right]da\\
		&=\int_{-\infty}^{+\infty} \psi_{\theta}(y_0,y_1,x_1,a)\frac{\partial m_{\beta}(y_1,x_2,a)}{\partial a}\frac{1}{\lambda}K_{\lambda}\left(\frac{m_{\beta}(y_1,x_2,a)-y_2}{\lambda}\right)da,
    \end{align*}
	which is well-defined (under suitable conditions so that the integrals can be interchanged). This is similar to the strategy used in kernel nonparametric deconvolution (\citealp{stefanski1990deconvolving}).\\
 The expectation of the regularized inverse is
\begin{align*}
    &\mathbb{E}\left[\phi_{\theta}^{\lambda}(Y_{0},Y_{1},Y_{2},X_{1},X_{2})\right]\\
    &=	\mathbb{E}\left[\int_{-\infty}^{+\infty} \psi_{\theta}(Y_{0},Y_{1},X_{1},a)\frac{\partial m_{\beta}(Y_{1},X_{2},a)}{\partial a}\frac{1}{\lambda}K_{\lambda}\left(\frac{m_{\beta}(Y_{1},X_{2},a)-Y_{2}}{\lambda}\right)da\right]\\
    &=\frac{1}{2\pi}\int_{-\infty}^{+\infty} \kappa(\lambda s)\int_{-\infty}^{+\infty} \mathbb{E}\bigg[\psi_{\theta}(Y_{0},Y_{1},X_{1},a)\frac{\partial m_{\beta}(Y_{1},X_{2},a)}{\partial a}\\
    &\quad \times\exp\left(\lambda \boldsymbol{i}s \frac{m_{\beta}(Y_{1},X_{2},a)-Y_{2}}{\lambda}+\frac{1}{2}\sigma^2s^2\right)\bigg]dads\\
    &=\frac{1}{2\pi}\int_{-\infty}^{+\infty} \kappa(\lambda s) \int_{-\infty}^{+\infty} \mathbb{E}\bigg[\psi_{\theta}(Y_{0},Y_{1},X_{1},a)\frac{\partial m_{\beta}(Y_{1},X_{2},a)}{\partial a}\\
    &\quad \times\exp\left(\lambda \boldsymbol{i}s \frac{m_{\beta}(Y_{1},X_{2},a)-\varepsilon_{2}-m_{\beta}(Y_{1},X_{2},A)}{\lambda}+\frac{1}{2}\sigma^2s^2\right)\bigg]dads  \\
    &= \int_{-\infty}^{+\infty} \mathbb{E}\bigg[\psi_{\theta}(Y_{0},Y_{1},X_{1},a)\frac{\partial m_{\beta}(Y_{1},X_{2},a)}{\partial a}\\
    &\quad \times\frac{1}{2\pi }\int_{-\infty}^{+\infty} \kappa(\lambda s)\exp\left( \boldsymbol{i}s (m_{\beta}(Y_{1},X_{2},a)-m_{\beta}(Y_{1},X_{2},A))\right)ds\bigg]da \displaybreak[0] \\
    &= \int_{-\infty}^{+\infty} \mathbb{E}\bigg[\psi_{\theta}(Y_{0},Y_{1},X_{1},a)\frac{\partial m_{\beta}(Y_{1},X_{2},a)}{\partial a}\\
    &\quad \times\frac{1}{2\pi }\int_{-\infty}^{+\infty} \frac{1}{\lambda}\kappa(u)\exp\left( \boldsymbol{i}\frac{u}{\lambda} (m_{\beta}(Y_{1},X_{2},a)-m_{\beta}(Y_{1},X_{2},A))\right)du\bigg]da.
\end{align*}
Recalling that $k(\tau)=\frac{1}{2\pi}\int_{-\infty}^{+\infty}\kappa(u)\exp\left(-\boldsymbol{i}u\tau\right)\mathrm{d}u$ (using that $k$ is symmetric), the last display reads
\begin{align*}
    \int_{-\infty}^{+\infty} \mathbb{E}\left[\psi_{\theta}(Y_{0},Y_{1},X_{1},a)\frac{\partial m_{\beta}(Y_{1},X_{2},a)}{\partial a}\frac{1}{\lambda}k\left(\frac{m_{\beta}(Y_{1},X_{2},a)-m_{\beta}(Y_{1},X_{2},A)}{\lambda}\right)\right]\mathrm{d}a.
\end{align*}
Now, by (i),  we may change variables to $\tau=\lambda^{-1}\left(m_{\beta}(Y_{1},X_{2},a)-m_{\beta}(Y_{1},X_{2},A)\right)$, so that \begin{align*}
&\mathbb{E}\left[\phi_{\theta}^{\lambda}(Y_{0},Y_{1},Y_{2},X_{1},X_{2})\right]\\&=\mathbb{E}\left[\int_{-\infty}^{+\infty} \psi_{\theta}\left(Y_{0},Y_{1},X_{1},m_{\beta}^{-1}\left(Y_1,X_2,m_{\beta}(Y_{1},X_{2},A)+\lambda \tau\right)\right)k(\tau)\mathrm{d}\tau\right].
\end{align*}
We now note that by (i) and (ii), 
\begin{align*}
    \lim_{\lambda \downarrow 0}\psi_{\theta}\left(Y_{0},Y_{1},X_{1},m_{\beta}^{-1}\left(Y_1,X_2,m_{\beta}(Y_{1},X_{2},A)+\lambda \tau\right)\right) = \psi_{\theta}\left(Y_{0},Y_{1},X_{1},A\right).
\end{align*}
Moreover, by (iii), for any $0<\lambda\leq 1$, by the $c_{r}$-inequality, there exists some constant $c_{r}$ such that
\begin{align*}
    &\abs{\psi_{\theta}\left(Y_{0},Y_{1},X_{1},m_{\beta}^{-1}\left(Y_1,X_2,m_{\beta}(Y_{1},X_{2},A)+\lambda \tau\right)\right)} \\
    &\leq B(1+\abs{m_{\beta}(Y_{1},X_{2},A)+\lambda \tau}^r) \\
    &\leq B(1+c_{r}\abs{m_{\beta}(Y_{1},X_{2},A)}^{r}+c_{r}\lambda^r \abs{\tau}^r) \\
    &\leq B(1+c_{r}\abs{m_{\beta}(Y_{1},X_{2},A)}^{r}+c_{r} \abs{\tau}^r),
\end{align*}
and 
\begin{align*}
    &\mathbb{E}\left[\int_{-\infty}^{+\infty} B(1+c_{r}\abs{m_{\beta}(Y_{1},X_{2},A)}^{r}+c_{r} \abs{\tau}^r)k(\tau)\mathrm{d}\tau\right] \\
    &=\mathbb{E}\left[B\right]+c_{r}\mathbb{E}\left[B\abs{m_{\beta}(Y_{1},X_{2},A)}^{r}\right]+c_{r}\mathbb{E}\left[B\right]\int_{-\infty}^{+\infty} \abs{\tau}^rk(\tau)\mathrm{d}\tau \\
    &\leq \mathbb{E}\left[B\right]\left(1+c_{r}\int_{-\infty}^{+\infty} \abs{\tau}^rk(\tau)\mathrm{d}\tau\right)+c_{r}\mathbb{E}\left[B^2\right]^{\frac{1}{2}}\mathbb{E}\left[\abs{m_{\beta}(Y_{1},X_{2},A)}^{2r}\right]^{\frac{1}{2}} \\
    &<\infty,
\end{align*}
by (iii),(iv),(v) and the Cauchy-Schwarz inequality. By dominated convergence, and the law of iterated expectations, we conclude that 
\begin{align*}
   & \lim_{\lambda \downarrow 0}\mathbb{E}\left[\int_{-\infty}^{+\infty} \psi_{\theta}\left(Y_{0},Y_{1},X_{1},m_{\beta}^{-1}\left(Y_1,X_2,m_{\beta}(Y_{1},X_{2},A)+\lambda \tau\right)\right)k(\tau)\mathrm{d}\tau\right]
   \\&=\mathbb{E}\left[\psi_{\theta}\left(Y_{0},Y_{1},X_{1},A\right)\right] \\
   &=\mathbb{E}\left[\mathbb{E}\left[\psi_{\theta}\left(Y_{0},Y_{1},X_{1},A\right)\,|\, Y_{0},X_{1}\right]\right] \\
   &=0.
\end{align*}

\section{Identification of the Weibull MPH model}
Theorem \ref{theo_identif_weibull_MPH}, presented below, formally establishes point-identification of the common parameter of the Weibull MPH model, as discussed in Example \ref{ex: mph_Weibull_identification} of Section \ref{sec: SEBs}.

\begin{theorem} \label{theo_identif_weibull_MPH}
    Consider the MPH model defined in Example \ref{ex: mph_intro} with $T=2$ and a Weibull baseline hazard, $\lambda_{\alpha}(y_t)=\alpha y_t^{\alpha-1}$. Let $\theta=(\alpha,\beta',\gamma)'\in \Theta$ denote the common parameter where $\Theta = \Theta_{\alpha}\times \Theta_{\beta}\times \Theta_{\gamma}$ with $ \Theta_{\beta}\times \Theta_{\gamma}\subset \mathbb{R}^{K-1}$ and $\Theta_{\alpha}=[\underline{\alpha},\overline{\alpha}]$ with $0<\underline{\alpha}<\overline{\alpha}<+\infty$. Let $W_{t}=(Y_{t-1},X_t')'$ for $t=1,2$ and  let $\Delta$ denote the first difference operator. Suppose that 
    \begin{enumerate}[label=\roman*)]
        \item $\mathbb{E}\left[\norm{W_1}\left(\abs{\Delta \ln Y_2}+\norm{\Delta W_2}\right)\right]<\infty$
        \item $\mathbb{E}\left[W_1\Delta W_2'\right]$ is nonsingular
        \item $\mathbb{E}\left[\abs{\ln Y_2}+\abs{\ln Y_1}+\norm{W_1}+\norm{W_2}\right]<\infty$
          \end{enumerate}
Let
\begin{align*}
    \phi_{\theta}^{a}(Y^2,X^2)=\ln \left(1-\widetilde{P}_{\theta,1}\right)-\ln \widetilde{P}_{\theta,1}, \quad \phi_{\theta}^{b}(Y^2,X^2)=2+\ln \left(1-\widetilde{P}_{\theta,1}\right)+\ln \widetilde{P}_{\theta,1}
\end{align*}
and
\begin{align*}
    \phi_{\theta}(Y^2,X^2)=\left(W_1\phi_{\theta}^{a}(Y^2,X^2), \phi_{\theta}^{b}(Y^2,X^2)\right)'
\end{align*}
Then, $\theta_{0}=(\alpha_0,\beta_0',\gamma_0)'$ is the unique solution to the moment equation $\mathbb{E}\left[\phi_{\theta}(Y^2,X^2)\right]=0$. Therefore, $\theta_{0}$ is identified.
\end{theorem}

\begin{proof}
     We establish the result by showing that $\theta_{0}$ is the unique solution to the normalized moment equation 
    \begin{align} \label{WeibullMPH_normalized_FHR_moment}
        \mathbb{E}\left[\check{\phi}_{\theta}(Y^2,X^2)\right]=0,
    \end{align}
    where $\check{\phi}_{\theta}(Y^2,X^2)=\frac{1}{\alpha}\phi_{\theta}(Y^2,X^2)$. Since $\alpha\in \Theta_{\alpha}$ is a strictly positive scalar, $\mathbb{E}\left[\phi_{\theta}(Y^2,X^2)\right]=0$ and $\mathbb{E}\left[\check{\phi}_{\theta}(Y^2,X^2)\right]=0$ share the same solution set.    
    
     Let $\tau=(\gamma,\beta')'$. We start by proving that $\frac{\tau_0}{\alpha_0}$ is identified by the first block of equations in \eqref{WeibullMPH_normalized_FHR_moment}. To see this, consider the moment function 
 \begin{align*}
   \check{\phi}_{\theta}^{a}(Y^2,X^2)&= \frac{1}{\alpha}\phi_{\theta}^{a}(Y^2,X^2) \\
   &=\frac{1}{\alpha}\ln \left(1-\widetilde{P}_{\theta,1}\right)-\frac{1}{\alpha}\ln \widetilde{P}_{\theta,1} \\
    &=(\ln Y_2-\ln Y_1)+(X_2-X_1)'\frac{\beta}{\alpha}+(Y_1-Y_0)\frac{\gamma}{\alpha} \\
    &=\Delta \ln Y_2+\Delta W_2'\frac{\tau}{\alpha}.
\end{align*}
By Lemma \ref{lem: helmert_tranformation_MPH} and Lemma \ref{lem_moment_charac_MPH}, we have $\mathbb{E}\left[\check{\phi}_{\theta_0}^{a}(Y^2,X^2)|Y_{0},\overline{P}_{\theta_0},X_1\right]=0$. Thus, after instrumenting $\check{\phi}_{\theta_0}^{a}$ with $W_1=(Y_{0},X_1')'$ and exploiting i)-ii), we recover
\begin{align*} 
    \frac{\tau_0}{\alpha_0}=-\mathbb{E}\left[W_1\Delta W_2'\right]^{-1}\mathbb{E}\left[W_1\Delta \ln Y_2\right].
\end{align*}
By the same logic, any candidate $\theta$ fulfilling \eqref{WeibullMPH_normalized_FHR_moment}  must verify
\begin{align} \label{WeibullMPH_identification_tauoveralpha}
    \frac{\tau}{\alpha} = -\mathbb{E}\left[W_1\Delta W_2'\right]^{-1}\mathbb{E}\left[W_1\Delta \ln Y_2\right] = \frac{\tau_0}{\alpha_0}.
\end{align}
Given \eqref{WeibullMPH_identification_tauoveralpha}, we now show that $\alpha_{0}$ is identified by the last equation in \eqref{WeibullMPH_normalized_FHR_moment}. To this end, consider the moment function
\begin{align*}
    \check{\phi}_{\theta}^{b}(Y^2,X^2) &= \frac{1}{\alpha}\phi_{\theta}^{b}(Y^2,X^2)\\
    &=\frac{1}{\alpha}\left(2+\ln \left(1-\widetilde{P}_{\theta,1}\right)+\ln \widetilde{P}_{\theta,1}\right) \\
    &=\frac{1}{\alpha}\left(\ln P_{\theta,1}+\ln P_{\theta,2}-2\ln\left(P_{\theta,1}+P_{\theta,2}\right)+2\right) \\
    &= \ln Y_1+ W_1'\frac{\tau}{\alpha}+ \ln Y_2+W_2'\frac{\tau}{\alpha}-\frac{2}{\alpha}\ln\left(Y_1^\alpha e^{\alpha W_1'\frac{\tau}{\alpha}}+Y_2^\alpha e^{\alpha W_2'\frac{\tau}{\alpha}}\right)+\frac{2}{\alpha}. 
\end{align*}
The fact that $\mathbb{E}\left[\check{\phi}_{\theta_0}^{b}(Y^2,X^2)|Y_{0},\overline{P}_{\theta_0},X_1\right]=0$ follows again from Lemma \ref{lem: helmert_tranformation_MPH} and Lemma \ref{lem_moment_charac_MPH}. Because our interest lies in the roots of \eqref{WeibullMPH_normalized_FHR_moment}, for which \eqref{WeibullMPH_identification_tauoveralpha} holds, we may just focus on 
\begin{align*}
    \check{\phi}_{\alpha}^{b}(Y^2,X^2)  &=\ln Y_1+ W_1'\frac{\tau_0}{\alpha_0}+ \ln Y_2+W_2'\frac{\tau_0}{\alpha_0}-\frac{2}{\alpha}\ln\left(\frac{\left(Y_1e^{W_1'\frac{\tau_0}{\alpha_0}}\right)^{\alpha}+\left(Y_2e^{W_2'\frac{\tau_0}{\alpha_0}}\right)^{\alpha}}{2}\right)+\frac{2}{\alpha}(1-\ln 2).
\end{align*}
 Let $Q(\alpha)=\mathbb{E}\left[ \check{\phi}_{\alpha}^{b}(Y^2,X^2)\right]$. Since $\min(z_1,z_2) \leq \left(\frac{z_1^{\alpha}+z_2^{\alpha}}{2}\right)^{\frac{1}{\alpha}}\leq \max(z_1,z_2)$ for any $z_1>0,z_2>0$, we have $\abs{\frac{1}{\alpha}\ln\left(\frac{z_1^{\alpha}+z_2^{\alpha}}{2}\right)} \leq \max(\abs{\ln z_1},\abs{\ln z_2}) \leq  \abs{\ln z_1}+\abs{\ln z_2}$, which by the Cauchy-Schwarz inequality implies that
\begin{align*}
    \abs{\check{\phi}_{\alpha}^{b}(Y^2,X^2)}\leq 3\left(1+\norm{\frac{\tau_0}{\alpha_0}}\right)\left(\abs{\ln Y_2}+\abs{\ln Y_1}+\norm{W_1}+\norm{W_2}\right) + \frac{2}{\underline{\alpha}}(1-\ln 2),
\end{align*}
ensuring that $Q(\alpha)$ is well defined over $\Theta_{\alpha}$ by iii). Now, note that the power-mean function $\alpha \mapsto \left(\frac{z_1^{\alpha}+z_2^{\alpha}}{2}\right)^{\frac{1}{\alpha}}$ is weakly increasing in $\alpha$ and since the log function is strictly increasing, the composition $\alpha \mapsto\frac{1}{\alpha}\ln\left(\frac{z_1^{\alpha}+z_2^{\alpha}}{2}\right)$ is weakly increasing.  Moreover, since $\ln2<1$, $\alpha \mapsto \frac{2}{\alpha}(1-\ln 2)$ is strictly decreasing. As a result, $Q(\alpha)$ is a strictly decreasing function in $\alpha$. Together with the validity of the moment function $\check{\phi}_{\theta}^{b}(Y^2,X^2)$ at the true parameter $\theta_{0}$, which implies $Q(\alpha_0)=0$, the monotonicity of $Q(\alpha)$ ensures that $\alpha_0$ is the unique solution to the equation
\begin{align*} 
    Q(\alpha)=0.
\end{align*}
This is precisely the last component of \eqref{WeibullMPH_normalized_FHR_moment} after imposing \eqref{WeibullMPH_identification_tauoveralpha}.
Taking stock, any candidate $\theta$ satisfying \eqref{WeibullMPH_normalized_FHR_moment} must verify $\frac{\tau}{\alpha} = \frac{\tau_0}{\alpha_0}$ by \eqref{WeibullMPH_identification_tauoveralpha} and $\alpha=\alpha_{0}$. It follows that $\tau=\tau_{0}$ and hence $\theta=\theta_{0}$. To conclude, since 
\eqref{WeibullMPH_normalized_FHR_moment} and  $\mathbb{E}\left[\phi_{\theta}(Y^2,X^2)\right]=0$ share the same solution set, the original unnormalized moment equation also admits $\theta_{0}$ as its unique solution.
\end{proof}

\section{Identified sets} \label{app: FHR_outerset}

This section draws a link between our characterization of FHR moment conditions in Theorem \ref{theo_charact} and the identified set for $\theta$. For ease of exposition, we focus on the case $T=2$, but the characterization holds for general $T$. Recall the notations $\theta_0$ and $\omega_{0}=(g_0,\pi_0,\nu_0)$ for the true parameter values characterizing the sampled population. The set of allowable feedback processes $\mathcal{G}$, heterogeneity distributions ${\cal{P}}$, and initial condition densities $\mathcal{N}$ introduced in the main text are formally given by
\begin{align*}
{\cal{P}} =
&\Bigg\{ \pi=\left(
\pi(\cdot\mid y_0,x_1)
\right)_{(y_0,x_1)\in\mathcal Y\times\mathcal X}
\,|\, \text{ for all } (y_0,x_1)\in\mathcal Y\times\mathcal X: \pi(\cdot\mid y_0,x_1)\in\mathcal D(\mathcal A)
\\
&\text{and } (a,y_0,x_1)
\mapsto
\pi(a\mid y_0,x_1)
\text{ is Borel measurable}
\Bigg\}, \\
\mathcal G
=
&\Bigg\{
g
=
\left(
g(\cdot\mid y_1,y_0,x_1,a)
\right)_{(y_1,y_0,x_1,a)\in\mathcal Y^2\times\mathcal X\times\mathcal A}
\,\Big|\,
\text{ for all }(y_1,y_0,x_1,a)\in
\mathcal Y^2\times\mathcal X\times\mathcal A:
\\
&\hspace{4.2cm}
g(\cdot\mid y_1,y_0,x_1,a)\in\mathcal D(\mathcal X)
\\
&\text{and }(x_2,y_1,y_0,x_1,a)
\mapsto
g(x_2\mid y_1,y_0,x_1,a)
\text{ is Borel measurable}
\Bigg\},
\end{align*}
and $\mathcal{N}= \mathcal{D}(\mathcal{Y}\times \mathcal{X})$, where for any Lebesgue measurable set $\mathcal S$
\begin{align*}
\mathcal D(\mathcal S)
=
\left\{
h\in L^1(\mathcal S):
h(s)\geq 0
\text{ for almost every }s\in\mathcal S,
\quad
\int_{\mathcal S} h(s)\,\mathrm{d}s=1
\right\}.
\end{align*}
The set of possible nuisance parameters is then $\Omega = {\cal{P}}\times \mathcal{G}\times \mathcal{N}$. \\

\noindent \textbf{Semiparametric identified set.} With these preliminaries in place, we define the identified set $\Theta^{\rm I}$ for $\theta_0$ as
\begin{align*}
    \Theta^{\rm I}
    &=
    \Bigg\{
    \theta\in\Theta |
    \exists \, (\pi,g,\nu)\in \Omega \text{ such that for all }  (y_2,y_1,y_0,x_2,x_1)\in
\mathcal Y^3\times\mathcal X^2, \\
   &\ell\left(\left.\theta_0,\omega_0\right|y^{2},x^{2}\right)
    =
    \int_{\mathcal{A}}
    f_{\theta}(y_2\mid y_1,x_2,a)
    g(x_2\mid y_1,y_0,x_1,a)
    f_{\theta}(y_1\mid y_0,x_1,a)
    \pi(a\mid y_0,x_1)
    \nu\left(y_{0},x_{1}\right) \mathrm{d}a 
    \Bigg\}.
\end{align*}
In words, this is the collection of $\theta$ for which there exists $\omega=\left(g,\pi,\nu\right)$ such that the induced likelihood $\ell(\theta,{\omega}\,|\, Y^2,X^2)$  is consistent with the observed data $\ell(\theta_0,{\omega}_0\,|\, Y^2,X^2)$. \\
\noindent Consider the following linear operators 
\begin{align*}
    &\left(\mathcal{C}_{1,\theta} \psi\right)(y^2,x^2) = \int_{\mathcal{A}} f_{\theta}(y_2\mid y_1,x_2,a)\,\psi(y^1,x^2,a)\, \mathrm{d}a \\
    &\left(\mathcal{C}_{2,\theta} \psi\right)(y^1,x_1,a) = f_{\theta}(y_1\mid y_0,x_1,a)\int_{\mathcal{Y}}\int_{\mathcal{X}}\psi(y_0,\widetilde y_1,x^2,a)\, \mathrm{d}x_2\,\mathrm{d}\widetilde y_1-\int_{\mathcal{X}}\psi(y^1,x^2,a)\,\mathrm{d}x_2
 \\
&\left(\mathcal{C}_{3,\theta} \psi\right) = \int_{\mathcal{A}}\int_{\mathcal{X}^2}\int_{\mathcal{Y}^2}\psi(y^1,x^2,a)\,\mathrm{d}x^2\,\mathrm{d}y^1\,\mathrm{d}a
\end{align*}
over functions $\psi \in L^{1}(\mathcal{Y}^2\times \mathcal{X}^2\times \mathcal{A})$.  Building on \citet{honore2006bounds}, \citet{bonhomme2023identification} show that $\Theta^{\rm I}$ may be recast as the set of $\theta$ values for which there exists $\psi$ such that
\begin{align}
&\ell\left(\left.\theta_0,\omega_0\right|y^{2},x^{2}\right)
= \left(\mathcal{C}_{1,\theta} \psi\right)(y^2,x^2)
\label{eq:main} \\
&\left(\mathcal{C}_{2,\theta} \psi\right)(y^1,x_1,a)=0
\label{eq:restr1}\\
&\left(\mathcal{C}_{3,\theta} \psi\right)
=1
\label{eq:normal}\\
&\psi(y^1,x^2,a) \geq 0
\label{eq:nonneg}
\end{align}
Constraints \eqref{eq:normal} and \eqref{eq:nonneg} require that $\psi$ be a valid probability distribution. The constraints \eqref{eq:restr1} and \eqref{eq:main} ensure that $\psi$ are consistent with the period-1 and period-2 likelihood of the panel data model respectively. The value of this formulation is that it shows that checking whether $\theta$ belongs to the identified set amounts to checking the constraints of a linear program.  This makes the computation of the identified set straightforward when the outcome variables, the covariates, and the heterogeneity have finite support (see \citealp{honore2006bounds}, \citealp{bonhomme2023identification}).  \\

\noindent \textbf{FHR outer identification region.} Now consider the identified set based solely on the FHR moments
\begin{align}
    \Theta^{\rm FHR}=\{\theta\in \Theta \,|\, \forall \phi_{\theta} \in \mathcal{F}_{\theta,0}^{\rm FHR}: \mathbb{E}_{\theta_0,\omega_0}\left[\phi_{\theta}\left(Y_{0},Y_{1},\ldots,Y_{T},X_{1},\ldots,X_{T}\right)\right]=0 \}, \label{def_of_outer_set}
\end{align}
where $$\mathcal{F}_{\theta,0}^{\rm FHR}=\{\phi\in L^{\infty}(\mathcal{Y}^3\times \mathcal{X}^2) \,|\, \eqref{eq_phi_hetero_robust} \text{ and }\eqref{eq_phi_feedback_robust} \text{ hold} \}$$ denotes the family of bounded measurable FHR test functions for a given value of $\theta$. The consideration of this set is natural in our framework: $\Theta^{\rm FHR}$ collects the distinct $\theta$ values that are roots of the FHR moment restrictions. Recall that the main text emphasizes the special case of point-identification, corresponding to $\Theta^{\rm FHR}=\{\theta_{0}\}$. Corollary \ref{cor:Theta_FHR} below shows the anatomy of $\Theta^{\rm FHR}$ and clarifies its identifying content, irrespective of whether $\Theta^{\rm FHR}$ is a singleton. \\
To introduce the result, consider the linear operator $T_{\theta}\psi = (\mathcal{C}_{1,\theta}\psi, \mathcal{C}_{2,\theta}\psi, \mathcal{C}_{3,\theta}\psi)
$ and let  
\begin{align*}
    \widetilde{\Theta}^{\rm I}=\{\theta\in \Theta \,|\, (\ell,0,1)\in  \overline{\mathcal{R}(T_{\theta})}\}
\end{align*}
where $\overline{\mathcal{R}(T_{\theta})}$ denotes the closure of the range of $T_{\theta}$. For a given $\theta$, the set $\widetilde{\Theta}^{\rm I}$ asks whether the model constraints \eqref{eq:main},\eqref{eq:restr1},\eqref{eq:normal} are satisfied but it ignores the positivity constraint \eqref{eq:nonneg}. We have the following result. 
\begin{corollary} \label{cor:Theta_FHR}
$\Theta^{\rm FHR} = \widetilde{\Theta}^{\rm I}$
\end{corollary}
\begin{proof}
    This follows immediately from Lemma \ref{lem:Theta_FHR} and Lemma \ref{lem:FHR_sets_equiv}. 
\end{proof}

 Corollary \ref{cor:Theta_FHR} shows that the main distinction between $\Theta^{\rm I}$ and $\Theta^{\rm FHR}$ is that the latter relaxes the positivity constraint on the nuisance parameter $\omega=\left(g,\pi,\nu\right)$. That is, $\Theta^{\rm FHR}$ allows the feedback process, the heterogeneity distribution, and the initial condition to be signed measures rather than probability distributions. This explains the absence of moment inequalities in the definition of $\Theta^{\rm FHR}$ and why $\Theta^{\rm I}\subset \Theta^{\rm FHR}$. Another (technical) distinction between the two sets is that $\Theta^{\rm FHR}$ is closed while $\Theta^{\rm I}$ may not be.  \\

To present the auxiliary results underpinning Corollary \ref{cor:Theta_FHR}, let
\begin{align*}
    \mathcal{F}_{\theta}^{\rm FHR}=\{&\phi\in L^{\infty}(\mathcal{Y}^3\times \mathcal{X}^2) \,|\,  \int_{\mathcal{Y}}\phi(y^2,x^2)f_{\theta}(y_2\mid y_1,x_2,a)\mathrm{d}y_2 \text{ does not depend on } x_2 \text{ a.s } \\
    &  \int_{\mathcal{Y}} \int_{\mathcal{Y}}\phi(y^2,x^2)f_{\theta}(y_2\mid y_1,x_2,a) f_{\theta}(y_1\mid y_0,x_1,a)\mathrm{d}y_1\mathrm{d}y_2 \text{ is constant in }(y_0,x_1,a) \text{ a.s }
    \},
\end{align*}
which corresponds to an unnormalized analog of $\mathcal{F}_{\theta,0}^{\rm FHR}$. 

\begin{lemma} \label{lem:Theta_FHR}
$\theta \in  \widetilde{\Theta}^{\rm I}$ if and only if for all $\phi \in \mathcal{F}_{\theta}^{\rm FHR}$,
\begin{align*}
    \mathbb{E}\left[\phi(Y^2,X^2)\right] = \int_{\mathcal{Y}}\int_{\mathcal{Y}}\phi(y^2,x^2)f_{\theta}(y_2\mid y_1,x_2,a)f_{\theta}(y_1\mid y_0,x_1,a)\mathrm{d}y_1\mathrm{d}y_2.
\end{align*}
\end{lemma}
\begin{proof}
      We prove the contrapositive. Since $\overline{\mathcal{R}(T_{\theta})}$ is a closed linear subspace of $L^{1}(\mathcal{Y}^3\times \mathcal{X}^2) \times L^{1}(\mathcal{Y}^2\times \mathcal{X} \times \mathcal{A})\times \mathbb{R}$, by the separating hyperplane theorem, $(\ell,0,1)\notin  \overline{\mathcal{R}(T_{\theta})}$ is equivalent to the existence of a triple   $(\phi,\mu,\lambda)$ in the dual space $ L^{\infty}(\mathcal{Y}^3\times \mathcal{X}^2) \times L^{\infty}(\mathcal{Y}^2\times \mathcal{X} \times \mathcal{A})\times \mathbb{R}$ such that
\begin{align*}
    \mathscr{L}_{\phi,\mu,\lambda}(\ell,0,1)>c \text{ and } \mathscr{L}_{\phi,\mu,\lambda}(v)\leq c \quad \forall v\in \overline{\mathcal{R}(T_{\theta})}
\end{align*}    
for some separating constant $c\in \mathbb{R}$ and the linear functional $\mathscr{L}_{\phi,\mu,\lambda}$ given by
\begin{align} \label{lem:Theta_FHR:separating_functional}
    \mathscr{L}_{\phi,\mu,\lambda}(v)=\int \phi(y^2,x^2) v_{1}(y^2,x^2)\mathrm{d}y^2\mathrm{d}x^2 + \int \mu(y^1,x_1,a) v_2(y^1,x_1,a)\mathrm{d}y^1\mathrm{d}x_1\mathrm{d}a+ \lambda v_3.
\end{align}
Given the linear subspace structure, we can set $c$ to zero. To see this, note that since $0\in \mathcal{R}(T_{\theta})$, $\mathscr{L}_{\phi,\mu,\lambda}(0)=0\leq c$. Moreover, if there exists $v\in \overline{\mathcal{R}(T_{\theta})}$ such that $\mathscr{L}_{\phi,\mu,\lambda}(v)\neq  0$, then for any $t>0 $, because $\text{sgn}(\mathscr{L}_{\phi,\mu,\lambda}(v)) t v\in \overline{\mathcal{R}(T_{\theta})}$, we would obtain $\mathscr{L}_{\phi,\mu,\lambda}(\text{sgn}(\mathscr{L}_{\phi,\mu,\lambda}(v)) t v)= t\abs{\mathscr{L}_{\phi,\mu,\lambda}(v)} \xrightarrow[t\to\infty]{}+\infty$, contradicting the upper bound by $c$. Thus, $\mathscr{L}_{\phi,\mu,\lambda}$ vanishes on $\overline{\mathcal{R}(T_{\theta})}$ or equivalently on $\mathcal{R}(T_{\theta})$ by continuity of $\mathscr{L}_{\phi,\mu,\lambda}$. \\
We now observe that 
\begin{align*}
    &\mathscr{L}_{\phi,\mu,\lambda}(T_{\theta}\psi)\\
    &=\int  \phi(y^2,x^2) \left(\mathcal{C}_{1,\theta} \psi\right)(y^2,x^2)\mathrm{d}y^2\mathrm{d}x^2 +\int \mu(y^1,x_1,a) \left(\mathcal{C}_{2,\theta} \psi\right)(y^1,x_1,a)\mathrm{d}y^1\mathrm{d}x_1\mathrm{d}a + \lambda \left(\mathcal{C}_{3,\theta} \psi\right) \\
    &=\int\limits_{\mathcal{Y}^2\times\mathcal{X}^2} \phi(y^2,x^2)\int_{\mathcal{A}} f_{\theta}(y_2\mid y_1,x_2,a)\,\psi(y^1,x^2,a)\, \mathrm{d}a\mathrm{d}y^2\mathrm{d}x^2 \\
        &+\int\limits_{\mathcal{Y}^2\times\mathcal{X}\times\mathcal{A}} \mu(y^1,x_1,a) \left(f_{\theta}(y_1\mid y_0,x_1,a)\int\limits_{\mathcal{Y}\times\mathcal{X}}\psi(y_0,\widetilde{y_1},x_1,\widetilde{x_2},a)\, d\widetilde{x_2}\,d\widetilde{y_1}-\int_{\mathcal{X}}\psi(y^1,x_1,\widetilde{x_2},a)\,d\widetilde{x_2}\right)\mathrm{d}y^1\mathrm{d}x_1\mathrm{d}a  \\
    &+\lambda \int\psi(y^1,x^2,a)\,dx^2\,dy^1\,da \\
    &=\int \left(R_{1,\theta}(x^2,y^1,a) + R_{2,\theta}(y_0,x_1,a)-\mu(y^1,x_1,a)+\lambda\right)\psi(y^1,x^2,a) \mathrm{d}x^2\mathrm{d}y^1\mathrm{d}a,
\end{align*}
where the last equality follows from rearranging integrals and Fubini's theorem, with 
\begin{align*}
        &R_{1,\theta}(x^2,y^1,a)=\int_{\mathcal{Y}}\phi(y^2,x^2)f_{\theta}(y_2\mid y_1,x_2,a)\mathrm{d}y_2 \\
    &R_{2,\theta}(y_0,x_1,a)= \int_{\mathcal{Y}}\mu(y^1,x_1,a)f_{\theta}(y_1\mid y_0,x_1,a)\mathrm{d}y_1.
\end{align*}
By $L^{1}-L^{\infty}$ duality, $\mathscr{L}_{\phi,\mu,\lambda}$ vanishes on $\mathcal{R}(T_{\theta})$ if and only if 
\begin{align*}
    R_{1,\theta}(x^2,y^1,a) + R_{2,\theta}(y_0,x_1,a)-\mu(y^1,x_1,a)+\lambda = 0 \quad\text{a.s}
\end{align*}
 This implies that $R_{1,\theta}(x^2,y^1,a)$ does not depend on $x_2$ as $R_{1,\theta}(x^2,y^1,a) = -R_{2,\theta}(y_0,x_1,a)+\mu(y^1,x_1,a)-\lambda$. Furthermore, integrating this equality against the period-1 conditional density yields
\begin{align*}
    \int_{\mathcal{Y}} R_{1,\theta}(x^2,y^1,a)f_{\theta}(y_1|y_{0},x_1,a)\mathrm{d}y_1=-\lambda,
\end{align*}
implying that $\int_{\mathcal{Y}} R_{1,\theta}(x^2,y^1,a)f_{\theta}(y_1|y_{0},x_1,a)\mathrm{d}y_1$ does not depend on $(y_0,x_1,a)$. Consequently, the function $\phi$ underlying the separating functional $\mathscr{L}_{\phi,\mu,\lambda}$ belongs to the set $\mathcal{F}_{\theta}^{\rm FHR}$. Next, $\mathscr{L}_{\phi,\mu,\lambda}(\ell,0,1) = \int \phi(y^2,x^2) \ell\left(\left.\theta_0,\omega_0\right|y^{2},x^{2}\right) dy^2dx^2 + \lambda$. Thus, the condition $\mathscr{L}_{\phi,\mu,\lambda}(\ell,0,1)>0$ is equivalent to 
\begin{align*}
     \mathbb{E}\left[\phi(Y^2,X^2)\right] > -\lambda =  \int_{\mathcal{Y}} R_{1,\theta}(x^2,y^1,a)f_{\theta}(y_1|y_{0},x_1,a)\mathrm{d}y_1.
\end{align*}
We have thus obtained that
\begin{align*}
    (\ell,0,1)\notin  \overline{\mathcal{R}(T_{\theta})}&\implies \exists \phi \in \mathcal{F}_{\theta}^{\rm FHR} \text{ such that }   \\
    &\mathbb{E}\left[\phi(Y^2,X^2)\right] \neq \int_{\mathcal{Y}}\int_{\mathcal{Y}}\phi(y^2,x^2)f_{\theta}(y_2\mid y_1,x_2,a)f_{\theta}(y_1\mid y_0,x_1,a)\mathrm{d}y_1\mathrm{d}y_2. 
\end{align*}
We now prove the converse. To this end, suppose that there exists $\phi \in \mathcal{F}_{\theta}^{\rm FHR}$ satisfying $\mathbb{E}\left[\phi(Y^2,X^2)\right]\neq I_{\theta,\phi}$
for $I_{\theta,\phi}=\int_{\mathcal{Y}}\int_{\mathcal{Y}}\phi(y^2,x^2)f_{\theta}(y_2\mid y_1,x_2,a)f_{\theta}(y_1\mid y_0,x_1,a)\mathrm{d}y_1\mathrm{d}y_2$ . Note that $I_{\theta,\phi}$ is constant in $(y_0,x_1,a)$ by definition of the set $\mathcal{F}_{\theta}^{\rm FHR}$. Since $\mathcal{F}_{\theta}^{\rm FHR}$ is a linear subspace, we may without loss of generality assume that
\begin{align*}
    \mathbb{E}\left[\phi(Y^2,X^2)\right]> I_{\theta,\phi}
\end{align*}
by changing $\phi$ into $-\phi$ if necessary. Since $\phi\in \mathcal{F}_{\theta}^{\rm FHR}$, the associated integral quantity $R_{1,\theta}(x^2,y^1,a)$ does not
depend on $x_2$. Let $\mu(y^1,x_1,a)=R_{1,\theta}(x^2,y^1,a)$ and $\lambda=-I_{\theta,\phi}$. Then
\begin{align*}
        R_{2,\theta}(y_0,x_1,a)= \int_{\mathcal{Y}}\mu(y^1,x_1,a)f_{\theta}(y_1\mid y_0,x_1,a)\mathrm{d}y_1
    =
    I_{\theta,\phi},
\end{align*}
and hence
\begin{align*}
   R_{1,\theta}(x^2,y^1,a) + R_{2,\theta}(y_0,x_1,a)-\mu(y^1,x_1,a)+\lambda
    =
    R_{1,\theta}(x^2,y^1,a)+I_{\theta,\phi}-R_{1,\theta}(x^2,y^1,a)-I_{\theta,\phi}
    =
    0.
\end{align*}
This implies that the linear functional $\mathscr{L}_{\phi,\mu,\lambda}$ induced by this particular triple $(\phi,\mu,\lambda)$ vanishes on $\mathcal{R}(T_{\theta})$ and by continuity on $\overline{\mathcal{R}(T_{\theta})}$. On the other hand, $\mathscr{L}_{\phi,\mu,\lambda}(\ell,0,1)=\mathbb{E}\left[\phi(Y^2,X^2)\right]-I_{\theta,\phi}>0$. Thus, $(\ell,0,1)\notin \overline{\mathcal{R}(T_{\theta})}$ which concludes the proof.
\end{proof}

\noindent The proof technique in Lemma \ref{lem:Theta_FHR} can be adapted to derive a dual characterization of the semiparametric identified set $\Theta^{\rm I}$ as well as the identified set for average effects, although we do not pursue this route in this paper.

\begin{lemma} \label{lem:FHR_sets_equiv}
The following two conditions are equivalent
\begin{enumerate}[label=(\Alph*)]
    \item $\forall \phi \in \mathcal{F}_{\theta}^{\rm FHR}$: $\mathbb{E}\left[\phi(Y^2,X^2)\right] = \int_{\mathcal{Y}}\int_{\mathcal{Y}}\phi(y^2,x^2)f_{\theta}(y_2\mid y_1,x_2,a)f_{\theta}(y_1\mid y_0,x_1,a)\mathrm{d}y_1\mathrm{d}y_2
    $
    \item $\forall \phi\in \mathcal{F}_{\theta,0}^{\rm FHR}$: $\mathbb{E}\left[\phi(Y^2,X^2)\right] = 0$
\end{enumerate}
\end{lemma}
\begin{proof}
    Clearly, $(A) \implies (B)$ since $\mathcal{F}_{\theta,0}^{\rm FHR}\subseteq \mathcal{F}_{\theta}^{\rm FHR}$. Conversely, suppose that (B) holds and fix a $\phi \in \mathcal{F}_{\theta}^{\rm FHR}$. If we define
    \begin{align*}
        \widetilde{\phi}(y^2,x^2) = \phi(y^2,x^2)-\int_{\mathcal{Y}} \int_{\mathcal{Y}}\phi(\widetilde y^2,x^2)f_{\theta}(\widetilde y_2\mid \widetilde y_1,x_2,a) f_{\theta}(\widetilde y_1\mid y_0,x_1,a)\mathrm{d}\widetilde y_1\mathrm{d}\widetilde y_2,
    \end{align*}
    then $\widetilde{\phi}(y^2,x^2)\in \mathcal{F}_{\theta,0}^{\rm FHR}$ by construction. Thus, $\mathbb{E}\left[\widetilde{\phi}(Y^2,X^2)\right] = 0$ which delivers (A): 
    \begin{align*}
    \mathbb{E}\left[\phi(Y^2,X^2)\right] 
    &= \int_{\mathcal{Y}}\int_{\mathcal{Y}}\phi(y^2,x^2)f_{\theta}(y_2\mid y_1,x_2,a)f_{\theta}(y_1\mid y_0,x_1,a)\mathrm{d}y_1\mathrm{d}y_2
    \end{align*}
\end{proof}

\clearpage

\begin{center}
\textbf{\LARGE Additional details for the MPH model and the numerical experiments}
\end{center}

\section{Calculations supporting the MPH results presented in main text}\label{app: MPH_derivations}

\subsection{Some useful distributional properties of the MPH model}

Our derivations of the results presented in the main text for the MPH model exploit a number of its special distributional properties.
Variants of these properties feature in prior work, for example that
of \cite{hahn1994efficiency} and \cite{Ridder_Woutersen_EM2003}. Our analysis requires these known implications of the MPH as well as several (apparently) new ones, specific to the feedback case. This section of the supplemental appendix derives and presents the needed preliminary results.

Recall the definitions $\rho_{\theta}\left(z_{t}\right)=\Lambda_{\alpha}\left(y_{t}\right)\exp\left(\gamma y_{t-1}+x_{t}'\beta\right)$
for $z_{t}=\left(y_{t},y_{t-1},x_{t}'\right)'$ and $P_{t}=\rho_{\theta}\left(Z_{t}\right)$
for $t=1,\ldots,T$ with $P_{t}=p_{t}=\rho_{\theta}\left(z_{t}\right)$ denoting -- when the context
is clear -- a \emph{specific} value of the random variable $P_{t}$. Unless there is a risk of confusion we omit the dependence of $P_t$ on $\theta$. 

The period $t$ conditional survival function equals
\begin{align}
\Pr\left(\left.Y_{t}>y_{t}\right|y^{t-1},x^{t},a\right)=S_{\theta}\left(\left.y_{t}\right|y^{t-1},x^{t},a\right) & =\exp\left(-\Lambda_{\alpha}\left(y_{t}\right)\exp\left(\gamma y_{t-1}+x_{t}'\beta+a\right)\right)\nonumber \\
 & =\exp\left(-\rho_{\theta}\left(z_{t}\right)e^{a}\right),\label{eq: mph_conditional_survival_function_y_t}
\end{align}
for $t=1,\ldots,T$. Using monotonicity of the integrated baseline
hazard and (\ref{eq: mph_conditional_survival_function_y_t}) we have
\begin{align}
\Pr\left(\left.P_{t}>p_{t}\right|y^{t-1},x^{t},a\right)= & \Pr\left(\left.Y_{t}>\Lambda_{\alpha}^{-1}\left(p_{t}\exp\left(-\gamma y_{t-1}-x_{t}'\beta\right)\right)\right|y^{t-1},x^{t},a\right)\nonumber \\
= & \exp\left(-\Lambda_{\alpha}\left(\Lambda_{\alpha}^{-1}\left(p_{t}\exp\left(-\gamma y_{t-1}-x_{t}'\beta\right)\right)\right)\exp\left(\gamma y_{t-1}+x_{t}'\beta+a\right)\right)\nonumber \\
= & \exp\left(-p_{t}e^{a}\right).\label{eq: mph_conditional_survival_function_rho_t}
\end{align}
Observe that (\ref{eq: mph_conditional_survival_function_rho_t})
is the survival function for an exponential random variable with rate
parameter $e^{a}$. Since the mapping from $y_{t}\rightarrow p_{t}$
is invertible we therefore have that
\[
\left.P_{t}\right|P^{t-1},X^{t},A\sim\mathrm{Exponential}\left(e^{A}\right),\ t=1,\ldots,T,
\]
and also that
\[
f\left(\left.x^{2:T},p^{1:T}\right|y_{0},x_{1},a\right)=\exp\left(Ta-\left(\sum_{t=1}^{T}p_{t}\right)e^{a}\right)\prod_{t=2}^{T}\widetilde{g}\left(\left.x_{t}\right|y_{0},p^{1:t-1},x^{t-1},a\right),
\]
where $\widetilde{g}\left(\left.x_{t}\right|y_{0},p^{1:t-1},x^{t-1},a\right)=g\left(\left.x_{t}\right|y_{0},y^{1:t-1},x^{t-1},a\right)$
for $p_{t}=\rho_{\theta}\left(z_{t}\right)$ and $t=1,\ldots,T$.
Integrating over $x^{2:T}$ gives the following lemma.
\begin{lemma}
\textsc{(Exponential Structure of the MPH Model)} \label{lem: P1_PT_are_exponential}For
the MPH model defined by Example \ref{ex: mph_intro} we have, for $t=1,\ldots,T$,
that (i) $\left.P_{t}\right|P^{t-1},X^{t},A\sim\mathrm{Exponential}\left(e^{A}\right)$
 and (ii) $\left.P_{t}\right|Y_{0},X_{1},A\overset{ind}{\sim}\mathrm{Exponential}\left(e^{A}\right)$. 
\end{lemma}
Our characterization of the set of FHR moments for the MPH model uses
a particular one-to-one transformation of $P^T$.
We next derive the properties of this transformation. Let $P_{t}\overset{ind}{\sim}\mathrm{Gamma}(\alpha_{t},\beta)$
for $t=1,\ldots,T$ and consider the bijective forward orthogonal
transformation $V^T=G\left(P^T\right)$:
\begin{align*}
V_{t} & =\frac{P_{t}}{\sum_{s=t}^{T}P_{s}},\ t=1,\ldots,T-1\\
V_{T} & =\sum_{t=1}^{T}P_{t},
\end{align*}
with an inverse given by $P^T=G^{-1}\left(V^T\right)$:
\begin{align*}
P_{1} & =V_{1}V_{T}\\
P_{t} & =\prod_{s=1}^{t-1}(1-V_{s})V_{t}V_{T},\ t=2,\ldots,T-1\\
P_{T} & =\prod_{s=1}^{T-1}(1-V_{s})V_{T}.
\end{align*}
By the change-of-variable formula, 
\begin{align*}
f_{V}(V_{1},V_{2},\ldots,V_{T}) & =f_{P_{1}}\left(V_{1}V_{T}\right)\prod_{t=2}^{T-1}f_{P_{t}}\left(\prod_{s=1}^{t-1}(1-V_{s})V_{t}V_{T}\right)f_{P_{T}}\left(\prod_{s=1}^{T-1}(1-V_{s})V_{T}\right)\left|\det\left(\dv{G^{-1}(V^T)}{V^T}\right)\right|.
\end{align*}
We shall show that: 
\begin{align}
V_{T}\sim & \mathrm{Gamma}\left(\sum_{t=1}^{T}\alpha_{t},\beta\right),\label{eq: gamma_beta_claim}\\
V_{t}\sim & \mathrm{Beta}\left(\alpha_{t},\sum_{s=t+1}^{T}\alpha_{s}\right),\quad t=1,\ldots,T-1,\nonumber 
\end{align}
with $(V_{1},V_{2},\ldots,V_{T-1},V_{T})$ additionally mutually independent
of one another. 

To this end, let $J^{(T)}(V_{1},\ldots,V_{T})=\det\left[\dv{G^{-1}(V^T)}{V^T}\right]$. In Subsection \ref{app: detailed calculations} we establish that:
\begin{equation} \label{eq: det_for_rho_tilde_rho_bar_change_of_variables}
    J^{(T)}(V_{1},\ldots,V_{T}) = V_{T}^{T-1}\prod_{s=1}^{T-2}(1-V_{s})^{T-(s+1)}.
\end{equation}

Applying the change-of-variables formula then yields, as we show in Subsection \ref{app: detailed calculations},
\begin{align}\label{eq: density_of_rho_tilde_rho_bar}
     f_{V}\left(V_{1},V_{2},\ldots,V_{T}\right) = & \frac{\beta^{\sum_{t=1}^{T}\alpha_{t}}}{\Gamma\left(\sum_{t=1}^{T}\alpha_{t}\right)}V_{T}^{\sum_{t=1}^{T}\alpha_{t}-1}e^{-\beta V_{T}}\\
    & \times\prod_{t=1}^{T-1}\frac{\Gamma\left(\sum_{s=t}^{T}\alpha_{s}\right)}{\Gamma\left(\alpha_{t}\right)\Gamma\left(\sum_{s=t+1}^{T}\alpha_{s}\right)}V_{t}^{\alpha_{t}-1}(1-V_{t})^{\sum_{s=t+1}^{T}\alpha_{s}-1},
\end{align}
from which claim \eqref{eq: gamma_beta_claim} immediately follows. This result, in combination
with Lemma \ref{lem: P1_PT_are_exponential} above gives, for $\widetilde{P}_{t}$
for $t=1,\ldots,T-1$ and $\overline{P}$ as defined in the main text,
the following lemma.
\begin{lemma}
\textsc{(``Helmert Transformation'' for the MPH Model) }\label{lem: helmert_tranformation_MPH} For
the MPH model defined by Example \ref{ex: mph_intro} we have, conditional on $Y_{0},X_{1}$
and $A$, (i)
\begin{align*}
\widetilde{P}_{t} & \sim\mathrm{Beta}\left(1,T-t\right),\quad t=1,\ldots,T-1,\\
\overline{P} & \sim\mathrm{Gamma}\left(T,e^{A}\right),
\end{align*}
with (ii) $(\widetilde{P}_{1},\widetilde{P}_{2},\ldots,\widetilde{P}_{T-1},\overline{P})$
additionally mutually independent of one another.
\end{lemma}

\subsection{Semiparametric efficiency bounds for the MPH model}

Here we present our semiparametric efficiency bound derivations for the MPH model. Our calculations exploit Lemma \ref{lem_moment_charac_MPH} and Theorem \ref{theo_eff} of the main text. In order to compute the efficient score for the MPH using equation \eqref{eq_effscore} in the main text, the following lemma is useful.
\begin{lemma}  \label{proj_lemma_MPH}
\textsc{(MPH Projection)}
    For any $L\times 1$ mean-zero random vector $\varphi_{\theta}(Y^{T},X^T)$ such that $\mathbb{E}\left[\norm{\varphi_{\theta}(Y^{T},X^T)}^2\right]<\infty$, we have
    \begin{align*}
        \Pi(\varphi_{\theta}(Y^{T},X^T)|\mathcal{T}_{\theta,\omega,L}^{\perp})&=\sum_{t=1}^{T-1}\varphi_{\theta,t}^{\perp}(Y_0,\widetilde{P}^{t},\overline{P},X^{t}),
    \end{align*}
    with
    \begin{multline*}
        \varphi_{\theta,t}^{\perp}(Y_0,\widetilde{P}^{t},\overline{P},X^{t}) = \mathbb{E}\left[\varphi_{\theta}(Y^{T},X^T)|Y_0,\widetilde{P}^{t},\overline{P},X^{t}\right] 
        -\mathbb{E}\left[\varphi_{\theta}(Y^{T},X^T)|Y_0,\widetilde{P}^{t-1},\overline{P},X^{t}\right].
    \end{multline*}
\end{lemma}
\begin{proof}
First, we note that by iterated expectations,
\begin{align*}
    \mathbb{E}\left[\varphi_{\theta,t}^{\perp}(Y_0,\widetilde{P}^{t},\overline{P},X^{t})|Y_0,\widetilde{P}^{t-1},\overline{P},X^{t}\right]=0.
\end{align*}
Thus, Lemma \ref{lem_moment_charac_MPH} and Theorem \ref{theo_eff} imply that $\sum_{t=1}^{T-1}\varphi_{\theta,t}^{\perp}(Y_0,\widetilde{P}^{t},\overline{P},X^{t})$
is an element of $\mathcal{T}_{\theta,\omega,L}^{\perp}$. Next fix $s,t\in\{1,\ldots,T-1\}$ and consider any function $\psi_{\theta,s}(Y_0,\widetilde{P}^{s},\overline{P},X^{s})\in \mathcal{T}_{\theta,\omega,L}^{\perp}$ such that
\begin{align*}
    \mathbb{E}\left[\psi_{\theta,s}(Y_0,\widetilde{P}^{s},\overline{P},X^{s})|Y_0,\widetilde{P}^{s-1},\overline{P},X^{s}\right]=0.
\end{align*}
Observe that if $s>t$
\begin{align*}
    &\mathbb{E}\left[\varphi_{\theta,t}^{\perp}(Y_0,\widetilde{P}^{t},\overline{P},X^{t})'\psi_{\theta,s}(Y_0,\widetilde{P}^{s},\overline{P},X^{s})\right] \\
    &= \mathbb{E}\left[\varphi_{\theta,t}^{\perp}(Y_0,\widetilde{P}^{t},\overline{P},X^{t})' \mathbb{E}\left[\psi_{\theta,s}(Y_0,\widetilde{P}^{s},\overline{P},X^{s})|Y_0,\widetilde{P}^{s-1},\overline{P},X^{s}\right]\right] \\
    &=0.
\end{align*}
Symmetrically, if $t>s$, 
\begin{align*}
    \mathbb{E}\left[\varphi_{\theta,t}^{\perp}(Y_0,\widetilde{P}^{t},\overline{P},X^{t})'\psi_{\theta,s}(Y_0,\widetilde{P}^{s},\overline{P},X^{s})\right]=0,
\end{align*}
since by construction $\mathbb{E}\left[\varphi_{\theta,t}^{\perp}(Y_0,\widetilde{P}^{t},\overline{P},X^{t})|Y_0,\widetilde{P}^{t-1},\overline{P},X^{t}\right]=0$. Using these observations we can show:
\begin{align*}
    \mathbb{E}&\left[\left(\varphi_{\theta}(Y^{T},X^T) -\sum_{t=1}^{T-1}\varphi_{\theta,t}^{\perp}(Y_0,\widetilde{P}^{t},\overline{P},X^{t})\right)'\psi_{\theta,s}\left( 
    Y_0,\widetilde{P}^{s},\overline{P},X^{s}\right)\right] \\
    =&\mathbb{E}\left[\left(\varphi_{\theta}(Y^{T},X^T)-\varphi_{\theta,s}^{\perp}(Y_0,\widetilde{P}^{s},\overline{P},X^{s})\right)'\psi_{\theta,s}\left(Y_0,\widetilde{P}^{s},\overline{P},X^{s}\right)\right] \\
     =&\mathbb{E}\left[\left(\varphi_{\theta}(Y^{T},X^T)-\mathbb{E}\left[\varphi_{\theta}(Y^{T},X^T)|Y_0,\widetilde{P}^{s},\overline{P},X^{s}\right]\right)'\psi_{\theta,s}\left(Y_0,\widetilde{P}^{s},\overline{P},X^{s}\right)\right] \\
     &+\mathbb{E}\left[\mathbb{E}\left[\varphi_{\theta}(Y^{T},X^T)|Y_0,\widetilde{P}^{s-1},\overline{P},X^{s}\right]'\psi_{\theta,s}\left(Y_0,\widetilde{P}^{s},\overline{P},X^{s}\right)\right] \\
     =&0,
\end{align*}
where the second equality uses the definition of $\varphi_{\theta,t}^{\perp}(Y_0,\widetilde{P}^{t},\overline{P},X^{t})$ given in the statement of the Lemma and where the last line follows from iterated expectations and the fact that $\mathbb{E}\left[\psi_{\theta,s}(Y_0,\widetilde{P}^{s},\overline{P},X^{s})|Y_0,\widetilde{P}^{s-1},\overline{P},X^{s}\right]=0$.
Since by Lemma \ref{lem_moment_charac_MPH} and Theorem \ref{theo_eff}, any function of $\mathcal{T}_{\theta,\omega,L}^{\perp}$ is of the form:
\begin{align*}
    \psi_{\theta}(Y_{0},\widetilde{P}^{T-1},\overline{P},X^{T-1})=\sum_{t=1}^{T-1}\psi_{\theta,t}(Y_0,\widetilde{P}^{t},\overline{P},X^{t}),
\end{align*}
with $ \mathbb{E}\left[\left.\psi_{\theta,t}(Y_{0},\widetilde{P}^{t},\overline{P},X^{t})\right|Y_{0},\widetilde{P}^{t-1},\overline{P},X^{t}\right]=0, \quad t=1,\ldots,T-1$, these calculations show that 
\begin{align*}
        &\mathbb{E}\left[\left(\varphi_{\theta}(Y^{T},X^T)-\sum_{t=1}^{T-1}\varphi_{\theta,t}^{\perp}(Y_0,\widetilde{P}^{t},\overline{P},X^{t})\right)\psi_{\theta}(Y_{0},\widetilde{P}^{T-1},\overline{P},X^{T-1})\right]=0.
\end{align*}
By the Projection Theorem (e.g Theorem 11.1 in \cite{van2000asymptotic}) we conclude that, 
\begin{align*}
     \Pi(\varphi_{\theta}(Y_0,Y^{T},X^T)|\mathcal{T}_{\theta,\omega,L}^{\perp})&=\sum_{t=1}^{T-1}\varphi_{\theta,t}^{\perp}(Y_0,\widetilde{P}^{t},\overline{P},X^{t}),
\end{align*}
as claimed.
\end{proof}
\subsubsection{MPH efficient score under feedback} 
We can use Lemma \ref{proj_lemma_MPH} to derive an explicit expression of the efficient score for $\theta=\left(\alpha,\beta',\gamma\right)'$. Given the historical and continued importance of the MPH model, the required calculations, although involved and tedious, are presented in detail here.

Recall the form the integrated likelihood for the semiparametric panel data model with feedback presented in equation (\ref{eq: feedback_complete_data_likelihood}) of the main text. The form of the parametric period $t$ panel data model equals, in the MPH setting,
\begin{equation*}
    f_{\theta}(y_t|x_t,y_{t-1},a)=\lambda_{\alpha}\left(y_{t}\right)\exp\left(\gamma y_{t-1}+x_{t}'\beta+a\right)\exp\left(-\rho_{\theta}\left(z_{t}\right)e^a\right).
\end{equation*}
Further recall the notation: $P_t= \rho_{\theta}\left(Z_{t}\right)=\Lambda_{\alpha}\left(Y_{t}\right)e^{\gamma Y_{t-1}+X_{t}'\beta}$ (when evaluated at population value $\theta=\theta_0$). The (submodel) score for $\theta$ equals
 \begin{align*}
     S^{\theta}(Y^T,X^T) = \mathbb{E}\left[\sum_{t=1}^T \frac{\partial \ln f_{\theta}(Y_t\,|\, Y_{t-1},X_t,A)}{\partial \theta}\,|\, Y^T,X^T\right].
 \end{align*}
By Theorem \ref{theo_eff} and Lemma \ref{proj_lemma_MPH} the efficient score, when $T=2$, equals\footnote{For conciseness, we omit the dependence of efficient scores on the nuisance parameter $\omega$ in this section of the appendix, unless there is a risk of confusion.}
\begin{align} \label{eq: eff_score_MPH_feedback_general_T=2}
        \phi_{\theta}^{\rm eff}(Y^2,X^2)=\mathbb{E}\left[S^{\theta}(Y^2,X^2)\,|\,Y_{0},\widetilde{P}_{1},\overline{P},X_{1}\right]-\mathbb{E}\left[S^{\theta}(Y^2,X^2)\,|\,Y_{0},\overline{P},X_{1}\right],    
\end{align}
while, in the general $T\geq 2$ case, it equals
\begin{align*}
    \phi_{\theta}^{\rm eff}(Y^T,X^T)&=\sum_{t=1}^{T-1}\phi_{\theta,t}^{\rm eff, \perp}(Y_0,\widetilde{P}^{1:t},\overline{P},X^{1:t}), 
\end{align*}
with, for $t=1,\ldots,T-1$,
\begin{multline}
    \phi_{\theta,t}^{\rm eff, \perp}(Y_0,\widetilde{P}^{1:t},\overline{P},X^{1:t}) \\
    = \mathbb{E}\left[S^{\theta}(Y^T,X^T)|Y_0,\widetilde{P}^{1:t},\overline{P},X^{1:t}\right]-\mathbb{E}\left[S^{\theta}(Y^T,X^T)|Y_0,\widetilde{P}^{1:t-1},\overline{P},X^{1:t}\right].
\end{multline}
In what follows we focus on the $T=2$ special case for simplicity (and because it is this case we study in the numerical experiments reported in Section \ref{sec: mph_deep_dive} of the main text). The calculations proceed by evaluating -- and simplifying as far as appears to be feasible -- equation \eqref{eq: eff_score_MPH_feedback_general_T=2} above.
We do this separately for $\beta$, $\gamma$ and $\alpha$. Our efficient score expression for $\beta$ holds for any baseline hazard, those we present for $\gamma$ and $\alpha$ assume that the baseline hazard takes the Weibull form: $\lambda_{\alpha}(y_t)=\alpha y_{t}^{\alpha-1}$. The calculations presented below can be adapted to study semiparametric efficiency bounds under other maintained baseline hazard assumptions.

\textbf{The efficient score for $\beta$: } The (submodel) score for $\beta$, recalling the notation re-established immediately above, equals:
\begin{align}
		S^{\beta}(Y^T,X^T)&=\mathbb{E}\left[\sum_{t=1}^T \frac{\partial \ln f_{\theta}(Y_t\,|\, Y_{t-1},X_t,A)}{\partial \beta}\,|\, Y^T,X^T\right] \label{eq: submodel_beta_score_MPH} \\
		&=\mathbb{E}\left[\sum_{t=1}^T X_{t}\left(1-P_t e^{a}\right)\,|\,Y^T,X^T\right] \notag \\
  &=\sum_{t=1}^T X_t-\sum_{t=1}^T X_tP_t\mathbb{E}\left[e^{A}\,|\, Y^T,X^T\right], \notag
\end{align}
which, when $T=2$, coincides with
\begin{align*}
		S^{\beta}(Y^2,X^2)&= X_{1}+X_{2}-\left(X_1P_1+X_2P_2\right)\mathbb{E}\left[e^{A}\,|\, Y^2,X^2\right].
\end{align*}
Reparameterizing in terms of $\widetilde{P}_1$ and $\overline{P}$ yields the equivalent expression
\begin{align*}
		S^{\beta}(Y^2,X^2)&= X_{1}+X_{2}-\left(X_1\widetilde{P}_1\overline{P}+X_2(1-\widetilde{P}_1)\overline{P}\right)\mathbb{E}\left[e^{A}\,|\, Y_0,\widetilde{P}_1,\overline{P},X^2\right].
\end{align*}
The efficient score for $\beta$, from equation \eqref{eq: eff_score_MPH_feedback_general_T=2} above, equals
\begin{align}
   \phi_{\theta}^{\rm eff,\beta}(Y^2,X^2)=\mathbb{E}\left[S^{\beta}(Y^2,X^2)|Y_0,\widetilde{P}_1,\overline{P},X_1\right]-\mathbb{E}\left[S^{\beta}(Y^2,X^2)|Y_0,\overline{P},X_1\right].\label{eq: efficient_score_beta_general_T=2}
\end{align}
We evaluate the two expectations to the right of the equality above in sequence. First,
\begin{align*}
    &\mathbb{E}\left[S^{\beta}(Y^2,X^2)|Y_0,\widetilde{P}_1,\overline{P},X_1\right]\\
    &=X_{1}+\mathbb{E}\left[X_{2}|Y_0,\widetilde{P}_1,\overline{P},X_1\right]-X_1\widetilde{P}_1\overline{P}\mathbb{E}\left[e^{A}\,|\, Y_0,\overline{P},X_1\right]-(1-\widetilde{P}_1)\overline{P}\mathbb{E}\left[X_2e^{A}\,|\, Y_0,\widetilde{P}_1,\overline{P},X_1\right].
\end{align*}
Second, and making use of Lemma \ref{lem: helmert_tranformation_MPH} when evaluating expectations over $\widetilde{P}_1$,
\begin{align*}
    &\mathbb{E}\left[S^{\beta}(Y^2,X^2)|Y_0,\overline{P},X_1\right]\\
    &=X_{1}+\mathbb{E}\left[X_{2}|Y_0,\overline{P},X_1\right]-X_1\overline{P}\mathbb{E}\left[\widetilde{P}_1e^{A}\,|\, Y_0,\overline{P},X_1\right]-\overline{P}\mathbb{E}\left[X_2(1-\widetilde{P}_1)e^{A}\,|\, Y_0,\overline{P},X_1\right] \\
    &=X_{1}+\mathbb{E}\left[X_{2}|Y_0,\overline{P},X_1\right]-\frac{1}{2}X_1\overline{P}\mathbb{E}\left[e^{A}\,|\, Y_0,\overline{P},X_1\right]-\overline{P}\mathbb{E}\left[X_2(1-\widetilde{P}_1)e^{A}\,|\, Y_0,\overline{P},X_1\right]. 
\end{align*}
Collecting terms we get an efficient score for $\beta$ of
\begin{align*}
    \phi_{\theta}^{\rm eff,\beta}(Y^2,X^2)&=-X_1\left(\widetilde{P}_1-\frac{1}{2}\right)\overline{P}\mathbb{E}\left[e^{A}\,|\, Y_0,\overline{P},X_1\right] \\
    &+\mathbb{E}\left[X_{2}|Y_0,\widetilde{P}_1,\overline{P},X_1\right]-\mathbb{E}\left[X_{2}|Y_0,\overline{P},X_1\right] \\
    &-\left((1-\widetilde{P}_1)\mathbb{E}\left[X_2e^{A}\,|\, Y_0,\widetilde{P}_1,\overline{P},X_1\right]-\mathbb{E}\left[X_2(1-\widetilde{P}_1)e^{A}\,|\, Y_0,\overline{P},X_1\right]\right)\overline{P},
\end{align*}
or more concisely
\begin{align}
    \phi_{\theta}^{\rm eff,\beta}(Y^2,X^2)=&-X_1\left(\widetilde{P}_1-\frac{1}{2}\right)\overline{P}\mathbb{E}\left[e^{A}\,|\, Y_0,\overline{P},X_1\right] \notag \\
    &+\mathbb{E}\left[X_{2}\left(1-(1-\widetilde{P}_1)\overline{P}e^{A}\right)|Y_0,\widetilde{P}_1,\overline{P},X_1\right] \notag \\
    &-\mathbb{E}\left[X_{2}\left(1-(1-\widetilde{P}_1)\overline{P}e^{A}\right)|Y_0,\overline{P},X_1\right],\label{eq: efficient_score_beta_specific_T=2}
\end{align}
as claimed in Section \ref{sec: mph_deep_dive}.

\textbf{The efficient score for $\gamma$: } The (submodel) score for $\gamma$ equals:
	\begin{align*}
		S^{\gamma}(Y^T,X^T)&=\mathbb{E}\left[\sum_{t=1}^T \frac{\partial \ln f_{\theta}(Y_t\,|\, Y_{t-1},X_t,A)}{\partial \gamma}\,|\, Y^T,X^T\right]	\\
		&=\mathbb{E}\left[\sum_{t=1}^T Y_{t-1}\left(1-P_t e^{a}\right)\,|\,Y^T,X^T\right] \\
  &=\sum_{t=1}^T Y_{t-1}-\sum_{t=1}^T Y_{t-1}P_t\mathbb{E}\left[e^{A}\,|\, Y^T,X^T\right],
	\end{align*}
which, when $T=2$, coincides with
\begin{align*}
		S^{\gamma}(Y^2,X^2)&= Y_{0}+Y_{1}-\left(Y_0P_1+Y_1P_2\right)\mathbb{E}\left[e^{A}\,|\, Y^2,X^2\right],
\end{align*}
 or, after reparameterization in terms of $(\widetilde{P}_1,\overline{P})$, with
 \begin{align*}
		S^{\gamma}(Y^2,X^2)&= Y_{0}+Y_{1}-\left(Y_0\widetilde{P}_1\overline{P}+Y_1(1-\widetilde{P}_1)\overline{P}\right)\mathbb{E}\left[e^{A}\,|\, Y_0,\widetilde{P}_1,\overline{P},X^2\right].
\end{align*}
The efficient score for $\gamma$, from equation \eqref{eq: eff_score_MPH_feedback_general_T=2} above, equals
\begin{align}
     \phi_{\theta}^{\rm eff,\gamma}(Y^2,X^2)=\mathbb{E}\left[S^{\gamma}(Y^2,X^2)|Y_0,\widetilde{P}_1,\overline{P},X_1\right]-\mathbb{E}\left[S^{\gamma}(Y^2,X^2)|Y_0,\overline{P},X_1\right]\label{eq: efficient_score_gamma_general_T=2},
\end{align}
where evaluate the first expectation to the right of equality as
\begin{align*}
    \mathbb{E}\left[S^{\gamma}(Y^2,X^2)|Y_0,\widetilde{P}_1,\overline{P},X_1\right]&=Y_{0}+Y_{1}-\left(Y_0\widetilde{P}_1\overline{P}+Y_1(1-\widetilde{P}_1)\overline{P}\right)\mathbb{E}\left[e^{A}\,|\, Y_0,\widetilde{P}_1,\overline{P},X_1\right] \\
    &=Y_{0}+Y_{1}-\left(Y_0\widetilde{P}_1\overline{P}+Y_1(1-\widetilde{P}_1)\overline{P}\right)\mathbb{E}\left[e^{A}\,|\, Y_0,\overline{P},X_1\right], 
\end{align*}
using the fact that $\widetilde{P}_1\perp (A,\overline{P})|Y_0,X_1$ (see Lemma \ref{lem: helmert_tranformation_MPH} above). 

In order to evaluate the second expectation entering the efficient score for $\gamma$ we consider the special case of a Weibull baseline hazard: $\Lambda_{\alpha}(y_t)=y_t^{\alpha}$, $\lambda_{\alpha}(y_t)=\alpha y_t^{\alpha-1}$. This hazard implies the convenient relation $Y_1=P_1^{\frac{1}{\alpha}}e^{-X_1'\frac{\beta}{\alpha}-\frac{\gamma}{\alpha}Y_0}=\widetilde{P}_1^{\frac{1}{\alpha}}\overline{P}^{\frac{1}{\alpha}}e^{-X_1'\frac{\beta}{\alpha}-\frac{\gamma}{\alpha}Y_0}$. This allows us to use Lemma \ref{lem: helmert_tranformation_MPH} to evaluate the expectation:
\begin{align*}
    \mathbb{E}\left[Y_1|Y_0,\overline{P},X_1\right]&=e^{-X_1'\frac{\beta}{\alpha}-\frac{\gamma}{\alpha}Y_0}\mathbb{E}\left[\widetilde{P}_1^{\frac{1}{\alpha}}|Y_0,\overline{P},X_1\right]\overline{P}^{\frac{1}{\alpha}} \\
    &=\frac{\alpha}{1+\alpha}e^{-X_1'\frac{\beta}{\alpha}-\frac{\gamma}{\alpha}Y_0} \overline{P}^{\frac{1}{\alpha}},
\end{align*}
where we have used the fact that $\widetilde{P}_1\sim U[0,1]=\mathrm{Beta}\left(1,1\right)$. Proceeding in the same general way we also get
\begin{align*}
    \mathbb{E}\left[Y_1(1-\widetilde{P}_1)|Y_0,\overline{P},X_1\right]&=e^{-X_1'\frac{\beta}{\alpha}-\frac{\gamma}{\alpha}Y_0}\mathbb{E}\left[\widetilde{P}_1^{\frac{1}{\alpha}}(1-\widetilde{P}_1)|Y_0,\overline{P},X_1\right]\overline{P}^{\frac{1}{\alpha}} \\
    &=\left(\frac{\alpha}{1+\alpha}-\frac{\alpha}{1+2\alpha}\right)e^{-X_1'\frac{\beta}{\alpha}-\frac{\gamma}{\alpha}Y_0} \overline{P}^{\frac{1}{\alpha}}.
\end{align*}
Using these two results we evaluate the second expectation in equation \eqref{eq: efficient_score_gamma_general_T=2} as:
\begin{align*}
    \mathbb{E}\left[S^{\gamma}(Y^2,X^2)|Y_0,\overline{P},X_1\right]&=Y_{0}+\mathbb{E}\left[Y_1|Y_0,\overline{P},X_1\right] \\
    &-\left(Y_0\mathbb{E}\left[\widetilde{P}_1|Y_0,\overline{P},X_1\right]+\mathbb{E}\left[Y_1(1-\widetilde{P}_1)|Y_0,\overline{P},X_1\right]\right)\overline{P}\mathbb{E}\left[e^{A}\,|\, Y_0,\overline{P},X_1\right] \\
    &=Y_{0}+\frac{\alpha}{1+\alpha}e^{-X_1'\frac{\beta}{\alpha}-\frac{\gamma}{\alpha}Y_0} \overline{P}^{\frac{1}{\alpha}} \\
    &-\left(Y_0\frac{1}{2}+\left(\frac{\alpha}{1+\alpha}-\frac{\alpha}{1+2\alpha}\right)e^{-X_1'\frac{\beta}{\alpha}-\frac{\gamma}{\alpha}Y_0} \overline{P}^{\frac{1}{\alpha}}\right)\overline{P}\mathbb{E}\left[e^{A}\,|\, Y_0,\overline{P},X_1\right]. 
\end{align*}
Collecting terms yields an efficient score for $\gamma$ of:
\begin{multline}
    \phi_{\theta}^{\rm eff,\gamma}(Y^2,X^2)=Y_1-\frac{\alpha}{1+\alpha}e^{-X_1'\frac{\beta}{\alpha}-\frac{\gamma}{\alpha}Y_0} \overline{P}^{\frac{1}{\alpha}}  \\
    -\left(Y_0\left(\widetilde{P}_1-\frac{1}{2}\right)+\left(Y_1(1-\widetilde{P}_1)-\left(\frac{\alpha}{1+\alpha}-\frac{\alpha}{1+2\alpha}\right)e^{-X_1'\frac{\beta}{\alpha}-\frac{\gamma}{\alpha}Y_0} \overline{P}^{\frac{1}{\alpha}}\right)\right)\\
    \times \overline{P}\mathbb{E}\left[e^{A}\,|\, Y_0,\overline{P},X_1\right].\label{eq: efficient_score_gamma_specific_T=2}
\end{multline}
This expression does not appear in the main text.

\textbf{The efficient score for $\alpha$: } 
The (submodel) score for $\alpha$ equals:
	\begin{align*}
		S^{\alpha}(Y^T,X^T)&=\mathbb{E}\left[\sum_{t=1}^T \frac{\partial \ln f_{\theta}(Y_t\,|\, Y_{t-1},X_t,A)}{\partial \alpha}\,|\, Y^T,X^T\right]	\\
		&=\mathbb{E}\left[\sum_{t=1}^T \pdv{\ln \lambda_{\alpha}(Y_t)
        }{\alpha}-\pdv{\ln \Lambda_{\alpha}(Y_t)
        }{\alpha}P_te^{A}\,|\,Y^T,X^T \right] \\
        &=\sum_{t=1}^T \pdv{\ln\lambda_{\alpha}(Y_t)
        }{\alpha}-\sum_{t=1}^T\pdv{\ln\Lambda_{\alpha}(Y_t)
        }{\alpha}P_t\mathbb{E}\left[e^{A}\,|\,Y^T,X^T \right],
\end{align*}
which, when $T=2$, equals (after reparameterization in terms of $\widetilde{P}_1,\overline{P}$):
	\begin{align*}
		S^{\alpha}(Y^2,X^2)&=\pdv{\ln \lambda_{\alpha}(Y_1)
        }{\alpha}+\pdv{\ln \lambda_{\alpha}(Y_2)
        }{\alpha}\\
        &-\left( \pdv{\ln\Lambda_{\alpha}(Y_1)
        }{\alpha}\widetilde{P}_1+\pdv{\ln\Lambda_{\alpha}(Y_2)
        }{\alpha}\times(1-\widetilde{P}_1)\right)\overline{P}\mathbb{E}\left[e^{A}\,|\,Y_0,\widetilde{P}_1,\overline{P},X_1,X_2 \right],
\end{align*}

To make further progress we again consider the Weibull baseline hazard case; implying $\ln \Lambda_{\alpha}(Y_t)=\alpha \ln Y_t$ and $\ln \lambda_{\alpha}(Y_t)=\ln \alpha +(\alpha-1)\ln Y_t$. This yields
\begin{align*}
    \pdv{\ln \Lambda_{\alpha}(Y_t)}{\alpha}&=\ln Y_t \\
        \pdv{\ln \lambda_{\alpha}(Y_t)}{\alpha}&=\frac{1}{\alpha}+\ln Y_t.
\end{align*}
In the parameterization $\widetilde{P}_1, \overline{P}$, with 
$Y_1=\widetilde{P}_1^{\frac{1}{\alpha}}\overline{P}^{\frac{1}{\alpha}}e^{-X_1'\frac{\beta}{\alpha}-\frac{\gamma}{\alpha}Y_0}$ and $Y_2=(1-\widetilde{P}_1)^{\frac{1}{\alpha}}\overline{P}^{\frac{1}{\alpha}}e^{-X_2'\frac{\beta}{\alpha}-\frac{\gamma}{\alpha}Y_1}$ we have
\begin{align*}
    \pdv{\ln \Lambda_{\alpha}(Y_1)}{\alpha}&=\ln Y_1= -\frac{1}{\alpha}\left(X_1'\beta+\gamma Y_0\right)+\frac{1}{\alpha}\ln \widetilde{P}_1+\frac{1}{\alpha}\ln \overline{P} \\
    \pdv{\ln \Lambda_{\alpha}(Y_2)}{\alpha}&=\ln Y_2= -\frac{1}{\alpha}\left(X_2'\beta+\gamma Y_1\right)+\frac{1}{\alpha}\ln (1-\widetilde{P}_1)+\frac{1}{\alpha}\ln \overline{P} \\
    &=-\frac{1}{\alpha}\left(X_2'\beta+\gamma \widetilde{P}_1^{\frac{1}{\alpha}}\overline{P}^{\frac{1}{\alpha}}e^{-X_1'\frac{\beta}{\alpha}-\frac{\gamma}{\alpha}Y_0}\right)+\frac{1}{\alpha}\ln (1-\widetilde{P}_1)+\frac{1}{\alpha}\ln \overline{P},
\end{align*}
and 
\begin{align*}
    \pdv{\ln \lambda_{\alpha}(Y_1)}{\alpha}&=\frac{1}{\alpha} -\frac{1}{\alpha}\left(X_1'\beta+\gamma Y_0\right)+\frac{1}{\alpha}\ln \widetilde{P}_1+\frac{1}{\alpha}\ln \overline{P} \\
    \pdv{\ln \lambda_{\alpha}(Y_2)}{\alpha}&=\frac{1}{\alpha}-\frac{1}{\alpha}\left(X_2'\beta+\gamma \widetilde{P}_1^{\frac{1}{\alpha}}\overline{P}^{\frac{1}{\alpha}}e^{-X_1'\frac{\beta}{\alpha}-\frac{\gamma}{\alpha}Y_0}\right)+\frac{1}{\alpha}\ln (1-\widetilde{P}_1)+\frac{1}{\alpha}\ln \overline{P}.
\end{align*}
Using this, we obtain
\begin{align*}
		S^{\alpha}(Y^2,X^2)&=\frac{1}{\alpha}\left[2+\ln \widetilde{P}_1+\ln (1-\widetilde{P}_1)+2\ln \overline{P}\right] \\
        &-\frac{1}{\alpha}\left(\widetilde{P}_1\ln \widetilde{P}_1+(1-\widetilde{P}_1)\ln (1-\widetilde{P}_1)+\ln \overline{P}\right)\overline{P}\mathbb{E}\left[e^{A}\,|\,Y_0,\widetilde{P}_1,\overline{P},X_1,X_2 \right] \\
        &-\left[X_1+X_2-\left(X_1\widetilde{P}_1+X_2(1-\widetilde{P}_1)\right)\overline{P}\mathbb{E}\left[e^{A}\,|\,Y_0,\widetilde{P}_1,\overline{P},X_1,X_2 \right]\right]'\frac{\beta}{\alpha} \\
        &-\left[Y_0+Y_1-\left(Y_0\widetilde{P}_1+Y_1(1-\widetilde{P}_1)\right)\overline{P}\mathbb{E}\left[e^{A}\,|\, Y_0,\widetilde{P}_1,\overline{P},X^2\right]\right]\frac{\gamma}{\alpha} \\
        &=\frac{1}{\alpha}\left[2+\ln \widetilde{P}_1+\ln (1-\widetilde{P}_1)+2\ln \overline{P}\right] \\
        &-\frac{1}{\alpha}\left(\widetilde{P}_1\ln \widetilde{P}_1+(1-\widetilde{P}_1)\ln (1-\widetilde{P}_1)+\ln \overline{P}\right)\overline{P}\mathbb{E}\left[e^{A}\,|\,Y_0,\widetilde{P}_1,\overline{P},X_1,X_2 \right] \\
        &-S^{\beta}(Y^2,X^2)'\frac{\beta}{\alpha}- S^{\gamma}(Y^2,X^2)\frac{\gamma}{\alpha} .
\end{align*}
We recall that the efficient score for $\alpha$, from equation \eqref{eq: eff_score_MPH_feedback_general_T=2} above, equals
\begin{align}
    \phi_{\theta}^{\rm eff,\alpha}(Y^2,X^2)=\mathbb{E}\left[S^{\alpha}(Y^2,X^2)|Y_0,\widetilde{P}_1,\overline{P},X_1\right]-\mathbb{E}\left[S^{\alpha}(Y^2,X^2)|Y_0,\overline{P},X_1\right].\label{eq: efficient_score_alpha_general_T=2}
\end{align}
The following implications of Lemma \ref{lem: helmert_tranformation_MPH} are useful for the calculations which follow: $\widetilde{P}_1\perp \overline{P}|Y_0,X_1$ with $\widetilde{P}_1\sim U[0,1]$ and hence, by the properties of the uniform distribution, $-\ln \widetilde{P}_1\sim \mathrm{Exp}\left(1\right)$. Using these facts we evaluate the following two expectations:
 \begin{align*}
     &\mathbb{E}\left[-\ln(1-\widetilde{P}_1)|Y_0,\overline{P},X_1\right]=\mathbb{E}\left[-\ln (1-\widetilde{P}_1)|Y_0,X_1\right]=\mathbb{E}\left[-\ln \widetilde{P}_1|Y_0,X_1\right]=+1 \\
    &\mathbb{E}\left[\ln \widetilde{P}_1 \widetilde{P}_1|Y_0,\overline{P},X_1\right]=\mathbb{E}\left[\ln \left(1-\widetilde{P}_1\right) \left(1-\widetilde{P}_1\right)|Y_0,\overline{P},X_1\right]=-\frac{1}{4}.
\end{align*}
We have 
\begin{align*}
    \mathbb{E}\left[S^{\alpha}(Y^2,X^2)|Y_0,\widetilde{P}_1,\overline{P},X_1\right]&=\frac{1}{\alpha}\left[2+\ln \widetilde{P}_1+\ln (1-\widetilde{P}_1)+2\ln \overline{P}\right] \\
        &-\frac{1}{\alpha}\left(\widetilde{P}_1\ln \widetilde{P}_1+(1-\widetilde{P}_1)\ln (1-\widetilde{P}_1)+\ln \overline{P}\right)\overline{P}\mathbb{E}\left[e^{A}\,|\,Y_0,\widetilde{P}_1,\overline{P},X_1\right] \\
        &-\mathbb{E}\left[S^{\beta}(Y^2,X^2)|Y_0,\widetilde{P}_1,\overline{P},X_1\right]'\frac{\beta}{\alpha}- \mathbb{E}\left[S^{\gamma}(Y^2,X^2)|Y_0,\widetilde{P}_1,\overline{P},X_1\right]\frac{\gamma}{\alpha} \\
        &=\frac{1}{\alpha}\left[2+\ln \widetilde{P}_1+\ln (1-\widetilde{P}_1)+2\ln \overline{P}\right] \\
        &-\frac{1}{\alpha}\left(\widetilde{P}_1\ln \widetilde{P}_1+(1-\widetilde{P}_1)\ln (1-\widetilde{P}_1)+\ln \overline{P}\right)\overline{P}\mathbb{E}\left[e^{A}\,|\,Y_0,\overline{P},X_1\right] \\
        &-\mathbb{E}\left[S^{\beta}(Y^2,X^2)|Y_0,\widetilde{P}_1,\overline{P},X_1\right]'\frac{\beta}{\alpha}- \mathbb{E}\left[S^{\gamma}(Y^2,X^2)|Y_0,\widetilde{P}_1,\overline{P},X_1\right]\frac{\gamma}{\alpha},
\end{align*}
where the simplification in the second equality uses the fact $\mathbb{E}\left[e^{A}\,|\,Y_0,\widetilde{P}_1,\overline{P},X_1\right]=\mathbb{E}\left[e^{A}\,|\,Y_0,\overline{P},X_1\right]$ by Lemma \ref{lem: helmert_tranformation_MPH}. Similarly,
\begin{align*}
    \mathbb{E}\left[S^{\alpha}(Y^2,X^2)|Y_0,\overline{P},X_1\right]&=\frac{1}{\alpha}\left[2+\mathbb{E}\left[\ln \widetilde{P}_1+\ln (1-\widetilde{P}_1)|Y_0,\overline{P},X_1\right]+2\ln \overline{P}\right] \\
        &-\frac{1}{\alpha}\overline{P}\mathbb{E}\left[\left(\widetilde{P}_1\ln \widetilde{P}_1+(1-\widetilde{P}_1)\ln (1-\widetilde{P}_1)\right)e^{A}\,|\,Y_0,\overline{P},X_1\right] \\
        &-\frac{1}{\alpha}\overline{P}\ln \overline{P}\mathbb{E}\left[e^{A}\,|\,Y_0,\overline{P},X_1\right] \\
        &-\mathbb{E}\left[S^{\beta}(Y^2,X^2)|Y_0,\overline{P},X_1\right]'\frac{\beta}{\alpha}- \mathbb{E}\left[S^{\gamma}(Y^2,X^2)|Y_0,\overline{P},X_1\right]\frac{\gamma}{\alpha} \\
        &=\frac{2}{\alpha}\ln \overline{P} \\
        &-\frac{1}{\alpha}\left(\ln \overline{P}-\frac{1}{2}\right)\overline{P}\mathbb{E}\left[e^{A}\,|\,Y_0,\overline{P},X_1\right] \\
        &-\mathbb{E}\left[S^{\beta}(Y^2,X^2)|Y_0,\overline{P},X_1\right]'\frac{\beta}{\alpha}- \mathbb{E}\left[S^{\gamma}(Y^2,X^2)|Y_0,\overline{P},X_1\right]\frac{\gamma}{\alpha},
\end{align*}
where the second equality exploits the fact that $\widetilde{P}_1\perp (A,\overline{P})|Y_0,X_1$ (see Lemma \ref{lem: helmert_tranformation_MPH} ) and the two expectations reported above. As a result,
\begin{align}
    \phi_{\theta}^{\rm eff,\alpha}(Y^2,X^2)&=\frac{1}{\alpha}\left[2+\ln \widetilde{P}_1+\ln (1-\widetilde{P}_1)\right] \notag \\
    &-\frac{1}{\alpha}\left(\widetilde{P}_1\ln \widetilde{P}_1+(1-\widetilde{P}_1)\ln (1-\widetilde{P}_1)+\frac{1}{2}\right)\overline{P}\mathbb{E}\left[e^{A}\,|\,Y_0,\overline{P},X_1\right] \notag\\
    &-\phi_{\theta}^{\rm eff,\beta}(Y^2,X^2)'\frac{\beta}{\alpha} - \phi_{\theta}^{\rm eff,\gamma}(Y^2,X^2)\frac{\gamma}{\alpha}. \label{eq: efficient_score_alpha_specific_T=2}
\end{align}
Like its counterpart for $\gamma$ this expression does not appear in the main text.

\subsubsection{MPH efficient score without feedback (i.e., under strict exogeneity)}
\cite{hahn1994efficiency} derived the SEB for the MPH hazards model with $T=2$ and strictly exogenous regressors. His analysis did not include lagged duration dependence, as ours does. For completeness, we sketch the derivation of the efficient scores for $\alpha$, $\beta$ and $\gamma$ under strict exogeneity here. Details can be filled in along the lines of our derivation for the scores with feedback and/or by studying the rigorous analysis in \cite{hahn1994efficiency}. 

By direct analogy, the efficient score when $T=2$ under strict exogeneity equals
\begin{align*}
        \phi_{\theta}^{\rm eff,SE}(Y^2,X^2)=\mathbb{E}\left[S^{\theta}(Y^2,X^2)\,|\,Y_{0},\widetilde{P}_{1},\overline{P},X^2\right]-\mathbb{E}\left[S^{\theta}(Y^2,X^2)\,|\,Y_{0},\overline{P},X^2\right].    
\end{align*}
It follows that the expressions for the scores are the same as those derived above under feedback except that each expectation now conditions on all ``leads and lags" of the strictly exogenous covariates.

This yields an efficient score for $\beta$ of
\begin{align}
     \phi_{\theta}^{\rm eff,SE,\beta}(Y^2,X^2)&=\left(X_2-X_1\right)\left(\widetilde{P}_1-\frac{1}{2}\right)\overline{P}\mathbb{E}\left[e^{A}\,|\, Y_0,\overline{P},X^2\right]. \label{eq: efficient_score_beta_specific_T=2_strictexo}
\end{align}
This expression is identical to the one found by \cite{hahn1994efficiency}.

Maintaining the Weibull baseline hazard assumption, the efficient score for $\gamma$ equals:
\begin{multline}
    \phi_{\theta}^{\rm eff,SE,\gamma}(Y^2,X^2)=Y_1-\frac{\alpha}{1+\alpha}e^{-X_1'\frac{\beta}{\alpha}-\frac{\gamma}{\alpha}Y_0} \overline{P}^{\frac{1}{\alpha}}  \\
    -\left(Y_0\left(\widetilde{P}_1-\frac{1}{2}\right)+\left(Y_1(1-\widetilde{P}_1)-\left(\frac{\alpha}{1+\alpha}-\frac{\alpha}{1+2\alpha}\right)e^{-X_1'\frac{\beta}{\alpha}-\frac{\gamma}{\alpha}Y_0} \overline{P}^{\frac{1}{\alpha}}\right)\right) \\
    \times \overline{P}\mathbb{E}\left[e^{A}\,|\, Y_0,\overline{P},X^2\right]. \label{eq: efficient_score_gamma_specific_T=2_strictexo}
\end{multline}
The analysis of \cite{hahn1994efficiency} did not accommodate lagged duration dependence. The expression above is therefore new.
 
Finally, the efficient score for $\alpha$ equals:
\begin{align} \label{eq: efficient_score_alpha_specific_T=2_strictexo}
    \phi_{\theta}^{\rm eff,SE,\alpha}(Y^2,X^2)=&  \frac{1}{\alpha}\left(  2+\ln \widetilde{P}_1+\ln (1-\widetilde{P}_1)\right)  \\
    &-\frac{1}{\alpha}\left(\widetilde{P}_1\ln \widetilde{P}_1+(1-\widetilde{P}_1)\ln (1-\widetilde{P}_1)+\frac{1}{2}\right)\overline{P}\mathbb{E}\left[e^{A}\,|\,Y_0,\overline{P},X^2\right] \notag\\
    &- \phi_{\theta}^{\rm eff,SE,\beta}(Y^2,X^2)'\frac{\beta}{\alpha}- \phi_{\theta}^{\rm eff,SE,\gamma}(Y^2,X^2)\frac{\gamma}{\alpha}. \notag 
\end{align}
This expression coincides exactly with the one given in Lemma 3 of \cite{hahn1994efficiency} except that his formulation maintains the additional \emph{a priori} restriction $\gamma=0$. Setting $\gamma=0$ in the expression above yields Hahn's expression.

\subsubsection{MPH efficient score for average effects} \label{app: MPH_derivations : averageeffects}

 Lemma \ref{proj_lemma_MPH}, in Appendix \ref{app: MPH_derivations}, can also be used to derive expressions for the efficient moment function of average effects $\mu(\theta,\omega)$. \\
 
\noindent \textbf{Average structural hazard.} Consider first the ASH $\overline{\lambda}(y_{t}|x_t,y_{t-1})$ defined in \eqref{WMPH_ashf}. In the main text, we showed that $\varphi_{\theta}(Y^T,X^T)=\lambda_{\alpha}(y_{t})e^{x_{t}'\beta+\gamma y_{t-1}}\frac{T-1}{\overline{P}}$ is an identifying FHR moment function for $\overline{\lambda}(y_{t}|x_t,y_{t-1})$. Applying Lemma \ref{proj_lemma_MPH} yields the projection
\begin{align*}
        \Pi(\varphi_{\theta}(Y^{T},X^T)|\mathcal{T}_{\theta,\omega,L}^{\perp})=\sum_{t=1}^{T-1}\varphi_{\theta,t}^{\perp}(Y_0,\widetilde{P}^{t},\overline{P},X^{t})=0,
\end{align*}
since, for all $ t\in \{1,\ldots,T-1\}$, 
\begin{align*}
        \varphi_{\theta,t}^{\perp}(Y_0,\widetilde{P}^{t},\overline{P},X^{t}) &=\mathbb{E}\left[\varphi_{\theta}(Y^{T},X^T)|Y_0,\widetilde{P}^{t},\overline{P},X^{t}\right] -\mathbb{E}\left[\varphi_{\theta}(Y^{T},X^T)|Y_0,\widetilde{P}^{t-1},\overline{P},X^{t}\right] \\
        &=\lambda_{\alpha}(y_{t})e^{x_{t}'\beta+\gamma y_{t-1}}\frac{T-1}{\overline{P}}-\lambda_{\alpha}(y_{t})e^{x_{t}'\beta+\gamma y_{t-1}}\frac{T-1}{\overline{P}} \\
        &=0.
\end{align*}
This verifies the claim in Section \ref{sec: mph_deep_dive}. Recalling from Section \ref{subsec: efficiency} that the efficient moment function is given by ${\varphi}^{\rm eff}_{\theta,\omega}(Y^T,X^T)=\varphi_{\theta}(Y^T,X^T)-\Pi(\varphi_{\theta}(Y^T,X^T)\,|\,\mathcal{T}_{\theta,\omega,L}^{\perp})$, we conclude that the efficient moment function for the ASH is
\begin{align*}
    {\varphi}^{\rm eff}_{\theta,\omega}(Y^T,X^T)= \lambda_{\alpha}(y_{t})e^{x_{t}'\beta+\gamma y_{t-1}}\frac{T-1}{\overline{P}}.
\end{align*}

\noindent Formally, the preceding calculation requires $T\geq 3$ and $\mathbb{E}\left[e^{2A}\right]<\infty$ to ensure that the moment function ${\varphi}^{\rm eff}_{\theta,\omega}(Y^T,X^T)$ is square integrable. Indeed, by Lemma \ref{lem: helmert_tranformation_MPH}, $\overline{P}^{-1}\sim \mathrm{InverseGamma}\left(T,e^{A}\right)$ conditional on $(Y_0,X_1,A)$, so that
\begin{align*}
    \mathbb{E}\left[\frac{1}{\overline{P}^2}|Y_{0},X_1,A\right]=\frac{e^{2A}}{(T-1)(T-2)}.
\end{align*}
Consequently, $\mathbb{E}\left[\varphi^{\rm eff}_{\theta,\omega}(Y^T,X^T)^2\right]<\infty$ under the stated conditions. When $T=2$, however, $\mathbb{E}\left[\frac{1}{\overline{P}^2}|Y_{0},X_1,A\right]=+\infty$, resulting in $\varphi^{\rm eff}_{\theta,\omega}$ having infinite variance. In fact, the failure of square-integrability cannot be avoided by choosing a different FHR identifying moment function for the ASH. To see this, let $\varphi_{\theta}$ denote an alternative FHR moment function for $\overline{\lambda}(y_{2}|x_2,y_{1})$,  and consider the difference $\phi_{\theta}(Y^2,X^2)=\varphi_{\theta}(Y^2,X^2)-\varphi^{\rm eff}_{\theta,\omega}(Y^2,X^2)$. By construction,  $\phi_{\theta}$ constitutes an FHR moment function for $\theta$. Therefore, by Lemma \ref{lem_moment_charac_MPH}, there exists an absolutely integrable moment function $\psi_{\theta,1}$ such that $\phi_{\theta}\left(Y^{2},X^{2}\right)=\psi_{\theta,1}(Y_{0},\widetilde{P}_{1},\overline{P},X_{1})$,
with $\mathbb{E}\left[\psi_{\theta,1}(Y_{0},\widetilde{P}_{1},\overline{P},X_{1})|Y_{0},\overline{P},X_1\right]=0$. It follows that
\begin{align*}
     \mathbb{E}\left[\varphi_{\theta}(Y^2,X^2)|Y_{0},\overline{P},X_1\right]= \mathbb{E}\left[\varphi^{\rm eff}_{\theta,\omega}(Y^2,X^2)|Y_{0},\overline{P},X_1\right]=\varphi^{\rm eff}_{\theta,\omega}(Y^2,X^2).
\end{align*}
Then, an application of Jensen's inequality delivers
\begin{align*}
    \mathbb{E}\left[\varphi_{\theta}(Y^2,X^2)^2|Y_{0},\overline{P},X_1\right]\geq \mathbb{E}\left[\varphi^{\rm eff}_{\theta,\omega}(Y^2,X^2)|Y_{0},\overline{P},X_1\right]^2 = \varphi^{\rm eff}_{\theta,\omega}(Y^2,X^2)^2
\end{align*}
whereupon  $\mathbb{E}\left[\varphi_{\theta}(Y^2,X^2)^2\right]\geq \mathbb{E}\left[\varphi^{\rm eff}_{\theta,\omega}(Y^2,X^2)^2\right]=+\infty$. This argument shows that no FHR identifying moment function for the ASH is square-integrable, entailing a zero information bound for the case $T=2$.  This contrasts with $T\geq 3$ where the efficient moment function is square-integrable, provided again that $\mathbb{E}[e^{2A}]<\infty$. \\
\indent Finally, the fact that $\varphi^{\rm eff}_{\theta,\omega}$ has infinite variance for $T=2$ does not necessarily jeopardize consistency of the ASH estimator
\begin{align*}
    \widehat{\overline{\lambda}}(y_{2}|x_2,y_{1};\widehat{\theta})=\frac{1}{N}\sum_{i=1}^N\varphi^{\rm eff}_{\hat{\theta},\omega}(Y_i^2,X_i^2)=\lambda_{\widehat{\alpha}}(y_{2})e^{x_2'\widehat{\beta}+\widehat{\gamma} y_{1}}\frac{1}{N}\sum_{i=1}^N\frac{T-1}{\overline{P}_{i\widehat{\theta}}},
\end{align*}
where $\widehat{\theta}=(\widehat{\alpha},\widehat{\beta}',\widehat{\gamma})'$ is a consistent estimate of the common parameter $\theta$. We have
\begin{align*}
    &\widehat{\overline{\lambda}}(y_{2}|x_2,y_{1};\widehat{\theta})-\overline{\lambda}(y_{2}|x_2,y_{1}) \\
    &= \frac{1}{N}\sum_{i=1}^N \left[{\varphi}^{\rm eff}_{\hat{\theta},\omega}(Y_i^2,X_i^2)-{\varphi}^{\rm eff}_{\theta,\omega}(Y_i^2,X_i^2)\right]+\frac{1}{N}\sum_{i=1}^N \left[{\varphi}^{\rm eff}_{\theta,\omega}(Y_i^2,X_i^2)-\overline{\lambda}(y_{2}|x_2,y_{1})\right].
\end{align*}
Now, since $\mathbb{E}\left[\frac{1}{\overline{P}}|Y_{0},X_1,A\right]=e^{A}$ for $T=2$,  $\mathbb{E}\left[\abs{{\varphi}^{\rm eff}_{\theta,\omega}(Y_i^2,X_i^2)}\right]<\infty$ provided that $\mathbb{E}\left[e^{A}\right]<\infty$. It then follows that the second term converges to zero in probability by the weak law of large numbers. Thus, provided that the first term is $o_{p}(1)$, $\widehat{\overline{\lambda}}(y_{2}|x_2,y_{1};\widehat{\theta})$ will converge in probability to $\overline{\lambda}(y_{2}|x_2,y_{1})$. \\

\noindent \textbf{Average structural function.} As a second example, consider the ASF $\mu(x_t,y_{t-1})$ defined in \eqref{WMPH_asf}. Under a Weibull baseline hazard, it takes the form
\begin{align*}
    \mu(x_t,y_{t-1})=\Gamma\left(1+\frac{1}{\alpha}\right)\exp(-x_t'\frac{\beta}{\alpha}-\frac{\gamma}{\alpha}y_{t-1})\mathbb{E}\left[\exp(-\frac{A}{\alpha})\right].
\end{align*}
In view of \eqref{eq: gamma_moments}, one candidate identifying moment function is 
\begin{align*}
    \varphi_{\theta}(Y^T,X^T) = \exp(-x_t'\frac{\beta}{\alpha}-\frac{\gamma}{\alpha}y_{t-1})\frac{\Gamma\left(1+\frac{1}{\alpha}\right)\Gamma(T)}{\Gamma(T+\frac{1}{\alpha})}\overline{P}^{\frac{1}{\alpha}}.
\end{align*}
Applying Lemma \ref{proj_lemma_MPH}, we again get $\Pi(\varphi_{\theta}(Y^{T},X^T)|\mathcal{T}_{\theta,\omega,L}^{\perp})=0$, which implies that the efficient moment function for the ASF is given by
\begin{align} \label{eq: eff_score_weibull_ASF}
    {\varphi}^{\rm eff}_{\theta,\omega}(Y^T,X^T) = \exp(-x_t'\frac{\beta}{\alpha}-\frac{\gamma}{\alpha}y_{t-1})\frac{\Gamma\left(1+\frac{1}{\alpha}\right)\Gamma(T)}{\Gamma(T+\frac{1}{\alpha})}\overline{P}^{\frac{1}{\alpha}}.
\end{align}
Alternatively, based on \eqref{eq: rho_is_exponential}, one could for example consider the FHR moment function 
\begin{align*}
    \varphi^{2}_{\theta_0}(Y^T,X^T) = \exp(-x_t'\frac{\beta}{\alpha}-\frac{\gamma}{\alpha}y_{t-1})P_{1}^{\frac{1}{\alpha}},
\end{align*}
noting that $\mathbb{E}\left[P_{1}^{\frac{1}{\alpha}}|Y_{0},X_{1},A\right] = \Gamma\left(1+\frac{1}{\alpha}\right)\exp(-\frac{A}{\alpha})$. For this choice, the identity $P_{1}=\widetilde{P}_{1}\overline{P}$ and an application of Lemma \ref{proj_lemma_MPH} yields
\begin{align*}
    \Pi(\varphi_{\theta}^{2}(Y^{T},X^T)|\mathcal{T}_{\theta,\omega,L}^{\perp}) &=  \exp(-x_t'\frac{\beta}{\alpha}-\frac{\gamma}{\alpha}y_{t-1})\left(\widetilde{P}_{1}^{\frac{1}{\alpha}}\overline{P}^{\frac{1}{\alpha}} -\mathbb{E}\left[\widetilde{P}_{1}^{\frac{1}{\alpha}}|Y_0,\overline{P},X^{1}\right]\overline{P}^{\frac{1}{\alpha}}\right) \\
    &=\exp(-x_t'\frac{\beta}{\alpha}-\frac{\gamma}{\alpha}y_{t-1})\left(\widetilde{P}_{1}^{\frac{1}{\alpha}}\overline{P}^{\frac{1}{\alpha}} -\frac{\Gamma\left(1+\frac{1}{\alpha}\right)\Gamma(T)}{\Gamma(T+\frac{1}{\alpha})}\overline{P}^{\frac{1}{\alpha}}\right) \\
    &=\varphi_{\theta}^{2}(Y^{T},X^T) - \exp(-x_t'\frac{\beta}{\alpha}-\frac{\gamma}{\alpha}y_{t-1})\frac{\Gamma\left(1+\frac{1}{\alpha}\right)\Gamma(T)}{\Gamma(T+\frac{1}{\alpha})}\overline{P}^{\frac{1}{\alpha}},
\end{align*}
where the second equality leverages the distributional properties reported in Lemma \ref{lem: helmert_tranformation_MPH}. The projection residual would then deliver the same expression of the efficient moment function as \eqref{eq: eff_score_weibull_ASF}. This aligns with the result discussed in Section \ref{subsec: efficiency} that the efficient moment function for average effects is invariant to the specific choice of $\varphi_{\theta}$. \\

\noindent \textbf{Feedback projection coefficient.} As our third example consider the feedback projection coefficient, $\mu$, defined in \eqref{WMPH_fpc}. It corresponds to the slope coefficient on $Y_1$ in the population linear regression of $X_2$ on $(Y_1,R')$, where $R = (Y_0,X_1,A)$ . By the Frisch-Waugh-Lovell theorem, we have
\begin{align*}
    \mu &= \frac{\mathbb{E}[X_2Y_1]-\mathbb{E}[X_2R']\mathbb{E}[RR']^{-1}\mathbb{E}[Y_1R]}{\mathbb{E}[Y_1^2]-\mathbb{E}[Y_1R']\mathbb{E}[RR']^{-1}\mathbb{E}[Y_1R]}
\end{align*}
which involves the average effects $\mathbb{E}[Y_0A],\mathbb{E}[X_1A], \mathbb{E}[Y_1A], \mathbb{E}[X_2A], \mathbb{E}[A^2]$. We will show that they each admit FHR identifying moments. To this end, define $Q_t=-\boldsymbol{\gamma}-\ln P_t$, where $\boldsymbol{\gamma}$ denotes Euler's constant. It follows from distributional properties of exponential variates and Lemma \ref{lem: P1_PT_are_exponential} that 
\begin{align*}
    &Q_t|Y_{0},P^{t-1},X^{t},A \sim\mathrm{Gumbel}\left(-\boldsymbol{\gamma}+A,1\right), \quad t=1,\ldots,T \\
    &Q_t|Y_{0},X_1,A \sim\mathrm{Gumbel}\left(-\boldsymbol{\gamma}+A,1\right), \quad t=1,\ldots,T.
\end{align*}
In particular, this implies $\mathbb{E}\left[Q_t|Y_0,X_1,A\right]=A$.
Using this result, and successive applications of the law of iterated expectations, we obtain
\begin{align*}
    &\mathbb{E}\left[Y_0Q_1\right]= \mathbb{E}\left[Y_0\mathbb{E}\left[Q_1|Y_0,X_1,A\right]\right]= \mathbb{E}\left[Y_0A\right] \\
    &\mathbb{E}\left[X_1Q_1\right]= \mathbb{E}\left[X_1\mathbb{E}\left[Q_1|Y_0,X_1,A\right]\right]= \mathbb{E}\left[X_1A\right] \\
    &\mathbb{E}\left[Y_1Q_2\right]= \mathbb{E}\left[Y_1\mathbb{E}\left[Q_2|Y_0,P_1,X^2,A\right]\right]= \mathbb{E}\left[Y_1A\right] \\
    &\mathbb{E}\left[X_2Q_2\right]= \mathbb{E}\left[X_2\mathbb{E}\left[Q_2|Y_0,P_1,X^2,A\right]\right]= \mathbb{E}\left[X_2A\right] \\
    &\mathbb{E}\left[Q_1Q_2\right]= \mathbb{E}\left[Q_1\mathbb{E}\left[Q_2|Y_0,P_1,X^2,A\right]\right]= \mathbb{E}\left[A^2\right]
\end{align*}
With a Weibull baseline hazard, Theorem \ref{theo_identif_weibull_MPH} establishes point-identification of $\theta_0$ under standard regularity conditions. In that case,  $\mu$ is identified from using the collection of FHR identifying moment functions derived above. This example illustrates that, even in short panels, it may be possible to learn features of the feedback process, similarly to how it is possible to learn features of the unobserved heterogeneity distribution in some models. \\ 
\indent Using Lemma \ref{proj_lemma_MPH} we can derive the forms of the efficient scores for each of the expectations entering the definition of $\mu$. The efficient moment function ${\varphi}^{\rm eff}_{\theta,\omega}(Y^T,X^T)$ for $\mu$ follows by applying the delta method to  this vector of efficient scores.

\subsection{Detailed calculations}\label{app: detailed calculations}

\textbf{Derivation of equation \eqref{eq: det_for_rho_tilde_rho_bar_change_of_variables}:} 
Using the inverse mapping defined in Appendix \ref{app: MPH_derivations} we can write the (determinant of the) Jacobian of the mapping from $\mathbf{V}$ back into $\mathbf{P}$ as:
\begin{align*}
J^{(T)}(V_{1},\ldots,V_{T})=\left|\begin{smallmatrix}V_{T} & 0 & 0 & \ldots & 0 & V_{1}\\
-V_{2}V_{T} & (1-V_{1})V_{T} & 0 & 0 & \ldots & (1-V_{1})V_{2}\\
-(1-V_{2})V_{3}V_{T} & -(1-V_{1})V_{3}V_{T} & (1-V_{1})(1-V_{2})V_{T} & 0 & \ldots & (1-V_{1})(1-V_{2})V_{3}\\
\vdots & \vdots & \vdots & \vdots & \ddots & \vdots\\
-\prod_{s=2}^{T-2}(1-V_{s})V_{T-1}V_{T} & \ldots & \ldots & \prod_{s=1}^{T-2}(1-V_{s})V_{T} &  & \prod_{s=1}^{T-2}(1-V_{s})V_{T-1}\\
-\prod_{s=2}^{T-1}(1-V_{s})V_{T} & \ldots & \ldots & -\prod_{s=1}^{T-2}(1-V_{s})V_{T} &  & \prod_{s=1}^{T-1}(1-V_{s})
\end{smallmatrix}\right|
\end{align*}
A Laplace (co-factor) expansion along the first row gives 
\begin{align*}
J^{(T)}(V_{1},\ldots,V_{T}) & =V_{T}\left|\begin{smallmatrix}(1-V_{1})V_{T} & 0 & 0 & \ldots & (1-V_{1})V_{2}\\
-(1-V_{1})V_{3}V_{T} & (1-V_{1})(1-V_{2})V_{T} & 0 & \ldots & (1-V_{1})(1-V_{2})V_{3}\\
\vdots & \vdots & \vdots & \ddots & \vdots\\
-(1-V_{1})\prod_{s=3}^{T-2}(1-V_{s})V_{T-1}V_{T} & \ldots & \prod_{s=1}^{T-2}(1-V_{s})V_{T} &  & \prod_{s=1}^{T-2}(1-V_{s})V_{T-1}\\
-(1-V_{1})\prod_{s=3}^{T-1}(1-V_{s})V_{T} & \ldots & -\prod_{s=1}^{T-2}(1-V_{s})V_{T} &  & \prod_{s=1}^{T-1}(1-V_{s})
\end{smallmatrix}\right|\\
 & +(-1)^{T+1}V_{1}\left|\begin{smallmatrix}-V_{2}V_{T} & (1-V_{1})V_{T} & 0 & 0\\
-(1-V_{2})V_{3}V_{T} & -(1-V_{1})V_{3}V_{T} & (1-V_{1})(1-V_{2})V_{T} & 0\\
\vdots & \vdots & \vdots & \vdots\\
-\prod_{s=2}^{T-2}(1-V_{s})V_{T-1}V_{T} & \ldots & \ldots & \prod_{s=1}^{T-2}(1-V_{s})V_{T}\\
-\prod_{s=2}^{T-1}(1-V_{s})V_{T} & \ldots & \ldots & -\prod_{s=1}^{T-2}(1-V_{s})V_{T}
\end{smallmatrix}\right|
\end{align*}
and further factorizing common factors across columns yields 
\begin{align*}
J^{(T)}(V_{1},\ldots,V_{T}) & =V_{T}(1-V_{1})^{T-1}\left|\begin{smallmatrix}V_{T} & 0 & 0 & \ldots & V_{2}\\
-V_{3}V_{T} & (1-V_{2})V_{T} & 0 & \ldots & (1-V_{2})V_{3}\\
\vdots & \vdots & \vdots & \ddots & \vdots\\
-\prod_{s=3}^{T-2}(1-V_{s})V_{T-1}V_{T} & \ldots & \prod_{s=2}^{T-2}(1-V_{s})V_{T} &  & \prod_{s=2}^{T-2}(1-V_{s})V_{T-1}\\
-\prod_{s=3}^{T-1}(1-V_{s})V_{T} & \ldots & -\prod_{s=2}^{T-2}(1-V_{s})V_{T} &  & \prod_{s=2}^{T-1}(1-V_{s})
\end{smallmatrix}\right|\\
\\
 & +(-1)^{T}V_{1}(1-V_{1})^{T-2}V_{T}\left|\begin{smallmatrix}V_{2} & V_{T} & 0 & 0\\
(1-V_{2})V_{3} & -V_{3}V_{T} & (1-V_{2})V_{T} & 0\\
\vdots & \vdots & \vdots & \vdots\\
\prod_{s=2}^{T-2}(1-V_{s})V_{T-1} & \ldots & \ldots & \prod_{s=2}^{T-2}(1-V_{s})V_{T}\\
\prod_{s=2}^{T-1}(1-V_{s}) & \ldots & \ldots & -\prod_{s=2}^{T-2}(1-V_{s})V_{T}
\end{smallmatrix}\right|.
\end{align*}
Observe that the matrices involved in the first and second determinants
have the same set of columns. More specifically, the second matrix coincides with the first one after relocating its first column to the last position. Such an operation requires $(T-2)$ column swaps, implying that the determinant of the second matrix is $(-1)^{T-2}$ that of the first matrix. \\
Next, observe that the first determinant has exactly the same structure
as $J^{(T)}(V_{1},\ldots,V_{T})$, except that it involves $T-1$
variables instead of $T$ variables, namely $(V_{2},\ldots,V_{T})$
instead of $(V_{1},\ldots,V_{T})$. More precisely, this determinant
coincides with $J^{(T-1)}(V_{2},\ldots,V_{T})$. 

Together, these observations imply that: 
\begin{align*}
J^{(T)}(V_{1},\ldots,V_{T}) & =V_{T}(1-V_{1})^{T-1}J^{(T-1)}(V_{2},\ldots,V_{T})+(-1)^{T+T-2}V_{1}(1-V_{1})^{T-2}V_{T}J^{(T-1)}(V_{2},\ldots,V_{T})\\
 & =(1-V_{1})^{T-2}V_{T}J^{(T-1)}(V_{2},\ldots,V_{T})\\
 & =(1-V_{1})^{T-2}V_{T}(1-V_{2})^{T-3}V_{T}J^{(T-2)}(V_{3},\ldots,V_{T})\\
 & =(1-V_{1})^{T-2}V_{T}(1-V_{2})^{T-3}V_{T}\times\ldots\times(1-V_{T-2})V_{T}J^{(2)}(V_{T-1},V_{T})\\
 & =(1-V_{1})^{T-2}V_{T}(1-V_{2})^{T-3}V_{T}\times\ldots\times(1-V_{T-2})V_{T}V_{T}\\
 & =V_{T}^{T-1}\prod_{s=1}^{T-2}(1-V_{s})^{T-(s+1)},
\end{align*}
which coincides with \eqref{eq: det_for_rho_tilde_rho_bar_change_of_variables} as required.

\textbf{Derivation of equation \eqref{eq: density_of_rho_tilde_rho_bar}}
The change-of-variables formula, in conjunction with the determinant given in \eqref{eq: det_for_rho_tilde_rho_bar_change_of_variables}, yields   
\begin{align*}
f_{V}(V_{1},V_{2},\ldots,V_{T})= & f_{P_{1}}\left(V_{1}V_{T}\right)\prod_{t=2}^{T-1}f_{P_{t}}\left(\prod_{s=1}^{t-1}(1-V_{s})V_{t}V_{T}\right)f_{P_{T}}\left(\prod_{s=1}^{T-1}(1-V_{s})V_{T}\right)\abs{\det\left[\dv{G^{-1}(\mathbf{V})}{\mathbf{V}}\right]}\\
= & \frac{\beta^{\alpha_{1}}}{\Gamma(\alpha_{1})}(V_{1}V_{T})^{\alpha_{1}-1}e^{-\beta V_{1}V_{T}}\\
 & \times\prod_{t=2}^{T-1}\frac{\beta^{\alpha_{t}}}{\Gamma(\alpha_{t})}\left(\prod_{s=1}^{t-1}(1-V_{s})V_{t}V_{T}\right)^{\alpha_{t}-1}e^{-\beta\prod_{s=1}^{t-1}(1-V_{s})V_{t}V_{T}}\\
 & \times\frac{\beta^{\alpha_{T}}}{\Gamma(\alpha_{T})}\left(\prod_{s=1}^{T-1}(1-V_{s})V_{T}\right)^{\alpha_{T}-1}e^{-\beta\prod_{s=1}^{T-1}(1-V_{s})V_{T}}\\
 & \times V_{T}^{T-1}\prod_{s=1}^{T-2}(1-V_{s})^{T-(s+1)}
\\
= & \frac{\beta^{\sum_{t=1}^{T}\alpha_{t}}}{\Gamma(\alpha_{T})}V_{T}^{\sum_{t=1}^{T}\alpha_{t}-1}e^{-\beta V_{T}}\\
 & \times\frac{1}{\Gamma(\alpha_{1})}V_{1}^{\alpha_{1}-1}(1-V_{1})^{\sum_{t=2}^{T}\alpha_{t}-(T-1)+(T-(1+1))}\\
 & \times\frac{1}{\Gamma(\alpha_{2})}V_{2}^{\alpha_{2}-1}(1-V_{2})^{\sum_{t=3}^{T}\alpha_{t}-(T-2)+(T-(2+1))}\\
 & \vdots\\
 & \times\frac{1}{\Gamma(\alpha_{T-1})}V_{T-1}^{\alpha_{T-1}-1}(1-V_{T-1})^{\alpha_{T}-1} \\ 
= & \frac{\beta^{\sum_{t=1}^{T}\alpha_{t}}}{\Gamma\left(\sum_{t=1}^{T}\alpha_{t}\right)}V_{T}^{\sum_{t=1}^{T}\alpha_{t}-1}e^{-\beta V_{T}}\times\prod_{t=1}^{T-1}\frac{\Gamma\left(\sum_{s=t}^{T}\alpha_{s}\right)}{\Gamma\left(\alpha_{t}\right)\Gamma\left(\sum_{s=t+1}^{T}\alpha_{s}\right)}V_{t}^{\alpha_{t}-1}(1-V_{t})^{\sum_{s=t+1}^{T}\alpha_{s}-1},
\end{align*}
where the simplification on the penultimate equality follows from:
\begin{align*}
V_{T}=\sum_{t=1}^{T}P_{t}=V_{1}V_{T}+\sum_{t=2}^{T-1}\prod_{s=1}^{t-1}(1-V_{s})V_{t}V_{T}+\prod_{s=1}^{T-1}(1-V_{s})V_{T}.
\end{align*}
This corresponds to equation \eqref{eq: density_of_rho_tilde_rho_bar} as needed.

\section{Additional details for the numerical experiments} \label{app: numericalexperiment_details}

In this section we again omit the dependence on $\theta$ in $P_{\theta}$ and related quantities. 

\subsection{Material for the results in subsection \ref{subsec: numerical_experiment}}

\begin{lemma} \label{lem: nofeedback_MPH} 
Consider the
MPH model defined by Example \ref{ex: mph_intro} with $T=2$. If there is no feedback, then (i) $\overline{P},X_2, \widetilde{P}_1$ are independent conditional on $Y_0,X_1,A$, and (ii) $\widetilde{P}_1$ is independent of $X_2, \overline{P},A,Y_{0},X_{1}$. If in addition $X_{2}$ is independent of $A$ conditional on $Y_0,X_1$, then (iii) $\overline{P},X_2, \widetilde{P}_1$ are independent conditional on $Y_0,X_1$
\end{lemma}
\begin{proof}
     If there is no feedback, \eqref{eq: no_feedback_restriction} and Lemma \ref{lem: helmert_tranformation_MPH}  imply that the joint density of $(\overline{P},X_{2},\widetilde{P}_1)$ conditional on $(Y_0,X_1,A)$ is given by
     \begin{align*}
         f(\overline{p},x_{2},\widetilde{p}_1|y_{0},x_{1},a)=\overline{p} \exp\left(2a-\overline{p}e^{a}\right)\mathds{1}\{\overline{p}>0\}g(x_2|y_0,x_1,a)\mathds{1}\left\{\widetilde{p}_1\in (0,1)\right\},
     \end{align*}
     which establishes (i). This also proves that $\widetilde{P}_1\perp A,Y_0,X_1$ and hence (ii). If $X_{2}$ is independent of $A$ conditional on $Y_0,X_1$, $g(x_2|y_0,x_1,a)=g(x_2|y_0,x_1)$ whereupon the form of the above joint density implies (iii).
\end{proof}

\noindent \textbf{Locally efficient score}. 
The working model in our numerical experiment is $\widetilde{\omega}=(\widetilde{g},\widetilde{\pi})$ where
$\widetilde{g}$ is a $\text{Bernoulli}(p)$ and $\widetilde{\pi}(v) = \frac{1}{v} \mathds{1}\{v > 0\}$ where $V= \exp(A)$. Since $\overline{P}|Y_{0},X_{1},V \sim\mathrm{Gamma}\left(T,V\right)$ by Lemma \ref{lem: helmert_tranformation_MPH}, it follows from Bayes rule that $V|Y_0,\overline{P},X_1\sim \mathrm{Gamma}\left(T,\overline{P}\right)$. Therefore, $\mathbb{E}_{\theta,\widetilde{\omega}}\left[V\,|\, Y_0,\overline{P},X_1\right]=\frac{T}{\overline{P}}$.  In view of \eqref{eq: efficient_score_beta_specific_T=2},  the locally efficient score for $\beta$ under $\widetilde{\omega}$ is
\begin{align}
    \widetilde{\phi}_{\theta,\widetilde{\omega}}^{\rm eff,\beta}(Y^2,X^2)&=-X_1\left(\widetilde{P}_1-\frac{1}{2}\right)\overline{P}\mathbb{E}_{\theta,\widetilde{\omega}}\left[V\,|\, Y_0,\overline{P},X_1\right]  \notag \\
    &+\mathbb{E}_{\theta,\widetilde{\omega}}\left[X_{2}|Y_0,\widetilde{P}_1,\overline{P},X_1\right] -\mathbb{E}_{\theta,\widetilde{\omega}}\left[X_{2}(1-\widetilde{P}_1)\overline{P}V|Y_0,\widetilde{P}_1,\overline{P},X_1\right] \notag\\
    &-\mathbb{E}_{\theta,\widetilde{\omega}}\left[X_{2}|Y_0,\overline{P},X_1\right]+\mathbb{E}_{\theta,\widetilde{\omega}}\left[X_{2}(1-\widetilde{P}_1)\overline{P}V|Y_0,\overline{P},X_1\right] \displaybreak[0] \notag\\
    &=-X_1\left(\widetilde{P}_1-\frac{1}{2}\right)\overline{P}\mathbb{E}_{\theta,\widetilde{\omega}}\left[V\,|\, Y_0,\overline{P},X_1\right] \notag\\
    &+\mathbb{E}_{\theta,\widetilde{\omega}}\left[X_{2}|Y_0,X_1\right] -(1-\widetilde{P}_1)\overline{P}\mathbb{E}_{\theta,\widetilde{\omega}}\left[X_{2}|Y_0,X_1\right]\mathbb{E}_{\theta,\widetilde{\omega}}\left[V|Y_0,\overline{P},X_1\right] \notag \\
    &-\mathbb{E}_{\theta,\widetilde{\omega}}\left[X_{2}|Y_0,X_1\right]+\overline{P}\mathbb{E}_{\theta,\widetilde{\omega}}\left[X_{2}|Y_0,X_1\right]\mathbb{E}_{\theta,\widetilde{\omega}}\left[(1-\widetilde{P}_1)\right] \mathbb{E}_{\theta,\widetilde{\omega}}\left[V|Y_0,\overline{P},X_1\right] \notag \\
    &=\left(p-X_1\right)\left(\widetilde{P}_1-\frac{1}{2}\right)T,  \label{eq: loceff_score_beta_T=2}
\end{align}
where the second equality follows from  the implications of Lemma \ref{lem: nofeedback_MPH}, and the third equality uses $\mathbb{E}_{\theta,\widetilde{\omega}}\left[V\,|\, Y_0,\overline{P},X_1\right]=\frac{T}{\overline{P}}$, $\mathbb{E}_{\theta,\widetilde{\omega}}\left[\widetilde{P}_1\right]=\frac{1}{2}$ by Lemma \ref{lem: helmert_tranformation_MPH}, and the definition of $\widetilde{g}$. For the remaining components of the efficient score, we use once more, $\mathbb{E}_{\theta,\widetilde{\omega}}\left[V\,|\, Y_0,\overline{P},X_1\right]=\frac{T}{\overline{P}}$ and obtain
\begin{align} 
    \widetilde{\phi}_{\theta,\widetilde{\omega}}^{\rm eff,\gamma}(Y^2,X^2)&=Y_1-\frac{\alpha}{1+\alpha}e^{-X_1'\frac{\beta}{\alpha}-\frac{\gamma}{\alpha}Y_0} \overline{P}^{\frac{1}{\alpha}}  \notag \\
    &-\left(Y_0\left(\widetilde{P}_1-\frac{1}{2}\right)+\left(Y_1(1-\widetilde{P}_1)-\left(\frac{\alpha}{1+\alpha}-\frac{\alpha}{1+2\alpha}\right)e^{-X_1'\frac{\beta}{\alpha}-\frac{\gamma}{\alpha}Y_0} \overline{P}^{\frac{1}{\alpha}}\right)\right)T \label{eq: loceff_score_gamma_T=2}
\end{align}
from \eqref{eq: efficient_score_gamma_specific_T=2}, and finally
\begin{align}
    \widetilde{\phi}_{\theta,\widetilde{\omega}}^{\rm eff,\alpha}(Y^2,X^2) =&  \frac{1}{\alpha}\left(  2+\ln \widetilde{P}_1+\ln (1-\widetilde{P}_1)\right) \notag \\
    &-\frac{1}{\alpha}\left(\widetilde{P}_1\ln \widetilde{P}_1+(1-\widetilde{P}_1)\ln (1-\widetilde{P}_1)+\frac{1}{2}\right)T  \notag \\
    &- \widetilde{\phi}_{\theta,\widetilde{\omega}}^{\rm eff,\beta}(Y^2,X^2)'\frac{\beta}{\alpha}- \widetilde{\phi}_{\theta,\widetilde{\omega}}^{\rm eff,\gamma}(Y^2,X^2)\frac{\gamma}{\alpha} \label{eq: loceff_score_alpha_T=2}
\end{align}
from \eqref{eq: efficient_score_alpha_specific_T=2}. \\

\noindent \textbf{Implementation details.} The asymptotic standard errors reported in Table \ref{tab1_montecarlo_SE} are obtained via Monte Carlo integration using  $N=1,000,000$ simulation draws from the DGP of Experiment (A).  As a benchmark, we computed the square root of the diagonal elements of  $\mathbb{E}\left[\phi_{\theta_0}^{\rm eff,SE}(Y_{i}^2, X_{i}^2)\phi_{\theta_0}^{\rm eff,SE}(Y_{i}^2, X_{i}^2)'\right]^{-1}$, which corresponds to the semiparametric efficiency bound under strict exogeneity.  This quantity serves as a reference for the remaining estimators, which explains the normalization in the first row. Similarly, the second row displays the square root of the (normalized) diagonal elements of $\mathbb{E}\left[\phi_{\theta_0,\omega_0}^{\rm eff}(Y_{i}^2, X_{i}^2)\phi_{\theta_0,\omega_0}^{\rm eff}(Y_{i}^2, X_{i}^2)'\right]^{-1}$, i.e., the semiparametric efficiency bound with feedback. For the locally efficient estimator and the simple moments approach, we report the (normalized) asymptotic standard errors using the method-of-moments variance formula:
$H^{-1}V(H^{-1})'$. In the case of the locally efficient estimator, the matrices are defined as  $H=\mathbb{E}\left[\pdv{\widetilde{\phi}_{\theta_0,\widetilde{\omega}}^{\rm eff}(Y_{i}^2, X_{i}^2)}{\theta}\right]$ and $V = \mathbb{E}\left[\widetilde{\phi}_{\theta_0,\widetilde{\omega}}^{\rm eff}(Y_{i}^2, X_{i}^2)\widetilde{\phi}_{\theta_0,\widetilde{\omega}}^{\rm eff}(Y_{i}^2, X_{i}^2)'\right]$. For the simple moment-based estimator, we use the same expressions, replacing $\widetilde{\phi}_{\theta_0,\widetilde{\omega}}^{\rm eff}(Y_{i}^2, X_{i}^2)$ with $\phi_{\theta_0}(Y_{i}^2, X_{i}^2)$.\\
\indent In Table \ref{tab1_montecarlo_FB}, we follow an analogous procedure with $N=1,000,000$ simulation draws from the DGP of Experiment (B) to evaluate population expectations. The main difference is that we now treat the square root of the diagonal elements of $\mathbb{E}\left[\phi_{\theta_0,\omega_0}^{\rm eff}(Y_{i}^2, X_{i}^2)\phi_{\theta_0,\omega_0}^{\rm eff}(Y_{i}^2, X_{i}^2)'\right]^{-1}$ as the benchmark for efficiency comparisons. Finally, in both experiments, we evaluate $p=\mathbb{E}\left[X_2\right]$ via Monte Carlo integration to compute the locally efficient score for $\beta$ in equation \eqref{eq: loceff_score_beta_T=2}. \\

\noindent \textbf{Computation of efficient scores.} In the absence of feedback, Lemma \ref{lem: nofeedback_MPH} and Lemma \ref{lem: helmert_tranformation_MPH} imply that \eqref{eq: efficient_score_beta_specific_T=2} simplifies to
\begin{align}
    \phi_{\theta}^{\rm eff,\beta}(Y^2,X^2)=\left(\widetilde{P}_1-\frac{1}{2}\right)\overline{P}\left(\mathbb{E}\left[X_{2}V|Y_0,\overline{P},X_1\right]-X_1\mathbb{E}\left[V\,|\, Y_0,\overline{P},X_1\right]\right) \label{eq: efficient_score_beta_specific_T=2_nofeedback}
\end{align}
where $V= \exp(A)$. As a result, the efficient score for $\alpha$ in  \eqref{eq: efficient_score_alpha_specific_T=2} also simplifies. On the other hand, the  efficient score for $\gamma$ in \eqref{eq: efficient_score_gamma_specific_T=2} remains unchanged. To compute the efficient score, we must evaluate two quantities: $\mathbb{E}\left[X_2 V \,|\,Y_0,\overline{P},X_1\right]$ and $\mathbb{E}\left[V\,|\,Y_0,\overline{P},X_1\right]$.  In both numerical experiments $V\sim \text{Gamma}(\kappa_0,\lambda_0)$ where $\kappa_0=\lambda_0=5$. Given that $\overline{P}|Y_{0},X_{1},V \sim\mathrm{Gamma}\left(T,V\right)$ by Lemma \ref{lem: helmert_tranformation_MPH}, Bayes rule implies that $V|Y_0,\overline{P},X_1\sim \mathrm{Gamma}\left(T+\kappa_0,\overline{P}+\lambda_0\right)$ and hence $\mathbb{E}\left[V\,|\, Y_0,\overline{P},X_1\right]=\frac{T+\kappa_0}{\overline{P}+\lambda_0}$. Furthermore, the success probability for $X_2$ in Experiment (A) is $p(Y_0,Y_1,X_1,V)=1-\exp(-\tau_{A}(Y_{0},X_{1},Y_{1})V)$ with $\tau_{A}(Y_{0},X_{1},Y_{1})=Y_{0}+X_{1}$. Using this specification, we derive:
\begin{align*}
    \mathbb{E}\left[X_2V\,|\, Y_0,\overline{P},X_1\right]= \frac{T+\kappa_0}{\left(\lambda_0+\overline{P}\right)} - (T+\kappa_0)\frac{\left(\lambda_0+\overline{P}\right)^{T+\kappa_0}}{\left(\lambda_0+\overline{P}+\tau_{A}(Y_{0},X_{1},Y_{1})\right)^{T+\kappa_0+1}},
\end{align*}
and 
\begin{align*}
    \mathbb{E}\left[e^{A}|Y_0,X_1,X_2,\overline{P}\right]&=X_2 \frac{w_1}{w_1-w_2}\frac{T+\kappa_0}{(\lambda_0+\overline{P})}+\left((1-X_2)-X_2\frac{w_2}{w_1-w_2}\right)\frac{T+\kappa_0}{(\lambda_0+\overline{P}+\tau_{A}(Y_{0},X_{1},Y_{1}))},
\end{align*}
for $w_1=\frac{\overline{P}^{T-1}\lambda_0^{\kappa_0}}{\Gamma(T)\Gamma(\kappa_0)}\frac{\Gamma(T+\kappa_0)}{\left(\lambda_0+\overline{P}\right)^{T+\kappa_0}}$, $w_2=\frac{\overline{P}^{T-1}\lambda_0^{\kappa_0}}{\Gamma(T)\Gamma(\kappa_0)}\frac{\Gamma(T+\kappa_0)}{\left(\lambda_0+\overline{P}+\tau_{A}(Y_{0},X_{1},Y_{1})\right)^{T+\kappa_0}}$. This last expression enables us to compute the efficient score under strict exogeneity given in \eqref{eq: efficient_score_beta_specific_T=2_strictexo}-\eqref{eq: efficient_score_gamma_specific_T=2_strictexo}-\eqref{eq: efficient_score_alpha_specific_T=2_strictexo} for Experiment (A). \\
\indent Equations \eqref{eq: efficient_score_beta_specific_T=2}–\eqref{eq: efficient_score_gamma_specific_T=2}–\eqref{eq: efficient_score_alpha_specific_T=2} reveal that, in the presence of feedback, computing the efficient score requires evaluating the following four conditional expectations: $ \mathbb{E}\left[X_{2}|Y_0,\widetilde{P}_1,\overline{P},X_1\right]$, $\mathbb{E}\left[X_{2}|Y_0,\overline{P},X_1\right]$, $ \mathbb{E}\left[X_{2}V|Y_0,\widetilde{P}_1,\overline{P},X_1\right]$, $ \mathbb{E}\left[X_{2}(1-\widetilde{P}_1)V|Y_0,\overline{P},X_1\right]$. In Experiment (B), the success probability for $X_2$ is  $p(Y_0,Y_1,X_1,V)=1-\exp(-\tau_{B}(Y_{0},X_{1},Y_{1})V)$ with $\tau_{B}(Y_{0},X_{1},Y_{1})=Y_{0}+X_{1}+Y_{1}$. Under this specification, we obtain the following expressions:
\begin{align*}
    \mathbb{E}\left[X_2|Y_0,X_1,\widetilde{P}_1,\overline{P}\right]
    &=1-\frac{\left(\lambda_0+\overline{P}\right)^{T+\kappa_0}}{\left(\lambda_0+\overline{P}+\tau_{B}(Y_{0},X_{1},Y_{1}) \right)^{T+\kappa_0}} \\
    \mathbb{E}\left[X_2|Y_0,X_1,\overline{P}\right]
    &=1-\left(\lambda_0+\overline{P}\right)^{T+\kappa_0}\times \frac{1}{C_1^{(T+\kappa_0)}} \times {}_2F_1\left(T+\kappa_0, \alpha; 1+ \alpha; -\frac{C_2}{C_1}\right) \\
    \mathbb{E}\left[X_2V\,|\, Y_0,\widetilde{P}_1,\overline{P},X_1\right]
    &=\frac{T+\kappa_0}{\lambda_0+\overline{P}}-(T+\kappa_0)\frac{\left(\lambda_0+\overline{P}\right)^{T+\kappa_0}}{\left(\lambda_0+\overline{P}+\tau_{B}(Y_{0},X_{1},Y_{1})  \right)^{T+\kappa_0+1}},
\end{align*}
and 
\begin{align*}
    &\mathbb{E}\left[X_2(1-\widetilde{P}_1)V\,|\, Y_0,\overline{P},X_1\right] \\
    &=\frac{1}{2}\frac{T+\kappa_0}{\lambda_0+\overline{P}} -(T+\kappa_0)\left(\lambda_0+\overline{P}\right)^{T+\kappa_0}\times \\
    &\frac{1}{C_1^{(T+\kappa_0+1)}}\left({}_2F_1\left(T+\kappa_0+1, \alpha; 1+ \alpha; -\frac{C_2}{C_1}\right)-\frac{1}{2}\times {}_2F_1\left(T+\kappa_0+1,2\alpha; 1+ 2\alpha; -\frac{C_2}{C_1}\right)\right),
\end{align*}
where we use the shorthands $C_1=\lambda_0+\overline{P}+Y_0+X_1,C_2=\overline{P}^{\frac{1}{\alpha}}e^{-X_1'\frac{\beta}{\alpha}-\frac{\gamma}{\alpha}Y_0}$, and ${}_2F_1\left(a, b; c; z\right)$ denotes the hypergeometric function
\begin{align*}
    {}_2F_1\left(a, b; c; z\right) = \frac{\Gamma(c)}{\Gamma(b)\Gamma(c-b)}\int_{0}^1 \frac{t^{b-1}(1-t)^{c-b-1}}{(1-zt)^a}\mathrm{d}t.
\end{align*}

\subsection{Auxiliary numerical experiment with known homogeneous feedback}

This subsection investigates the efficiency cost of leaving the feedback process unrestricted when its structure is actually known.  We focus on the case where the true feedback process is homogeneous and $T=2$.

\paragraph{Orthocomplement of tangent set.} When the feedback process $g$ is known to the econometrician, the nuisance parameter reduces to $\omega=(\pi,\nu)\in {\cal{P}}\times \mathcal{N}$ which in turn alters the characterization of $\mathcal{T}_{\theta_0,\omega_0,K}^{\perp}$ presented in Theorem \ref{theo_eff}. However, by suitably adapting the proof of Lemma \ref{lem_moment_perp_score}, one can show that $\mathcal{T}_{\theta_0,\omega_0,K}^{\perp}$ consists instead in the set of square-integrable, moment functions $\phi_{\theta_0}$ satisfying
\begin{align} \label{eq_phi_hetero_robust_knownhomogfeedback}
    0=\int \phi_{\theta_0}(y^2,x^2)f_{\theta_0}(y_{2}\,|\, y_{1},x_{2},a)g(x_{2}\,|\, y_0, y_{1},x_{1},a)f_{\theta_0}(y_{1}\,|\, y_{0},x_{1},a)\mathrm{d}y_2\mathrm{d}x_2\mathrm{d}y_1.
\end{align}
This is the analog of the heterogeneity robust integral restriction \eqref{eq_phi_hetero_robust} when treating $g$ as known.  In the context of the MPH model with known homogeneous feedback, after the familiar change of variables $(y_{1},y_{2})\longrightarrow  (p_{1},p_{2}) \longrightarrow (\widetilde{p}_{1},\overline{p})$, \eqref{eq_phi_hetero_robust_knownhomogfeedback} specializes to
\begin{align*}
    0=\int_{0}^{+\infty}\underbrace{\left(\int_{\mathcal{X}}\int_{0}^1 \psi_{\theta_0}(y_0,\widetilde p_1,\overline p,x_2,x_1)g(x_{2}\,|\, y_0, \Lambda_{\alpha_0}^{-1}\left(
    \widetilde p_1\overline p
    e^{-x_1'\beta_0-\gamma_0y_0}
    \right),x_{1})\mathrm{d}x_2\mathrm{d}\widetilde p_1\right)}_{\equiv \Upsilon_{\theta}(y_{0},\overline{p},x_{1})}e^{2a}\overline{p}e^{-e^{a}\overline{p}}d\overline{p},
\end{align*}
where the mapping from $\phi_{\theta_0}$ to $\psi_{\theta_0}$ was detailed in the proof of Lemma \ref{lem_moment_charac_MPH}. Consequently, using ${\cal{L}}$ to denote the Laplace transform operator, we have $\mathcal{L}[\overline{p}\mapsto \overline{p} \Upsilon_{\theta}(y_{0},\overline{p},x_{1})][e^{a}]=0$
for any $a\in \mathbb{R}$. This implies that
\begin{align*}
    0&=\int_{\mathcal{X}}\int_{0}^1 \psi_{\theta_0}(y_0,\widetilde p_1,\overline p,x_2,x_1)g(x_{2}\,|\, y_0, \Lambda_{\alpha_0}^{-1}\left(
    \widetilde p_1\overline p
    e^{-x_1'\beta_0-\gamma_0y_0}
    \right),x_{1})\mathrm{d}x_2\mathrm{d}\widetilde p_1 \\
    &=\mathbb{E}\left[\psi_{\theta_0}(Y_0,\widetilde P_1,\overline P,X_2,X_1)|Y_0=y_0,\overline{P}=\overline p,X_1=x_1\right],
\end{align*}
where the second equality results from the fact that the joint density of $(\overline{P},X_{2},\widetilde{P}_1)$ conditional on $(Y_0,X_1,A)$ is 
\begin{align} \label{density_knownhomogfeedback_MPH}
         f(\overline{p},x_{2},\widetilde{p}_1|y_{0},x_{1},a)=\overline{p} e^{2a-\overline{p}e^{a}}\mathds{1}\{\overline{p}>0\}g(x_2|y_0,\Lambda_{\alpha}^{-1}(\widetilde{p}_1\overline{p}e^{-\gamma y_{0}-x_{1}'\beta}),x_1)\mathds{1}\left\{\widetilde{p}_1\in (0,1)\right\},
\end{align}
by Lemma \ref{lem: helmert_tranformation_MPH}. We conclude that elements of $\mathcal{T}_{\theta_0,\omega_0,K}^{\perp}$ are those $\phi_{\theta_0}$ such that we can write $\phi_{\theta_0}(Y^2,X^2)=\psi_{\theta_0}(Y_0,\widetilde P_1,\overline P,X_2,X_1)$
where $\mathbb{E}\left[\psi_{\theta_0}(Y_0,\widetilde P_1,\overline P,X_2,X_1)|Y_{0}, \overline P,X_1\right]=0$. Note that the difference relative to the unrestricted heterogeneous feedback case in Lemma \ref{lem_moment_charac_MPH} is that the function $\psi_{\theta_0}$ is allowed to depend on $X_2$ with $T=2$. Note also that the expectation depends on the feedback process $g$, which here is assumed known.

\paragraph{Efficient score with known homogeneous feedback.} With the previous characterization of $\mathcal{T}_{\theta_0,\omega_0,K}^{\perp}$ in hand, one can adjust arguments in Lemma \ref{proj_lemma_MPH} to show that for any $L\times 1$ mean-zero and square integrable random vector $\varphi_{\theta_0}(Y^{2},X^2)$, 
\begin{align*}
        \Pi(\varphi_{\theta_0}(Y^{2},X^2)|\mathcal{T}_{\theta_0,\omega_0,L}^{\perp})&=\varphi_{\theta_0}(Y^{2},X^2) 
        -\mathbb{E}\left[\varphi_{\theta_0}(Y^{2},X^2)|Y_0,\overline{P},X_1\right].
\end{align*}
Substituting $\varphi_{\theta_0}$ by the score $S^{\theta}(Y^2,X^2)$ under a Weibull baseline hazard delivers the efficient score $\phi_{\theta_0}^{\rm eff,KHFB}(Y^2,X^2)$ we seek for the known homogeneous feedback configuration (hence the superscript KHFB). Adapting the derivations in Appendix Section \ref{app: MPH_derivations} for the unrestricted feedback case, we find that the efficient score for $\gamma$ is unchanged
\begin{align} \label{eq: efficient_score_gamma_specific_T=2_KHF}
    \phi_{\theta_0}^{\rm eff,KHFB,\gamma}(Y^2,X^2) = \phi_{\theta_0}^{\rm eff,\gamma}(Y^2,X^2).
\end{align}
The efficient score for $\beta$ is different
\begin{align} \label{eq: efficient_score_beta_specific_T=2_KHF}
        \phi_{\theta_0}^{\rm eff,KHFB,\beta}(Y^2,X^2)&=-X_1\left(\widetilde{P}_1-\frac{1}{2}\right)\overline{P}\mathbb{E}\left[e^{A}\,|\, Y_0,\overline{P},X_1\right]\\
        &+X_2 - \mathbb{E}\left[X_{2}|Y_0,\overline{P},X_1\right]  \notag \\
        &-\left(X_2(1-\widetilde{P}_1)-\mathbb{E}\left[X_2(1-\widetilde{P}_1)\,|\, Y_0,\overline{P},X_1\right]\right)\overline{P}\mathbb{E}\left[e^{A}\,|\, Y_0,\overline{P},X_1\right] \notag, 
\end{align} 
and the efficient score for $\alpha$ retains an identical structure but differs due to its dependence on $\phi_{\theta_0}^{\rm eff,KHFB,\beta}$
\begin{align} \label{eq: efficient_score_alpha_specific_T=2_KHF}
    \phi_{\theta_0}^{\rm eff,KHFB,\alpha}(Y^2,X^2) =&  \frac{1}{\alpha}\left(  2+\ln \widetilde{P}_1+\ln (1-\widetilde{P}_1)\right) \notag \\
    &-\frac{1}{\alpha}\left(\widetilde{P}_1\ln \widetilde{P}_1+(1-\widetilde{P}_1)\ln (1-\widetilde{P}_1)+\frac{1}{2}\right)\overline{P}\mathbb{E}\left[e^{A}\,|\,Y_0,\overline{P},X_1\right] \notag \\
    &-\phi_{\theta_0}^{\rm eff,KHFB,\beta}(Y^2,X^2)'\frac{\beta_0}{\alpha_0}-\phi_{\theta_0}^{\rm eff,KHFB,\gamma}(Y^2,X^2)\frac{\gamma_0}{\alpha_0}.
\end{align}
The key identity to obtain these expressions is $\mathbb{E}\left[e^{A}\,|\,Y_0,\widetilde{P}_1,\overline{P},X_1,X_2 \right]=\mathbb{E}\left[e^{A}\,|\,Y_0,\overline{P},X_1 \right]$, which follows from the form of the joint density in \eqref{density_knownhomogfeedback_MPH}.

\paragraph{Numerical experiment.} We rely on the simulation design of subsection \ref{subsec: numerical_experiment} and now set the success probability for $X_2$ to $p(Y_0,Y_1,X_1,V)=1-\exp(-\tau_{C}(Y_{0},X_{1},Y_{1}))$, with $\tau_{C}(Y_{0},X_{1},Y_{1})=Y_{0}+X_{1}+Y_{1}$, to produce a homogeneous feedback process. In view of \eqref{eq: efficient_score_beta_specific_T=2}–\eqref{eq: efficient_score_gamma_specific_T=2}–\eqref{eq: efficient_score_alpha_specific_T=2} and  \eqref{eq: efficient_score_gamma_specific_T=2_KHF}-\eqref{eq: efficient_score_beta_specific_T=2_KHF}-\eqref{eq: efficient_score_alpha_specific_T=2_KHF}, computing the efficient scores under unrestricted feedback and known homogeneous feedback requires evaluating the following six conditional expectations:
$\mathbb{E}\left[V\,|\, Y_0,\overline{P},X_1\right]$, $ \mathbb{E}\left[X_{2}|Y_0,\widetilde{P}_1,\overline{P},X_1\right]$, $\mathbb{E}\left[X_{2}|Y_0,\overline{P},X_1\right]$, $ \mathbb{E}\left[X_{2}V|Y_0,\widetilde{P}_1,\overline{P},X_1\right]$, $\mathbb{E}\left[X_{2}(1-\widetilde{P}_1)|Y_0,\overline{P},X_1\right]$, $\mathbb{E}\left[X_{2}(1-\widetilde{P}_1)V|Y_0,\overline{P},X_1\right]$, with $V= \exp(A)$. Within this setup, we obtain the following expressions
\begin{align*}
    \mathbb{E}\left[V\,|\, Y_0,\overline{P},X_1\right]&=\frac{T+\kappa_0}{\overline{P}+\lambda_0} \\
    \mathbb E\left[X_2\,|\, Y_0,\widetilde P_1,\overline P,X_1\right]&=1-e^{-Y_0-X_1-\left(\widetilde P_1\overline Pe^{-X_1'\beta_0-\gamma_0 Y_0}\right)^{1/\alpha_0}} \\
    \mathbb E[X_2\,|\, Y_0,\overline P,X_1]&=1-e^{-Y_0-X_1}{}_1F_1\left(\alpha_0;\alpha_0+1;-C\right) \\
\mathbb E\left[X_2V\,|\,Y_0,\widetilde P_1,\overline P,X_1 \right]&=\frac{T+\kappa_0}{\overline P+\lambda_0}\left(
1-e^{-Y_0-X_1-\left(\widetilde P_1\overline Pe^{-X_1'\beta_0-\gamma_0 Y_0}\right)^{1/\alpha_0}}
\right) \\
    \mathbb{E}\left[X_2(1-\widetilde{P}_1)\,|\, Y_0,\overline{P},X_1\right]&=
\frac{1}{2}
-e^{-Y_0-X_1}
\left(
{}_1F_1(\alpha_0;\alpha_0+1;-C)
-\frac{1}{2}{}_1F_1(2\alpha_0;2\alpha_0+1;-C)
\right) \\
\mathbb E\left[
X_2(1-\widetilde P_1)V\,|\, Y_0,\overline P,X_1\right]&=
\frac{T+\kappa_0}{\overline P+\lambda_0}\left[\frac{1}{2}-e^{-Y_0-X_1}\left({}_1F_1(\alpha_0;\alpha_0+1;-C)-\frac{1}{2}{}_1F_1(2\alpha_0;2\alpha_0+1;-C)\right)\right],
\end{align*}
where $C=\left(\overline P e^{-\gamma_0Y_0-X_1'\beta_0}\right)^{1/\alpha_0}$, and where ${}_1F_1(a;b;z)$ denotes the confluent hypergeometric function of the first kind
\begin{align*}
{}_1F_1(a;b;z)
=
\frac{\Gamma(b)}
{\Gamma(a)\Gamma(b-a)}
\int_0^1
e^{zt}t^{a-1}(1-t)^{b-a-1}\,\mathrm{d}t.
\end{align*}

Table \ref{tab1_montecarlo_KHHF} reports asymptotic standard errors for estimates of $\theta_0$ based on the efficient score under unrestricted feedback, the locally efficient score, and the simple moment vector \eqref{eq_simple_moment}. The semiparametric efficiency bound under known homogeneous feedback serves as the benchmark.  \\
\begin{table}[!htb]
\caption{Asymptotic standard errors relative to the semiparametric efficiency bound with known homogeneous feedback} \label{tab1_montecarlo_KHHF}
\centerline{
\scalebox{1}{
\begin{tabular}{lccccccccccccc} \toprule  \toprule 
  & $\alpha$ & $\beta$ & $\gamma$  \\ 
Eff.score KHFB &1.0&1.0&1.0\\ 
Eff.score FB &1.000&1.186&1.007\\ 
Locally Eff.score FB &1.033&1.431&1.342\\ 
Simple moments &3.552&4.246&10.621\\ 
\bottomrule \bottomrule 
\end{tabular}
}}
\par
{\footnotesize\textit{\textsc{Notes:} KHFB abbreviates known homogeneous feedback, FB refers to unrestricted feedback.}}
\end{table}
Comparing the first two rows, we note that the penalty incurred by leaving the feedback process unspecified is relatively mild, with a 19\% increase in the standard error for $\beta$ and virtually no change in estimation precision for $\alpha$ and $\gamma$. Inspecting the third and fourth rows yields two additional conclusions. First, replacing the efficient score under unrestricted feedback with the locally efficient score results in a noticeable, though not large, loss of precision. Second, it is clear that the estimator based on the working models once again significantly outperforms the estimator based on the simple moment vector. These results are broadly consistent with our findings under strict exogeneity reported in Subsection \ref{subsec: numerical_experiment}.

\end{document}